\documentclass[noinfoline]{imsart}

\RequirePackage[OT1]{fontenc}
\RequirePackage{amsthm,amsmath}
\RequirePackage{natbib}

\usepackage[hyperindex,breaklinks]{hyperref}

\usepackage{verbatim}
\usepackage{comment}

\allowdisplaybreaks

\usepackage{fullpage}

\newcommand{\E}{\mathbb{E}}
\def\Ez#1{\mathbb{E} \left[ #1 \right]}

\newcommand {\cov}{\textrm{Cov}}
\newcommand {\var}{\textrm{Var}}

\newcommand {\QBartlett}{L}
\newcommand{\Ltwo}{L^2([0,1])}
\DeclareMathOperator{\I}{i}
\DeclareMathOperator{\cum}{cum}
\DeclareMathOperator{\sign}{sign}
\DeclareMathOperator*{\argmin}{arg\,min}
\DeclareMathOperator*{\argmax}{arg\,max}

\DeclareMathOperator{\tr}{tr}
\DeclareMathOperator{\Prob}{\mathbb{P}}
\newcommand{\D}[1]{\ensuremath{\operatorname{d}\!{#1}}}

\newcommand{\FARno}{\textbf{(FAR(1))}}
\newcommand{\FMAno}{\textbf{(FMA(4))}}

\newcommand{\Ba}{\textbf{(reg1)}}
\newcommand{\Bb}{\textbf{(reg2)}}
\newcommand{\Bc}{\textbf{(reg3)}}

\newcommand{\mathbbSigma}{\mathbb{S}}
\newcommand{\Komega}[1]{K_\omega^{(#1)}}

\newcommand\Item[1][]{%
  \ifx\relax#1\relax  \item \else \item[#1] \fi
  \abovedisplayskip=0pt\abovedisplayshortskip=0pt~\vspace*{-\baselineskip}}

\usepackage{subfigure}
\usepackage{layout}

\usepackage{dsfont}
\usepackage[mathscr]{eucal}
\usepackage[toc,page]{appendix}
\usepackage{mathrsfs}
\usepackage{color}
\usepackage{pifont}
\usepackage{bm}
\usepackage{latexsym}
\usepackage{amsfonts}
\usepackage{amssymb}
\usepackage{epsfig}
\usepackage{graphicx}
\usepackage{multirow}
\usepackage{enumitem}
\usepackage{mathbbol}
\usepackage[dvipsnames]{xcolor}
\usepackage{placeins}
\usepackage{stackrel}

\newtheorem{theorem}{Theorem}
\newtheorem{proposition}{Proposition}

\newtheorem{lemma}{Lemma}

\usepackage{array}
\newcolumntype{P}[1]{>{\centering\arraybackslash}p{#1}}

\startlocaldefs
\numberwithin{equation}{section}
\theoremstyle{plain}
\endlocaldefs

\begin{document}

\begin{frontmatter}

\title{{\large Functional Lagged Regression with Sparse Noisy Observations}}


\begin{aug}
\author{\fnms{Tom{\'a}{\v s}} \snm{Rub{\'i}n}\ead[label=e1]{tomas.rubin@epfl.ch}} \and
\author{\fnms{Victor M.} \snm{Panaretos}\ead[label=e2]{victor.panaretos@epfl.ch}}


\runauthor{T. Rub{\'i}n \& V.M. Panaretos}

\affiliation{Ecole Polytechnique F\'ed\'erale de Lausanne}

\address{Institut de Math\'ematiques\\
Ecole Polytechnique F\'ed\'erale de Lausanne\\
\printead{e1}, \printead*{e2}}

\end{aug}

\begin{abstract}
A functional (lagged) time series regression model involves the regression of scalar response time series on a time series of regressors that consists of a sequence of random functions. In practice, the underlying regressor curve time series are not always directly accessible, but are latent processes observed (sampled) only at discrete measurement locations. In this paper, we consider the so-called sparse observation scenario where only a relatively small number of measurement locations have been observed, possibly different for each curve. The measurements can be further contaminated by additive measurement error. A spectral approach to the estimation of the model dynamics is considered. The spectral density of the regressor time series and the cross-spectral density between the regressors and response time series are estimated by kernel smoothing methods from the sparse observations. The impulse response regression coefficients of the lagged regression model are then estimated by means of ridge regression (Tikhonov regularisation) or PCA regression (spectral truncation). The latent functional time series are then recovered by means of prediction, conditioning on all the observed observed data. The performance and implementation of our methods are illustrated by means of a simulation study and the analysis of meteorological data.
\end{abstract}


\end{frontmatter}

\tableofcontents

\newpage
\section{Introduction} \label{intro}

A (lagged) time series regression model is perhaps the most basic --and certainly one of the most and longest studied \citep{kolmogoroff1941interpolation,wiener1950extrapolation}-- forms of coupled analysis of two time series. In a general context, given two discrete-time stationary time series $\{X_t\}_{t\in\mathbb{Z}}$ and $\{Z_t\}_{t\in\mathbb{Z}}$, the input (or regressor) and output (or response), valued in some vector spaces $\mathcal{H}_1$ and $\mathcal{H}_2$, such a model postulates that
$$Z_t = a + \sum_{k\in\mathbb{Z}} \mathcal{B}_k X_{t-k} + e_t,\quad t\in\mathbb{Z}, $$
for some constant $a\in\mathcal{H}_2$, a sequence of random disturbances $\{e_t\}_{t\in\mathbb{Z}}$ valued in $\mathcal{H}_2$ and a sequence of linear mappings $ \mathcal{B}_k:\mathcal{H}_1\rightarrow \mathcal{H}_2,\,k\in\mathbb{Z}$. This linear coupling would be the typical dependence model, for instance, if $\{(X_t,Z_t)\}_{t\in\mathbb{Z}}$ were a jointly Gaussian stationary process in $\mathcal{H}_1\times \mathcal{H}_2$, and is also known as a time-invariant linearly filtered time series model. The estimation problem is then to estimate the unknown transformations $\{\mathcal{B}_k\}_{k\in\mathbb{Z}}$ given the realisation of a finite stretch of the joint series $\{(Z_t,X_t)\}_{t\in\mathbb{Z}}$ (a problem also known as system identification, particularly in signal processing).

This problem is very well understood and has been extensively studied in the classical context where the spaces $\mathcal{H}_j, j=1,2,$ coincides with the Euclidean spaces of potentially different dimensions \citep{brillinger1981time,priestley1981spectral,shumway2000time}. Nevertheless, generalising these results to the case where either space, and particularly the regressor space $\mathcal{H}_1$, may be an infinite dimensional vector space is far from straightforward, and has been comparatively less studied. The difficulty in this case is that one needs to manipulate operations involving the inverses of compact or even trace-class operators, which fail to exist boundedly on the entire codomain. Consequently, analysis of such models requires drawing on tools from functional analysis, and developing novel methodology that incorporates suitable regularising schemes. Such a setting has only recently been considered for functional time series regression models, for example by \citet{hormann2015estimation}, who treat the problem of estimation of the filter coefficients by means of spectral truncation regularisation (PCA regression), and \citet{pham2018methodology}, who deduce convergence rates for the estimated coefficients when using Tikhonov regularisation (ridge regression).

While this may seem to be an artificial abstraction at first sight, such infinite-dimensional time series are becoming increasingly prominent for applications. Indeed the abstraction of an infinite dimensional (Hilbert) space, captures the scenario where the value of a series at each time is a square-integrable function, e.g. a curve. Examples include time series of DNA minicircles evolving in solution (seen as a time series of closed curves in 3D indexed by discrete time \citep{tavakoli2016detecting}) or the data constructed by dividing a continuously observed scalar time series into segments of an obvious periodicity, usually days. Examples of the latter form are particularly prominent in environmental applications, for example the analysis of particulate matter atmospheric pollution
\citep{hormann2010weakly,hormann2015dynamic,hormann2018testing,aue2015prediction}, traffic data modelling \citep{klepsch2017prediction}, or financial applications of intra-day trading \citep{muller2011functional,kokoszka2017dynamic}. Another promising financial application of functional time series emerges in the yield curve modelling \citep{hays2012functional,kowal2017functional,sen2019timeseries}. 

We focus here primarily on the case $\mathcal{H}_2=\mathbb{R}$, i.e. when the response process is scalar. This is indeed the case that is most often studied in the literature, due in large part due to the many examples it covers, but also because it is the simplest version of the problem that captures the essence of higher generality: the difficulty in estimating the filter coefficients  lies in the ill-possessedness of the spectral density operator inversion that requires regularisation. For a scenario involving $\mathcal{H}_2$ being infinite dimensional, i.e. with a functional response, see \citet{hormann2015estimation} or Appendix~\ref{sec:appendix_functional_response}.

In practice, one can never observe the regressor time series $\{X_t\}_{t\in\mathbb{Z}}$ in its ``Platonic" continuum form. For example, if $\mathcal{H}_1$ is a space of functions on $[0,1]$ and each $X_t:[0,1]\rightarrow\mathbb{R}$ is a curve, then one might be able to observe evaluations of $X_t(\cdot)$ at various locations of its domain. In some cases, the sampled locations are sufficiently dense, and the measurement instrumentation sufficiently precise to be free of any additional noise contamination, so that one can disregard the effects of sampling. This is essentially the approach taken in \citet{hormann2015estimation} and \citet{pham2018methodology} where the regressors are treated as being fully observed as elements of $\mathcal{H}_1=L^2[0,1]$. However, it may well happen that $\{X_t\}_{t\in\mathbb{Z}}$ is only measured at few and randomly positioned locations in its domain, indeed varying with time, and that the measurements are themselves contaminated by noise. That is, instead of observing $X_t(u)$ for all $u$ in $[0,1]$, we instead observe
$$Y_{tj} = X_t(x_{tj}) + \epsilon_{tj}, \qquad j=1,\dots,N_t,\quad t=1,...,T.$$
for a sequence of point processes $(x_{t1},...,x_{tN_t}), t=1,\dots,T$, independent in time, and a white noise measurement error sequence. This observation scheme is illustrated on Figure \ref{fig:sampling} and is exhibited, for example, when dealing with fair weather atmospheric electricity data \citep{rubin2020sparsely,tammet2009joint}.

\begin{figure}[h]
\centering
\includegraphics[width=\textwidth]{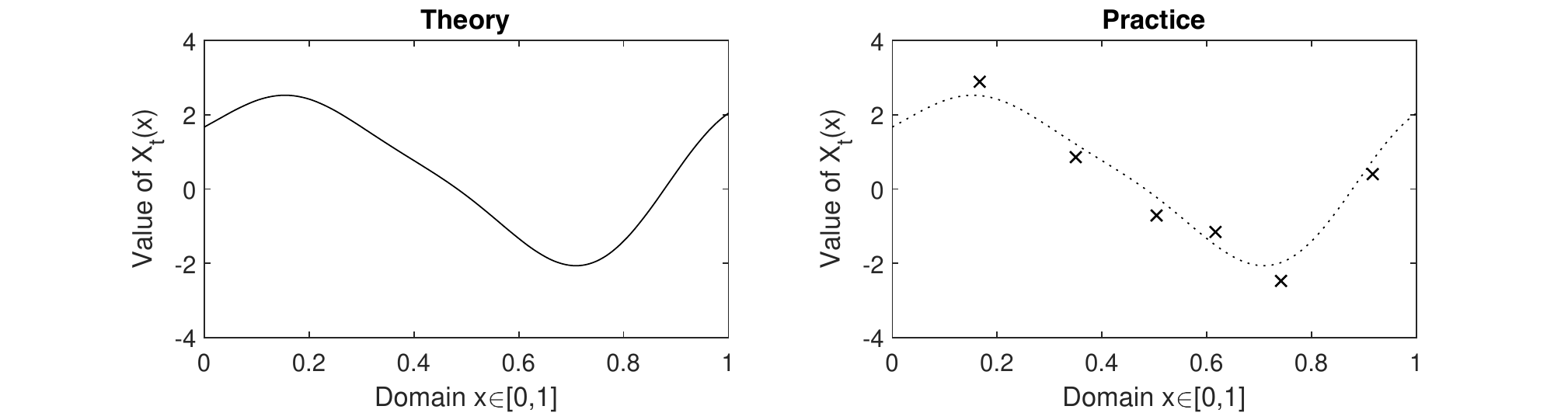}
\caption{
\textbf{Left:} The ``Theory" picture, where $X_t(x), x\in[0,1]$, is a fully observed functional datum (solid line).
\textbf{Right:} The ``Practice" picture,  where the functional datum $X_t(x), x\in[0,1],$ is latent, and one can either observe dense noiseless observations (dotted line) or sparse noisy observations (crosses) with additive noise $Y_{tj} = X_t(x_{tj}) + \epsilon_{tj}$ observed at locations $x_{tj},j=1,\dots,6$. In the dense case, one can typically behave as if the true latent function were observed. The sparse case, however, needs new tools.
}
\label{fig:sampling}
\end{figure}

Such a setting escapes both the methods and the theory developed at the level of continuum, and instead requires methodology that accounts for the observation scheme, as well as theory that incorporates the latent/emission processes explicitly. The purpose of this paper is precisely to construct such a methodology and develop the associated asymptotic theory.

Our approach consists in estimating the complete space-time covariance structure of the data by estimating the spectral density operator of the regressor time series, and the cross-spectral density operator between the regressor and response time series. 
Our estimators are based on the kernel smoothing methods as discussed for example in \citep{yao2005functional,yao2005functional_linear_regression,hall2006properties,li2010uniform} and extended to the time series case in \cite{rubin2020sparsely}. Once these are estimated, one obtains estimating equations whose solution yields the estimators of the filter coefficients. The solution of the estimating equations, however, comprises an ill-posed inverse problem, hence regularisation is required. We offer two regularisation techniques, spectral truncation regularisation \citep{hormann2015estimation} and Tikhonov regularisation \citep{pham2018methodology}. The forecasting of the response process is then implemented by first predicting the latent functional regressor data (using their estimated spectral characteristics) and then plugging-in these predictions into the estimated lagged regression model.

Sparsely observed functional time series have only recently received attention \citep{kowal2017functional,kowal2017bayesian,rubin2020sparsely,sen2019timeseries}, and our results appear to be the first in the context of the lagged regression model where the regressor process is functional. A related problem of dynamic function-on-scalar regression was studied by \citet{kowal2018dynamic} by means of Bayesian factor models.

Our presentation is organised as follows. Section \ref{sec:model} defines the framework of sparsely observed functional time series as well as the functional lagged regression model and its analysis in the spectral domain. The section explains the estimation methodology for the model components, the spectral and the cross-spectral densities, and the regularised estimator of the filter coefficients of the lagged regression. Furthermore, the forecasting algorithm is introduced. Section~\ref{sec:asymptotics} presents the asymptotic results of the proposed method: consistency and convergence rates are established. Section~\ref{sec:numerical experiments} probes the finite sample properties of the proposed methodology on a simulation study.
As a proof-of-concept illustration, we included an example featuring meteorological data in Section~\ref{sec:data analysis} of the paper.
The properties of the two considered regularisation methods are summarised in Section~\ref{sec:discussion}.
We present the extension of the proposed methodology for lagged regression models with a functional response in Appendix~\ref{sec:appendix_functional_response} and for lagged regression models with multiple regressors time series in Appendix~\ref{sec:appendix_muptiple_inputs}.
The proofs of the formal statements presented in Appendix~\ref{sec:appendix_proofs}.

\section{Model and estimation methodology}
\label{sec:model}

\subsection{Functional time series regression model}

The classical theory of functional data analysis views its data, random curves, as random elements in a suitable Hilbert space. Throughout our presentation we assume this space to be the Hilbert space of real square integrable function, $\mathcal{H}_1 \equiv \mathcal{H} = \Ltwo$, equipped with the inner product $\langle\cdot,\cdot\rangle$ and the norm $\|\cdot\|$. The functional time series is then a sequence of random elements in such space indexed by the integer variable $t\in\mathbb{Z}$, interpreted as times (for example days), and denoted as $\{ X_t\} \equiv \{ X_t(\cdot) \}_{t\in\mathbb{Z}}$. The argument of the individual functions, for the most part of this paper denoted as $x\in[0,1]$, can be interpreted as either the spatial location or intra-day time (for the sequence of intra-day data order day-by-day into a sequence). The basic assumption we impose is:

\begin{enumerate}[label=(A{\arabic*})]
\item\label{assumption:A.1} The functional time series $\{X_t\}$ is assumed to be square-integrable, i.e. $\E \|X_t\|^2 < \infty,\, t\in\mathbb{Z}$, and second-order stationary in the variable $t$. Furthermore, the sample paths (trajectories) of $\{X_t\}$ are smooth. 
\end{enumerate}

The functional time series $\{X_t\}$ is assumed to be only latent and therefore not accessible directly.
Denote $N_t$ the number of measurements available for the random curve at time $t=1,\dots,T$ where $T\in\mathbb{N}$ is the time of the latest curve with at least one measurement.
We denote $Y_{tj}$ the $j$-th measurement on the $t$-th curve at spatial position $x_{tj}\in[0,1]$ where $j=1,\dots,N_t$ and $t=1,\dots,T$. The measurements are assumed to be contaminated by an additive measurement errors $\epsilon_{tj},\, j=1,\dots,N_t, \, t=1,\dots,T$ which constitute an independent identically distributed zero-mean array of scalar random variables with variance $\sigma^2>0$. The number of measurement locations $\{N_t\}$, the measurement errors $\{\epsilon_{tj}\}$, and the underlying functional time series $\{X_t\}$ are assumed to be independent.
The above outlined sparse observation regime leads to the model equation
\begin{equation}\label{eq:observation_scheme}
Y_{tj} = X_t(x_{tj}) + \epsilon_{tj}, \qquad j=1,\dots,N_t,\quad t=1,...,T.
\end{equation}

The assumption \ref{assumption:A.1} allows us to define the mean function as $\mu(x) = \E X_0(x),\, x\in[0,1],$ and the lag-$h$ autocovariance kernel
$R_h^X(x,y) = \Ez{ \left(X_h(x)-\mu(x)\right)\left(X_0(y)-\mu(y)\right) }, \,x,y\in[0,1],$ for $h\in\mathbb{Z}$. The integral operator $\mathscr{R}_h^X$ associated with this kernel, called the lag-$h$ autocovariance operator, is defined by the right integration
$\left(\mathscr{R}_h^X f\right)(x) = \int_0^1 R_h^X(x,y)f(y)\D y, \, f\in \Ltwo $.
The autocovariance kernels and operators are assumed to be summable in the supremum norm (denoted $\|\cdot\|_\infty$) and the nuclear norm (denoted $\|\cdot\|_1$) respectively,
\begin{enumerate}[label=(A{\arabic*}),resume]
\Item\label{assumption:A.2}
$$
\sum_{h\in\mathbb{Z}} \|R^X_h\|_\infty =  \sum_{h\in\mathbb{Z}} \sup_{x,y\in[0,1]} |R^X_h(x,y)| < \infty ,
\qquad
\sum_{h\in\mathbb{Z}} \|\mathscr{R}^X_h\|_1 = \sum_{h\in\mathbb{Z}}  \mbox{trace}\left\{\sqrt{(\mathscr{R}^X_h)^*\mathscr{R}^X_h}\right\}< \infty.
$$
\end{enumerate}

Throughout the paper we assume that the response time series $\{Z_t\}$ is scalar. We make this assumption to facilitate the presentation while keeping in mind the fact that the scalar case already demonstrates the difficulty of the ill-posed estimation problem and the requirement of regularisation. The extension to a functional response time series is presented in Appendix~\ref{sec:appendix_muptiple_inputs}.

The functional lagged regression model with a scalar response time series $\{Z_t\}$ is given by the equation
$$
Z_t = a + \sum_{k\in\mathbb{Z}} \mathcal{B}_k X_{t-k} + e_t,\qquad t\in\mathbb{Z},
$$
where $a\in\mathbb{R}$ is a constant, called the intercept; $\mathcal{B}_k : \mathcal{H}\to\mathbb{R},\,k\in\mathbb{Z},$
are fixed continuous linear functionals called the filter coefficients, with Riesz-representers $\{b_k\}_{k\in\mathbb{Z}}$ (i.e. $\mathcal{B}_kf=\langle f,b_k\rangle$ for $f\in\mathcal{H}$); and $\{e_t\}_{t\in\mathbb{Z}}$ is a sequence of independent identically distributed zero-mean real random variables with variance $\tau^2>0$.
The time series $\{X_t\}$, the sampling counts $\{N_t\}$, the observation locations $\{x_{tj}\}$, and the errors $\{\epsilon_{tj}\}$ and $\{e_t\}$ are assumed to be independent.
The filter coefficients are assumed to be summable in the norm
\begin{enumerate}[label=(A{\arabic*}),resume]
\Item\label{assumption:A.3}
$$
\sum_{k\in\mathbb{Z}}\| \mathcal{B}_k \| < \infty.
$$
\end{enumerate}
where given any continuous linear functional $\mathcal{G}: \mathcal{H}\to\mathbb{R}$ the notation  $\| \mathcal{G} \|$ will be used interchangeably with $\| g \|$ for $g$ the Riesz-representer of $\mathcal{G}$, i.e. $\mathcal{G}f=\langle g,f\rangle$.

We make here a few simplifying assumptions concerning the first order structure of the data, to keep focus on the more challenging second order structure. Firstly, we assume that the regressor time series $\{X_t\}$ is centred, i.e. $\mu = 0$. If that was not the case, the mean function $\mu(\cdot)$ could be estimated by the local-linear kernel smoother \citep{yao2005functional,rubin2020sparsely} and all the sparse observations centred. Secondly, we assume that the intercept is null, i.e. $a = 0$. Otherwise it could be estimated using the relation $a = \Ez{Z_0} - \sum_{k\in\mathbb{Z}}\mathcal{B}_k \mu$ where $\Ez{Z_0}$ is estimated by the sample mean and the coefficients $\mathcal{B}_k$ using the methods of this paper. As a consequence of these assumptions, the response time series is also centred, i.e. $\Ez{Z_0}=0$, leading to the model
\begin{equation}
\label{eq:regression_model}
Z_t = \sum_{k\in\mathbb{Z}} \mathcal{B}_k X_{t-k} + e_t, \qquad t\in\mathbb{Z},
\quad\qquad \left( \E X_t =0,\, \E Z_t = 0 \right).
\end{equation}

Finally, we define the lag-$h$ cross-covariance kernel and the lag-$h$ cross-covariance operator between the scalar time series $\{Z_t\}$ and the functional time series $\{X_t\}$, for $h\in\mathbb{Z}$, as
$$ R^{ZX}_h(x) = \Ez{ Z_h X_0(x) }, \quad x\in[0,1], \qquad
\mathscr{R}^{ZX}_h = \Ez{ Z_h \left\langle X_0, \cdot \right\rangle } :\mathcal{H}\to\mathbb{R}.$$

\subsection{Spectral analysis of functional lagged regression}

Under the assumptions \ref{assumption:A.1} and \ref{assumption:A.2}, the sums
\begin{equation}
\label{eq:spectral_density_kernel_operator}
f^X_\omega(\cdot,\cdot) = \frac{1}{2\pi} \sum_{h\in\mathbb{Z}} R^X_h(\cdot,\cdot)e^{-\I \omega h},
\qquad
\mathscr{F}^X_\omega = \frac{1}{2\pi} \sum_{h\in\mathbb{Z}} \mathscr{R}^X_h e^{-\I \omega h},
\end{equation}
converge in the supremum and the trace-class norm respectively,
define the spectral density kernel and the spectral density operator at frequency $\omega\in[-\pi,\pi]$ respectively \citep{panaretos2013fourier}.
The spectral density operator $\mathscr{F}^X_\omega$ is a non-negative self-adjoint trace-class operator for all $\omega\in[-\pi,\pi]$. Hence it admits the spectral representation
\begin{equation}
\label{eq:spectral_density_operator_eigendecomposition}
\mathscr{F}^X_\omega = \sum_{m=1}^\infty \lambda_m^\omega \varphi_m^\omega \otimes \varphi_m^\omega
= \sum_{m=1}^\infty \lambda_m^\omega \left\langle \varphi_m^\omega, \cdot\right\rangle \varphi_m^\omega
\end{equation}
where $\otimes$ is the tensor product in $\mathcal{H}$, defined by the inner product formula on the right-hand side of \eqref{eq:spectral_density_operator_eigendecomposition}. The elements of the sequence $\lambda_1^\omega \geq \lambda_2^\omega \geq \dots \geq 0$ are called the harmonic eigenvalues and the corresponding functions $\{\varphi_m^\omega\}_{m=1}^\infty$ are called the harmonic eigenfunctions.

Furthermore, the lagged autocovariance kernels and operators can be recovered by the inversion formula \citep{panaretos2013fourier} that holds in the supremum and the nuclear norm, respectively:
\begin{equation}
\label{eq:inversion_formula_operators_kernels}
R^X_h(\cdot,\cdot) = \int_{-\pi}^\pi f^X_\omega(\cdot,\cdot) e^{\I \omega h} \D\omega,
\qquad
\mathscr{R}^X_h = \int_{-\pi}^\pi \mathscr{F}^X_\omega e^{\I \omega h} \D\omega,
\qquad h\in\mathbb{Z}.
\end{equation}

\citet{hormann2015estimation} defined the cross-spectral density operator between $\{Z_t\}_{t\in\mathbb{Z}}$ and $\{X_t\}_{t\in\mathbb{Z}}$ by the formula
\begin{equation}
\label{eq:crossspectral_density_kernel_operator}
\mathscr{F}^{ZX}_\omega = \frac{1}{2\pi} \sum_{h\in\mathbb{Z}} \mathscr{R}_h^{ZX} e^{-\I \omega h},
\qquad\omega\in[-\pi,\pi],
\end{equation}
and showed that this sum converges in the Hilbert-Schmidt norm, provided the autocovariance function $\mathscr{R}_h^X$ is summable in the Hilbert-Schmidt norm and the sequence of filter coefficient operators are also summable in the operator norm.
Since our assumptions \ref{assumption:A.1} --- \ref{assumption:A.3} are stronger and we consider only scalar response, \eqref{eq:crossspectral_density_kernel_operator} is well defined in the setting of our paper.

Likewise, one defines the cross-spectral density kernel between $\{Z_t\}_{t\in\mathbb{Z}}$ and $\{X_t\}_{t\in\mathbb{Z}}$ by
\begin{equation}
\label{eq:crossspectral_density_kernel}
f^{ZX}_\omega(\cdot) = \frac{1}{2\pi} \sum_{k\in\mathbb{Z}} R^{ZX}_h(\cdot) e^{-\I \omega h},
\qquad\omega\in[-\pi,\pi].
\end{equation}
The sum on the right-hand side of \eqref{eq:crossspectral_density_kernel} converges in the supremum norm.

Further, \cite{hormann2015estimation} introduced the frequency response operator
$$
\mathscr{B}_\omega = \sum_{h\in\mathbb{Z}} \mathcal{B}_k e^{-\I h\omega},
\qquad \omega\in [-\pi,\pi].
$$
and obtained the relation between the spectral density operators, cross-spectral density operators, and the frequency response operators
\begin{equation}
\label{eq:spectral_normal_equation}
\mathscr{F}^{ZX}_\omega = \mathscr{B}_\omega \mathscr{F}^X_\omega, \qquad\omega\in[-\pi,\pi],
\end{equation}
which provides the basis for the estimation of the filter coefficients introduced in the next section. In our case of scalar response, $\mathscr{B}_\omega$ is in fact a functional, i.e. $\mathscr{B}_\omega:\mathcal{H}\to\mathbb{R}$.
The filter coefficients $\{\mathcal{B}_k\}_{k\in\mathbb{Z}}$ can be recovered by the formula
\begin{equation}
\label{eq:inversion_formula_cross_operators_kernels}
\mathcal{B}_k = \frac{1}{2\pi} \int_{-\pi}^\pi \mathscr{B}_\omega e^{\I h\omega} \D\omega,
\qquad k\in\mathbb{Z}.
\end{equation}

\subsection{Nonparametric estimation in the lagged regression}
\label{subsec:nonparam_estimation}

The nonparametric estimation in functional models given sparse noisy measurements has been frequently based on local-polynomial kernel smoothers \citep{yao2005functional,hall2006properties,li2010uniform,rubin2020sparsely}. 
Throughout the paper we use the Epanechnikov kernel $K(v) = \frac{3}{4}(1-v^2)$ for $v\in[-1,1]$, and $0$ otherwise, because it is known to minimise the mean square error in the classical local-polynomial smoothing problem \citep[Section 3.2.6]{FanJianqing1996Lpma}.
Moreover, the Epanechnikov kernel's bounded support speeds up the numerical calculations.
Nevertheless, the choice of the kernel function is not crucial for the statistical performance of the smoothing estimators \citep[Section 3.2.6]{FanJianqing1996Lpma} and any other commonly used kernel function could be adopted.

We start by defining the ``raw'' covariances of
$\{X_t\}$ as
$ G_{h,t}^X\left( x_{t+h,j}, x_{tk} \right) = Y_{t+h,j}Y_{tk}$ for
$h=-T,\dots,T,\, j=1,\dots,N_{t+h},\,k=1,\dots,N_t$.
The lag-$0$ covariance kernel $R_0^X(\cdot,\cdot)$ is estimated by the local-linear surface smoother, setting
$\hat{R}^X_0(x,y) = \hat{c}_0^{(1)}$ where
\begin{multline*}
\left(\hat c_0^{(1)},\, \hat c_1^{(1)},\, \hat c_2^{(1)} \right) = \argmin_{c_0^{(1)},\,c_1^{(1)},\,c_2^{(1)}}
\sum_{t=1}^T
\stackrel[j\neq k]{}{\sum_{j=1}^{N_t} \sum_{k=1}^{N_t}}
K\left( \frac{x_{tj}-x }{B_R} \right)
K\left( \frac{x_{tk}-y }{B_R} \right)
\times\\\times
\left\{ G_{0,t}^X( x_{tj}, x_{tk} ) - c_0^{(1)} - c_1^{(1)}(x-x_{tj}) - c_2^{(1)}(y-x_{tk}) \right\}^2
\end{multline*}
for $x,y\in[0,1]$ and where $B_R>0$ is the bandwidth parameter for the surface smoother.

In the next step we aim to estimate the measurement error variance $\sigma^2$ by the approach suggested by \citep{yao2003shrinkage,yao2005functional} for which we need the two following ingredients: the estimator of the diagonal of the lag-$0$ covariance kernel of $\{X_t\}$ with and without the measurement noise contamination.
Firstly, we estimate the diagonal of $R^X_0$ without the measurement noise contamination by the local-quadratic smoother along the direction perpendicular to the diagonal. For $x\in[0,1]$ we set $\bar{R}_0^X(x) = \bar{c}_0$ where
\begin{multline*}
\left(c_0^{(2)},\,c_1^{(2)},\,c_2^{(2)} \right) =
\argmin_{c^{(2)}_0,\,c^{(2)}_1,\,c^{(2)}_2}
\stackrel[j\neq k]{}{\sum_{j=1}^{N_t} \sum_{k=1}^{N_t}}
K\left( \frac{x_{tj}-x }{B_R} \right)
K\left( \frac{x_{tk}-x }{B_R} \right)
\times\\\times
\left\{ G_{0,t}^X( x_{tj}, x_{tk} )
	- c_0^{(2)}
	- c_1^{(2)} \Delta(x_{tj},x_{tk})
	- c_2^{(2)} \Delta(x_{tj},x_{tk})^2
\right\}^2
\end{multline*}
where $\Delta(x_{tj},x_{tk})$ is the distance of the point $(x_{tj},x_{tk})$ from the diagonal equipped with the positive sign if the point $(x_{tj},x_{tk})$ is above the diagonal, and with the negative sign if below. Formally
$$ \Delta(x_{tj},x_{tk}) = \sign\left( x_{tk}-x_{tj}\right) \sqrt{ \left( P(x_{tj},x_{tk}) - x_{tj} \right)^2 + \left(P(x_{tj},x_{tk}) - x_{tk} \right)^2 }$$
where $\sign(\cdot)\in\{-1,0,1\}$ is the sign function and
$P(x_{tj},x_{tk})$ is the first coordinate of the point $(x_{tj},x_{tk})$ projected onto the diagonal of $[0,1]^2$.

Secondly, we estimate the function $x \mapsto R^X_0(x,x) + \sigma^2,\, x\in[0,1]$, i.e. the noise contaminated diagonal of the lag-$0$ covariance operator.
For $x\in[0,1]$ and a bandwidth parameter $B_V>0$, we use the local-linear line smoother
and set $\hat{V}^X(x) = c^{(3)}_0$ where 
$$
\left( c^{(3)}_0,\, c^{(3)}_1 \right)
= \argmin_{c^{(3)}_0,\, c^{(3)}_1} \sum_{t=1}^T \sum_{j=1}^{N_t} K\left( \frac{x_{tj}-x}{B_V} \right)
\left\{ G_{0,t}^X(x_{tj},x_{tj}) - c^{(3)}_0 - c^{(3)}_1(x-x_{tj}) \right\}^2
.
$$

Having the estimates $\bar{R}^X_0(\cdot)$ and $\hat{V}^X(\cdot)$, the measurement error variance $\sigma^2$ is estimated by integrating the difference
\begin{equation}
\label{eq:estimator_sigma2}
\hat{\sigma}^2 = \int_0^1 \left( \hat{V}^X(x) - \bar{R}^X_0(x) \right) \D x. 
\end{equation}
In case the right-hand side of \eqref{eq:estimator_sigma2} is negative, it is recommended to replace it with a small positive number \citep{yao2005functional}.

The spectral density kernels $\{f^X_\omega\}_{\omega\in[-\pi,\pi]}$ are estimated by Bartlett's approach \citep{hormann2015dynamic,rubin2020sparsely}, weighing down higher order lags using Barlett's (triangular) weights defined as
$W_h = (1-|h|/L)$ for $|h|<L$ and $0$ otherwise. The parameter $L\in\mathbb{N}$ is called Bartlett's span parameter and controls the amount of regularisation involved in the spectral density estimation.
For a fixed $\omega\in[-\pi,\pi]$, the spectral density kernel at frequency $\omega$ and $(x,y)\in[0,1]^2$ is estimated as
\begin{equation}\label{eq:estimator_spectral_density_first_formula}
\hat{f}^X_\omega(x,y) = \frac{L}{2\pi} c^{(4)}_0 \qquad\left( \in\mathbb{C} \right)
\end{equation}
where $c^{(4)}_0 \in \mathbb{C}$ is obtained by minimising the following weighted sum of squares
\begin{multline}
\label{eq:minimization-spectral-density}
\left(
	c^{(4)}_0,\, c^{(4)}_1,\, c^{(4)}_2
\right) =
\argmin_{\left(c^{(4)}_0,\,c^{(4)}_1,\,c^{(4)}_2\right) \in \mathbb{C}^3}
\sum_{h=-\QBartlett}^{\QBartlett}
\frac{1}{\mathcal{N}_h}
\sum_{t=\max(1,1-h)}^{\min(T,T-h)}
\stackrel[j\neq k \text{ if } h=0]{}{\sum_{j=1}^{N_{t+h}} \sum_{k=1}^{N_t}}
\big|
G_{h,t}^X(x_{t+h,j}, x_{tk})
e^{-\I h \omega}
-\\- c^{(4)}_0-c^{(4)}_1(x_{t+h,j} - x) - c^{(4)}_2(x_{tk} -y)\big|^2
W_h
\frac{1}{B_R^2} K\left(\frac{x_{t+h,j} - x}{B_R}\right)K\left(\frac{x_{tk} -y}{B_R} \right)
\end{multline}
where $\mathcal{N}_h = (T-|h|) (\bar{N})^2$ for $h\neq 0$, $\mathcal{N}_0 = T (\overline{N^2} - \bar{N})$, and where $\bar{N} = (1/T) \sum_{t=1}^T N_t$ and $\overline{N^2} = (1/T) \sum_{t=1}^T N_t^2$.

The spectral density estimate $\hat{f}^X_\omega(\cdot,\cdot)$ and its operator counterpart $\hat{\mathscr{F}}^X_\omega$ enable us to estimate all autocovariance kernels and operators, thus to estimate the complete second-order structure of $\{X_t\}$, by the inversion formula
\begin{equation}\label{eq:estimated_f_inversion_formula}
\hat{R}^X_h(\cdot,\cdot) = \int_{-\pi}^\pi \hat{f}^X_\omega(\cdot,\cdot) e^{\I h\omega}\D\omega,
\qquad
\hat{\mathscr{R}}^X_h = \int_{-\pi}^\pi \hat{\mathscr{F}}^X_\omega e^{\I h\omega} \D\omega.
\end{equation}

In the following paragraphs we extend these kernel smoothing techniques to estimate the cross-spectral density between the response time series $\{Z\}_{t\in\mathbb{Z}}$ and the regressor time series $\{X_t\}_{t\in\mathbb{Z}}$.

Define the ``raw'' lag-$h$ cross covariances $ G^{ZX}_{h,t}(x_{tj}) = Z_{t+h}Y_{tk} $ where $h=-T+1,\dots,T-1$, $t=\max(1,T-h),\dots,\min(T,T-h)$, and $j=1,\dots,N_t$. Using the local linear kernel smoothing techniques, we estimate the cross-spectral density at frequency $\omega\in[-\pi,\pi]$ and at $x\in[0,1]$ as
\begin{equation}
\label{eq:estimator_cross_density}
\hat{f}^{ZX}_\omega(x) = \frac{L}{2\pi} c^{(5)}_0 \qquad\left( \in\mathbb{C} \right)
\end{equation}
where $c^{(5)}_0$ is realised as the minimiser of the following weighted sum of squares
\begin{multline}
\label{eq:estimator_cross_density_minimization}
\left(
	c^{(5)}_0,\, c^{(5)}_1
\right) =
\argmin_{\left( c^{(5)}_0,\, c^{(5)}_1\right)\in\mathbb{C}^2}
\sum_{h=-L}^L
\sum_{t=\max(1,1-h)}^{\min(T,T-h)}
\sum_{j=1}^{N_t}
\frac{W_h}{B_C} K\left( \frac{x_{tj}-x}{B_C} \right)
\times\\\times
\left| G^{ZX}_{h,t}(x_{tj}) e^{-\I h\omega} - c^{(5)}_0 - c^{(5)}_1(x_{tj}-x)\right|^2
\end{multline}
where $B_C>0$ is a bandwidth parameter.

The solutions to all above least square optimisation problems can be found explicitly by using a standard argument in local-polynomial regression \citep[Section 3.1]{FanJianqing1996Lpma} or \citep[Section B.2]{rubin2020sparsely}. Moreover, the solutions to the spectral density estimator
\eqref{eq:minimization-spectral-density} and \eqref{eq:estimator_cross_density_minimization} depend on a handful of terms independent of the frequency $\omega\in[-\pi,\pi]$, that can be precalculated, and multiplication by complex exponentials. This allows a computationally feasible evaluation even on a fine grid of frequencies.

Once the estimates of the spectral density kernels $\{ \mathscr{F}^X_{\omega} \}_{\omega\in[-\pi,\pi]}$ and the cross-spectral density $\{ \mathscr{F}^{ZX}_{\omega} \}_{\omega\in[-\pi,\pi]}$ have been constructed, we focus attention on the estimation of the frequency response operators $\{\mathscr{B}_\omega\}_{\omega\in[-\pi,\pi]}$. Heuristically, from relation \eqref{eq:spectral_normal_equation}, we would like to write
$\mathscr{B}_\omega = \mathscr{F}^{ZX}_\omega \left( \mathscr{F}^X_\omega\right)^{-1},\:\omega\in[-\pi,\pi]$. This formula is indeed only heuristic because the operator $\mathscr{F}^X_\omega$, being trace class, is not boundedly invertible. The same issue is present also for its empirical counterpart 
$\hat{\mathscr{F}}^{ZX}_\omega \left( \hat{\mathscr{F}}^X_\omega\right)^{-1}$.
Therefore, to achieve consistent estimation, a regularisation of the inverse $\left( \hat{\mathscr{F}}^X_\omega\right)^{-1}$ is required.

Being a self-adjoint trace class operator, $\hat{\mathscr{F}}^X_\omega$ admits the spectral representation
$$ \hat{\mathscr{F}}^X_\omega = \sum_{j=1}^\infty \hat{\lambda}_j^\omega \hat{\varphi}_j^\omega \otimes \hat{\varphi}_j^\omega,\qquad\omega\in[-\pi,\pi],$$
which can be viewed as the empirical version of \eqref{eq:spectral_density_operator_eigendecomposition}.
The difficulty in inverting $\hat{\mathscr{F}}^X_\omega$ can be seen from the fact that $\sum_j\hat{\lambda}_j^\omega=\mathrm{trace}\{ \hat{\mathscr{F}}^X_\omega\} <\infty$, implying that $\hat{\lambda}_j^\omega$ decays at least as fast as $j^{-(1+\delta)}$, $\delta>0$. It is the small values of $\lambda_j^\omega$ that cause problems and there are two classical strategies to overcome the issue: spectral truncation and the Tikhonov regularisation.
\begin{enumerate}
\item \textit{Spectral truncation.} The inverse  $\left( \hat{\mathscr{F}}^X_\omega\right)^{-1}$ is replaced by
$$\sum_{j=1}^{K_\omega} \frac{1}{\lambda_j^\omega} \hat{\varphi}_j^\omega \otimes \hat{\varphi}_j^\omega ,\qquad\omega\in[-\pi,\pi],$$
where $K_\omega \in\mathbb{N},\,\omega\in[-\pi,\pi],$ is the spectral truncation parameter that needs to grow to infinity sufficiently slowly to allow for the consistency. It may or may not depend on the frequency $\omega\in[-\pi,\pi]$.

The estimator of the spectral transfer function becomes
\begin{equation}
\label{eq:regularisation_spectral_transfer_truncation}
\hat{\mathscr{B}}_\omega^{trunc} =
\sum_{j=1}^{K_\omega}
\frac{1}{\lambda_j^\omega}
\left\langle
\hat{\varphi}_j^\omega,\cdot
\right\rangle
\hat{\mathscr{F}}^{ZX}_\omega  \hat{\varphi}_j^\omega,
\qquad\omega\in [-\pi,\pi].
\end{equation}
We opt to implement the spectral truncation by relying on \textit{eigenvalue thresholding} approach \citep{hormann2015estimation} where we implement the eigenvalue threshold selection by cross-validation (more in Section~\ref{sec:numerical experiments}).

\item \textit{Tikhonov regularisation.} Here, the inverse of $\hat{\mathscr{F}}^X_\omega$ is replaced by
$$
\left( \hat{\mathscr{F}}_\omega^X + \rho \mathcal{I} \right)^{-1} =
\sum_{j=1}^\infty \frac{1}{\lambda_j^\omega + \rho^\omega} \hat{\varphi}_j^\omega \otimes \hat{\varphi}_j^\omega,\qquad\omega\in[-\pi,\pi],
$$
where $\mathcal{I}$ is the identity operator on $\mathcal{H}$ and the Tikhonov regularisation parameter $\rho >0$ tends to zero as $T\to\infty$ slowly enough to allow for consistency. Even though the parameter $\rho$ may, in general, depend on $\omega$ we carry out further analysis with the frequency independent parameter. We do so because in the implementation (more in Section~\ref{sec:numerical experiments}) we select the tuning parameter $\rho$ using the cross-validation where it is feasible to optimize over a single (frequency independent) tuning parameter.

The estimator of the spectral transfer function becomes
\begin{equation}
\label{eq:regularisation_spectral_transfer_Tikhonov}
\hat{\mathscr{B}}_\omega^{Tikh} = \hat{\mathscr{F}}^{ZX}_\omega \left( \hat{\mathscr{F}}^X_\omega + \rho \mathcal{I} \right)^{-1} =
\sum_{j=1}^\infty \frac{1}{\lambda_j^\omega + \rho}
\left\langle
\hat{\varphi}_j^\omega,\cdot
\right\rangle
\hat{\mathscr{F}}^{ZX}_\omega  \hat{\varphi}_j^\omega,
\qquad\omega\in [-\pi,\pi].
\end{equation}
This form of regularisation adopted and studied by \cite{pham2018methodology}.
\end{enumerate}

Once the estimators of the spectral transfer operator $\hat{\mathscr{B}}_\omega$ have been constructed by either of the above regularisation techniques, the filter coefficients are estimated by
\begin{align}
\label{eq:fourier_truncation_filter}
\hat{\mathcal{B}}_k^{trunc} &= \frac{1}{2\pi}\int_{-\pi}^\pi \hat{\mathscr{B}}_\omega^{trunc} e^{\I \omega k}\D\omega,
\qquad k\in\mathbb{Z}, \\
\label{eq:fourier_tikhonov_filter}
\hat{\mathcal{B}}_k^{Tikh} &= \frac{1}{2\pi}\int_{-\pi}^\pi \hat{\mathscr{B}}_\omega^{Tikh} e^{\I \omega k}\D\omega,
\qquad k\in\mathbb{Z}.
\end{align}

\subsection{Forecasting the response process}
\label{subsec:forecasting_the_response_process}

In Section \ref{sec:model} we assumed the data to be available up to time $T$ for both the regressor time series $\{X_t\}$ as well as the response time series $\{Z_t\}$.
It may very well happen, though, that the measurement of the response variable is terminated at an earlier time $S$ where $1<S<T$ but the measurements of the regressor time series $\{X_t\}$ are available until time $T$. In this case it is of interest to forecast the unobserved values of $Z_{S+1},\dots,Z_{T}$.
It turns out that the forecast of $Z_{S+1},\dots,Z_{T}$ is surprisingly straightforward, once the model dynamics have been estimated.
We present our forecasting method firstly by assuming the model dynamics are known, and then we plug-in its estimates.

Since the filter coefficients satisfy the assumption \ref{assumption:A.3}, their norms converge to zero as $|k|\to\infty$. Choose a constant $M\in\mathbb{N}$ such that $\|\mathcal{B}_k\|$ is negligible for $|k|>M$.
Denote $\mathbb{X} = [X_{-M+1},\dots,X_{T+M}] \in\mathcal{H}^{T+2M}$ the random element representing ``stacked'' curves of the latent regressors time series.

The vector $\mathbb{Y} = (Y_{11},\dots,T_{1N_1}, \dots, Y_{T,1}, \dots,Y_{T,N_{T}} ) \in \mathbb{R}^{\mathscr{N}_1^{T}}$ consists of all observed data of the regressor time series where $\mathscr{N}_1^{T} = \sum_{t=1}^{T} N_t$ is the total number of observations up to time $T$.
The measurement errors $\{\epsilon_{tj}\}$ are stacked into a vector denoted $\mathcal{E}\in\mathbb{R}^{\mathscr{N}_1^{T}}$. 
Further define the evaluation operators
$H_t : \mathcal{H} \to \mathbb{R}^{N_t}, \xi\mapsto (\xi(x_{t1}),\dots,\xi(x_{tN_t})) $ for each $t=1,\dots,T$ and
the ``joint'' evaluation operator
$\mathbb{H} : \mathcal{H}^{T+2M} \to \mathbb{R}^{\mathscr{N}_1^{T}}, [\xi_{-M+1},\dots,\xi_{T+M}]\mapsto (H_1 \xi_1, \dots, H_{T} \xi_{T})$.
Hence the observation scheme \eqref{eq:observation_scheme} becomes $\mathbb{Y} = \mathbb{H} \mathbb{X} + \mathcal{E}$.

By the symbol $\Pi(\cdot| \mathbb{Y})$ we denote the best linear unbiased predictor of the term represented by the dot, given the data $\mathbb{Y}$. A key observation that simplifies the forecasting is the fact that we may predict the latent functional data first, and then plug them into the filter coefficients. The following proposition summarises this assertion formally.

\begin{proposition}
\label{prop:equivalence_of_BLUP}
The best linear unbiased predictor of $Z_s$ given $\mathbb{Y}$, denoted as $\Pi(Z_s\vert \mathbb{Y})$, is equivalent to constructing the best linear unbiased predictors of $X_t$ given $\mathbb{Y}$, denoted as $\Pi(X_t\vert \mathbb{Y})$,  for all $t\in\mathbb{Z}$ and then applying the filter coefficients $\{\mathcal{B}_k\}_{k\in\mathbb{Z}}$ to these predictions:
$$
\Pi\left(Z_s \vert \mathbb{Y} \right) =
\sum_{k\in\mathbb{Z}} \mathcal{B}_k 
\Pi\left( X_{s-k}  \vert \mathbb{Y} \right), \qquad s\in\mathbb{Z}.
$$
\end{proposition}

In the following paragraph we explain how to construct the predictors $\Pi\left( X_{t} \vert \mathbb{Y} \right),\,t=-M+1,\dots,T+M$, following \citet{rubin2020sparsely}. First, we comment on the distributional properties of the above defined terms. The random element $\mathbb{X}$ inherits the second order structure of $\{X_t\}$, thus $\E\mathbb{X} = 0$ and
\begin{equation}
\label{eq:mathbb_X_cov}
\var(\mathbb{X}) \equiv \mathbbSigma =
\begin{bmatrix}
\mathscr{R}^X_0 & \left(\mathscr{R}^X_1\right)^* & \dots & \left(\mathscr{R}^X_{T+2M-2}\right)^* & \left(\mathscr{R}^X_{T+2M-1}\right)^* \\
\mathscr{R}^X_1 & \mathscr{R}^X_0 & \dots & \left(\mathscr{R}^X_{T+2M-1}\right)^* & \left(\mathscr{R}^X_{T+2M-2}\right)^* \\
\vdots & \vdots & \ddots & \vdots & \vdots \\
\mathscr{R}^X_{T+2M-2} & \mathscr{R}^X_{T+2M-3} & \dots & \mathscr{R}^X_0 & \left(\mathscr{R}^X_1\right)^* \\
\mathscr{R}^X_{T+2M-1} & \mathscr{R}^X_{T+2M-2} & \dots & \mathscr{R}^X_{-1} & \mathscr{R}^X_0 \\
\end{bmatrix}
\end{equation}
understood as a linear operator acting on $\mathcal{H}^{T+2M}$ and where $*$ denotes the adjoint operator.
Moreover, due to independence of the measurement errors, $\var(\mathcal{E}) = \sigma^2 I_{\mathscr{N}_1^{T}}$ where $I_{\mathscr{N}_1^{T}}$ is the identity matrix of size $\mathscr{N}_1^{T} \times \mathscr{N}_1^{T}$.

The best linear unbiased predictor of $\mathbb{X}$ given $\mathbb{Y}$ is then given by
\begin{equation}\label{eq:BLUP}
\Pi(\mathbb{X} \vert \mathbb{Y}) =
\mathbbSigma \mathbb{H}^*
\left( \mathbb{H} \mathbbSigma \mathbb{H}^* + \sigma^2 I_{\mathcal{N}_1^T} \right)^{-1}
\mathbb{Y}
\quad\in\mathcal{H}^{T+2M}
\end{equation}
where the term $ \mathbb{H} \mathbbSigma \mathbb{H}^* + \sigma^2 I_{\mathcal{N}_1^T}$ is a matrix of size $\mathscr{N}_1^T\times\mathscr{N}_1^T$ and its inverse is well conditioned since $\sigma^2>0$.
The best linear predictor of each functional datum $X_t, t=-M+1,\dots,T+M$ given the observed data $\mathbb{Y}$, denoted as $\Pi(X_t\vert\mathbb{Y})$, is then given by the projection
$\Pi(X_t\vert\mathbb{Y}) = P_t \Pi( \mathbb{X}_T \vert\mathbb{Y})$ where
$P_t : \mathcal{H}^{T+2M}\to\mathcal{H}, [\xi_{-M+1},\dots,\xi_{T+M}]\mapsto \xi_t$ is the projection operator for $t=-M+1,\dots,T+M$.

The predictor \eqref{eq:BLUP} requires the knowledge of the unknown dynamics of the regressor time series through the autocovariance operators \eqref{eq:mathbb_X_cov} as well as the measurement error variance $\sigma^2$.
Instead of these parameters we plug-in the estimated counterparts \eqref{eq:estimated_f_inversion_formula} and \eqref{eq:estimator_sigma2} and denote these estimated predictors as $\hat{\Pi}(\cdot \vert \mathbb{Y})$.

In summary, the forecasting algorithm consists of the following steps:
\begin{enumerate}
\item From the measurements $\mathbb{Y}$ realised on the regressor time series $\{ X_t \}$ estimate the spectral density $\{ \hat{\mathscr{F}}^X_\omega \}_{\omega\in[-\pi,\pi]}$ and the measurement error variance $\hat{\sigma}^2$. 
Using the formula \eqref{eq:estimated_f_inversion_formula}, integrate the estimated spectral density to obtain the complete space time covariance $\{\hat{\mathscr{R}}^X_h\}_{h\in\mathbb{Z}}$ of the regressor time series $\{ X_t\}$.

\item From the measurements $\mathbb{Y}$ and the observed response times series $Z_1,\dots,Z_S$, estimate the cross-spectral density $\{\hat{\mathscr{F}}^{ZX}_\omega\}_{\omega\in[-\pi,\pi]}$. Using either truncation regularisation \eqref{eq:regularisation_spectral_transfer_truncation} or Tikhonov regularisation \eqref{eq:regularisation_spectral_transfer_Tikhonov}, estimate the spectral transfer function
$\{ \mathscr{B}_\omega \}_{\omega\in[-\pi,\pi]}$.
By means of formula \eqref{eq:fourier_truncation_filter} or \eqref{eq:fourier_tikhonov_filter}, depending on the regularisation scheme, integrate the spectral transfer function to obtain the filter coefficients $\{\hat{\mathcal{B}}_k^{trunc}\}_{k\in\mathbb{Z}}$ or $\{\hat{\mathcal{B}}_k^{Tikh}\}_{k\in\mathbb{Z}}$.

\item \label{item:algorithm_step_3}
Choose $M$ such that the estimated filter coefficients $\hat{\mathcal{B}}_{k}^{trunc}$ (or $\hat{\mathcal{B}}_{k}^{Tikh}$) are negligible for $|k|>M$.
Using the methodology explained at the beginning of this section, construct the prediction of the latent functional data $\hat{\Pi}(X_{-M+1}\vert \mathbb{Y}),\dots, \hat{\Pi}(X_{T+M}\vert \mathbb{Y})$.

\item \label{item:algorithm_step_4}
For each $s=S+1,\dots,T$, construct the forecast 
$ \hat{\Pi}(Z_s \vert \mathbb{Y}) = \sum_{k = -M}^{M} \hat{\mathcal{B}}_k^{trunc} \hat{\Pi}({X}_{s-k}\vert \mathbb{Y}) $, in the case of spectral truncation or 
$ \hat{\Pi}(Z_s \vert \mathbb{Y}) = \sum_{k = -M}^{M} \hat{\mathcal{B}}_k^{Tikh} \hat{\Pi}({X}_{s-k}\vert\mathbb{Y} ) $), in the case of Tikhonov regularisation.
\end{enumerate}

\section{Asymptotic results}
\label{sec:asymptotics}

In this section we turn to establishing consistency of the proposed estimators of the spectral density \eqref{eq:estimator_spectral_density_first_formula} and the cross-spectral density \eqref{eq:estimator_cross_density}, and the regularised estimators of the filter coefficients \eqref{eq:fourier_truncation_filter} and \eqref{eq:fourier_tikhonov_filter}.
A key requirement for consistency is to control the temporal dependence of $\{X_t\}$. Here we employ cumulant mixing conditions which have been successfully used in functional time series before \citep{panaretos2013fourier,rubin2020sparsely}. Alternatives to the cumulant assumptions are $L^p$-$m$-approximability \citep{hormann2010weakly,hormann2015dynamic,hormann2015estimation} and strong mixing conditions \citep{rubin2020sparsely}.

Bellow we list the assumptions that we will made use of to study the asymptotics of the spectral density estimator $\{\hat{f}_\omega^X\}_{\omega\in[-\pi,\pi]}$ and the cross-spectral density estimator $\{\hat{f}_\omega^{ZX}\}_{\omega\in[-\pi,\pi]}$.

\begin{enumerate}[label=(B{\arabic*})]
\item \label{assumption:newB.1}
The number of measurements $N_t$ at time $t$  are independent  random variables, identically distributed with a random variable $N$ satisfying $N\geq 0$, $\Ez{ N } < \infty$ and $\Prob(N>1)>0$.

\item \label{assumption:newB.2}
The measurement locations $x_{tj}, j=1,\dots,N_t, t=1,\dots,T$ are independent random variables generated from the density $g(\cdot)$ and are independent of the number of measurements $\{N_t\}_{t=1,\dots,T}$. The density $g(\cdot)$ is assumed to be twice continuously differentiable and strictly positive on $[0,1]$.

\item \label{assumption:newB.3}
The autocovariance kernels, $R_h(\cdot,\cdot)$, are twice continuously differentiable on $[0,1]^2$ for each $h\in\mathbb{Z}$.
Moreover, the terms
$$\sup_{x,y\in[0,1]} \left| \frac{\partial^2}{\partial y^{\alpha_1} \partial x^{\alpha\_2}} R^X_h(y,x) \right|$$
are uniformly bounded in $h$ for all combinations of $\alpha_1,\alpha_2 \in\mathbb{N}_0$ satisfying $\alpha_1+\alpha_2=2$.

\item \label{assumption:newB.4}
\newcommand{\xta}{X_{t_1}(x_{t_1})}
\newcommand{\xtb}{X_{t_2}(x_{t_2})}
\newcommand{\xtc}{X_{t_3}(x_{t_3})}
\newcommand{\xtd}{X_{t_4}(x_{t_4})}

The 4-th order cumulant kernel of $\{X_t\}$ is summable in the supremum norm
$$
\sum_{h_1,h_2,h_3 = -\infty}^\infty \sup_{x_1,x_2,x_3,x_4 \in [0,1]} \left|
\cum(X_{h_1}, X_{h_2}, X_{h_3}, X_0)(x_1,x_2,x_3,x_4)
\right| < \infty.
$$
where
\begin{multline*}
\cum\left(X_{t_1}, X_{t_2}, X_{t_3}, X_{t_4}\right)(x_{t_1},x_{t_2},x_{t_3},x_{t_4}) =
\Ez{ \xta\xtb\xtc\xtd }
-\\- \Ez{\xta\xtb}\Ez{\xtc\xtd}
- \Ez{\xta\xtc}\Ez{\xtb\xtd}
-\\- \Ez{\xta\xtd}\Ez{\xtb\xtc}
\end{multline*}
for $t_1,t_2,t_3,t_4 \in\mathbb{N}$ and $x_{t_1},x_{t_2},x_{t_3},x_{t_4} \in[0,1]$.

\item \label{assumption:newB.5}
The sequence $\{R^X_h\}_{h\in\mathbb{Z}}$ satisfies the weak dependence condition
$$ \sum_{h=-\infty}^\infty |h| \sup_{x,y\in[0,1]} \left| R^X_h(x,y) \right| <\infty.$$
\end{enumerate}

\noindent Finally, we state assumptions on the bandwidth parameters $B_R,B_V,B_C$ and the Bartlett span $L$:
\begin{enumerate}[label=(B{\arabic*}),resume]
\Item \label{assumption:newB.6}
$$ B_R \to 0,\qquad T B_R^6\to \infty, $$
\Item \label{assumption:newB.7}
$$ B_V \to 0,\qquad T B_V^4\to \infty, $$
\Item \label{assumption:newB.8}
$$ B_C \to 0,\qquad T B_C^4 \to \infty, $$
\Item \label{assumption:newB.9}
$$ \QBartlett \to \infty, \qquad\QBartlett = o(\sqrt{T} B_R^2),\qquad \QBartlett = o(\sqrt{T} B_C). $$

\end{enumerate}

Under the assumptions \ref{assumption:A.1} --- \ref{assumption:A.3}, \ref{assumption:newB.1} --- \ref{assumption:newB.6}, \ref{assumption:newB.9} the spectral density estimator $\{\hat{f}_\omega^X\}_{\omega\in[-\pi,\pi]}$ is consistent uniformly in the supremum norm \citep{rubin2020sparsely}. Moreover an upper bound for the convergence rate was established.
Assuming further \ref{assumption:newB.7} ensures the consistency of the measurement error variance estimator $\hat{\sigma}^2$ given by formula \eqref{eq:estimator_sigma2}.

Likewise, the estimator of the cross-spectral density $\{\hat{f}_\omega^{ZX}\}_{\omega\in[-\pi,\pi]}$ is consistent by the following proposition.

\begin{proposition}
\label{prop:asymptotics_of_f^ZX}
Under the conditions \ref{assumption:A.1} --- \ref{assumption:A.3}, \ref{assumption:newB.1} --- \ref{assumption:newB.3}, \ref{assumption:newB.8}, \ref{assumption:newB.9}, the cross-spectral density is estimated consistently: 
$$ \sup_{\omega\in[-\pi,\pi]} \sup_{x\in[0,1]}
\left|
\hat{f}^{ZX}_\omega(x) - f^{ZX}_\omega(x)
\right|
= o_p(1)
\qquad\text{as}\quad T\to\infty.
$$
Assuming further the condition \ref{assumption:newB.5}, we obtain the convergence rate:
\begin{equation}
\label{eq:asymptotics_rate_of_f^ZX}
\sup_{\omega\in[-\pi,\pi]} \sup_{x\in[0,1]}
\left|
\hat{f}^{ZX}_\omega(x) - f^{ZX}_\omega(x)
\right| =
O_p\left(L \frac{1}{\sqrt{T}}\frac{1}{B_C}\right)
\qquad\text{as}\quad T\to\infty.
\end{equation}
\end{proposition}

The estimators of the spectral density and the cross-spectral density are essential building blocks for the estimation of the filter coefficients. 
In the following paragraphs we list the conditions for the consistency of the filter coefficients obtained via Tikhonov regularisation \eqref{eq:fourier_tikhonov_filter} and via truncation regularisation \eqref{eq:fourier_truncation_filter}.

First, the following condition is required for the regression model \eqref{eq:regression_model} to be identifiable, regardless of the regularisation method used.
\begin{enumerate}[label=(C)]
\item \label{assumption:C}
For all $\omega\in[-\pi,\pi]$ the operators $\mathscr{F}^X_\omega:\mathcal{H}\to\mathcal{H}$ satisfy $\ker\left( \mathscr{F}^X_\omega \right) = 0$.
\end{enumerate}

To ensure the consistency of the filter coefficients estimator by the Tikhonov method we only need to guarantee that the regularisation parameter vanishes slowly.
\begin{enumerate}[label=(D)]
\item \label{assumption:D}
The Tikhonov regularisation parameter satisfies
\begin{align*}
\frac{1}{\rho}   &L \frac{1}{\sqrt{T}}\frac{1}{B_C} \to 0, \qquad\text{as}\quad T\to\infty,\\
\frac{1}{\rho^2}  &L \frac{1}{\sqrt{T}}\frac{1}{B_R^2} \to 0, \qquad\text{as}\quad T\to\infty.
\end{align*}
\end{enumerate}

The following theorem establishes the consistency of the Tikhonov filter coefficient estimators.
\begin{theorem}
\label{thm:asymptotics_of_B_Tikh}
Under the conditions \ref{assumption:A.1} --- \ref{assumption:A.3}, \ref{assumption:newB.1} --- \ref{assumption:newB.6}, \ref{assumption:newB.8}, \ref{assumption:newB.9}, \ref{assumption:C}, \ref{assumption:D}, the filter coefficient estimators \eqref{eq:fourier_tikhonov_filter} constructed by means of Tikhonov regularisation are consistent in the sense that:
$$ \sup_{k\in\mathbb{Z}} \left\| \hat{\mathcal{B}}^{Tikh}_k - \mathcal{B}_k \right\| = o_p(1)
\qquad\text{as}\quad T\to\infty.
$$
\end{theorem}

We now turn to the truncation estimator \eqref{eq:fourier_truncation_filter} of the filter coefficients,  whose consistency requires more technical assumptions.
We use the result by \citet[Theorem 1]{hormann2015estimation} relies on having consistent estimators of the spectral density and cross-spectral density operators with a known rate of convergence, on a condition on the eigenvalue spacing, and on an assumption that the spectral truncation parameter $K_\omega$ grows sufficiently slowly. In what follows, we review their conditions and adapt them to the setting when the spectral density kernels and the cross-spectral density are estimated by the kernel smoothing methods from sparse noisy observations.

Recall the eigendecomposition of the spectral frequency operator \eqref{eq:spectral_density_operator_eigendecomposition} and that its harmonic eigenvalues and harmonic eigenfunction are denoted $\{ \lambda^\omega_k \}_{k=1}^\infty$ and $\{ \varphi^\omega_k \}_{k=1}^\infty$ respectively.
Define
\begin{align*}
\Lambda_1^\omega &= \lambda_1^\omega - \lambda_2^\omega, \\
\Lambda_k^\omega &= \min\left\{ \lambda_k^\omega-\lambda_{k+1}^\omega, \lambda_{k-1}^\omega - \lambda_k^\omega \right\}, \qquad k\geq 2.
\end{align*}
The following condition guarantees that the eigenspaces belonging to each of the eigenvalues $\{\lambda_m^\omega\}_{m=1}^\infty$ are one-dimensional, hence the eigenfunctions $\{\varphi_m^\omega\}_{m=1}^\infty$ can be identified (up to multiplication by a complex number with modulus 1).
\begin{enumerate}[label=(E{\arabic*})]
\item \label{assumption:E.1}
For all $k\geq 1$ we assume $ \inf_{\omega\in[-\pi,\pi]} \Lambda_k^\omega > 0 $.
\end{enumerate}

Furthermore we need to assume that the truncation parameter $K_\omega$ needs to grow sufficiently slowly.
\begin{enumerate}[resume,label=(E{\arabic*})]
\item \label{assumption:E.2}

$$K_\omega = \min\{ K^{(i)}, 1\leq i \leq 4 \}$$
where
\begin{align*}
K^{(1)} &= \max \left\{ k\geq 1 : \inf_{\omega\in[-\pi,\pi]} \hat{\lambda}_k \geq 2 L T^{-1/2} B_R^{-2} \right\}, \\
K^{(2)} &= \max \left\{ k\geq 1 : L T^{-1/2} B_C^{-1} \int_{-\pi}^\pi W^K_\lambda(\omega)\D\omega \leq 1 \right\}, \\
K^{(3)} &= \max \left\{ k\geq 1 : \int_{-\pi}^\pi \left( W_\lambda^k(\omega)\right)^2 \D\omega
\leq L^{-1/2} T^{1/4} B_R
\right\}, \\
K^{(4)} &= \max \left\{ k\geq 1 :\int_{-\pi}^\pi \left( W_\Lambda^k (\omega) \right)^2 \D\omega
\leq L^{-1/2} T^{1/4} B_R
\right\}
\end{align*}
and where we further define
$$ W_\lambda^k(\omega) =
\left( \sum_{m=1}^k \frac{1}{\left[ \hat{\lambda}_m^\omega \right]^2} \right)^{1/2},\qquad
W_\Lambda^k(\omega) = \left( \sum_{m=1}^k \frac{1}{\left[ \hat{\Lambda}_m^\omega \right]^2} \right)^{1/2}
$$
\end{enumerate}
and $\{\hat{\Lambda}_m^\omega\}$ are the empirical counterparts of $\{\Lambda_m^\omega\}$ where the estimates $\{\hat{\lambda}_j^\omega\}$ are plugged-in.

Under the above stated assumptions, the filter coefficient estimator \eqref{eq:fourier_truncation_filter} obtained by means of truncation regularisation is consistent.

\begin{theorem}
\label{thm:asymptotics_of_B_trunc}
Under the conditions \ref{assumption:A.1} --- \ref{assumption:A.3}, \ref{assumption:newB.1} --- \ref{assumption:newB.6}, \ref{assumption:newB.8}, \ref{assumption:newB.9}, \ref{assumption:C}, \ref{assumption:E.1}, \ref{assumption:E.2}, the filter coefficients estimates \eqref{eq:fourier_truncation_filter} constructed by the spectral truncation regularisation are consistent:
$$ \sup_{k\in\mathbb{Z}} \left\| \hat{\mathcal{B}}^{trunc}_k - \mathcal{B}_k \right\| = o_p(1)
\qquad\text{as}\quad T\to\infty.$$
\end{theorem}

\section{Numerical experiments}
\label{sec:numerical experiments}

\subsection{Simulation setting}
\label{subsec:simulation_setting}

In this simulation study we asses the performance of the proposed methodology on the basis of two criteria:
the estimation error of the filter coefficients estimator \eqref{eq:fourier_tikhonov_filter}, and 
the prediction error of the forecasts of the response process (Section \ref{subsec:forecasting_the_response_process}).
We also compare the performance of the two regularisation techniques, and corroborate that neither dominates the other. To illustrate this, we introduce two different filter coefficient function for the lagged regression, see \eqref{eq:shape_A} and \eqref{eq:shape_B}, in which one technique is expected to perform better than the other, and vice versa.

The MATLAB code and the results of the simulation are openly available on GitHub \citep{rubin2019github_lagged_regression}.

We simulate the functional regressor series $\{X_t\}_{t\in\mathbb{Z}}$ as functional linear processes: either a  functional autoregressive process or a functional moving average process.
We define $\{E_t\}_{t\in\mathbb{Z}}$ to be the stochastic innovation term for these linear processes.
It is assumed to be a sequence of independent identically distributed zero-mean Gaussian random variables in $\mathcal{H}= \Ltwo$ with the covariance kernel $K(x,y)$ given by
\begin{align}
\label{eq:covariance_kernel_in_simulations}
K(x,y) =
           &\sin(  2\pi x )\sin(  2\pi y ) +\\ \nonumber
    + 0.6  &\cos(  2\pi x )\cos(  2\pi y ) +\\ \nonumber
    + 0.3  &\sin(  4\pi x )\sin(  4\pi y ) +\\ \nonumber
    + 0.1  &\cos(  4\pi x )\cos(  4\pi y ) +\\ \nonumber
    + 0.1  &\sin(  6\pi x )\sin(  6\pi y ) +\\ \nonumber
    + 0.1  &\cos(  6\pi x )\cos(  6\pi y ) +\\ \nonumber
    + 0.05 &\sin(  8\pi x )\sin(  8\pi y ) +\\ \nonumber
    + 0.05 &\cos(  8\pi x )\cos(  8\pi y ) +\\ \nonumber
    + 0.05 &\sin( 10\pi x )\sin( 10\pi y ) +\\ \nonumber
    + 0.05 &\cos( 10\pi x )\cos( 10\pi y ), \qquad x,y\in[0,1].
\end{align}

We simulate realisations of the functional autoregressive process of order 1 and the functional moving average process of order 4 defined in the following two settings:
\begin{itemize}
\item[\FARno] \label{item:sim_FAR}
The process $\{X_t\}_{t\in\mathbb{Z}}$ is a functional autoregressive process of order 1 \citep{bosq2012linear} defined by the iteration
\begin{equation}\label{eq:item:sim_FAR}
X_{t+1} = \mathcal{A} X_t + E_t, \qquad t\in\mathbb{Z}.
\end{equation}
The operator $\mathcal{A}$ is assumed to be a Hilbert-Schmidt operator and we define its kernel as $A(x,y) = \kappa \sin( x-y ), x,y\in[0,1]$, where $\kappa>0$ is chosen such that $\|\mathcal{A}\|_{L(\mathcal{H})} = 0.7$ and $\|\cdot\|_{L(\mathcal{H})}$ is the operator norm in the space of linear operators on $\mathcal{H}$.

\item[\FMAno] \label{item:sim_FMA}
The process $\{X_t\}_{t\in\mathbb{Z}}$ is considered to be the functional moving average process of order 4 defined by
\begin{equation}\label{eq:item:sim_FMA}
X_t = E_t + \mathcal{M}_1 E_{t-1} + \mathcal{M}_2 E_{t-2} + \mathcal{M}_3 E_{t-3} + \mathcal{M}_4 E_{t-4}, \qquad t\in\mathbb{Z}.
\end{equation}
The operators $\mathcal{M}_1,\dots,\mathcal{M}_4$ are assumed to be Hilbert-Schmidt and given by their kernels 
$
M_1(x,y) = \kappa_1 \sin(  x+y ),
M_2(x,y) = \kappa_2 \sin( 1-x+y ),
M_3(x,y) = \kappa_3 \sin( 1+x-y ),
M_4(x,y) = \kappa_4 \sin( 2-x-y ),
$ for $x,y\in[0,1]$, respectively. The constants $\kappa_1>0,\dots,\kappa_4>0$ are chosen so that
$
\|\mathcal{M}_1\|_{L(\mathcal{H})} = 0.8,
\|\mathcal{M}_2\|_{L(\mathcal{H})} = 0.6,
\|\mathcal{M}_3\|_{L(\mathcal{H})} = 0.4,
\|\mathcal{M}_4\|_{L(\mathcal{H})} = 0.2
$
respectively.
\end{itemize}

The functional autoregressive process \FARno{}, defined uniquely by the equation \eqref{eq:item:sim_FAR}, and the functional moving average process \FMAno{} are stationary and Gaussian \citep{bosq2012linear}.

Each of the above defined functional processes is simulated with a varying time length $T\in\{300,600,900,1200\}$. The sparse observations \eqref{eq:observation_scheme} are generated by fixing the maximal number of observations per curve $N^{max}\in\{10,20,40,60\}$. For each curve, an integer valued random variable is drawn with uniform distribution on $\{0,\dots,N^{max}\}$ corresponding to the number of spatial locations where the $X_t$ is observed (with measurement error, to be defined). The measurement locations $x_{tj}$ are sampled as uniform random variables on $[0,1]$. At each $x_{tj}$ location, the measurement error is added as a realisation of a centred Gaussian random variable with variance $\sigma^2>0$.
The variance $\sigma^2>0$ is chosen so the signal-to-noise ratio is $\tr( \mathscr{R}_0^X )/\sigma^2 = 20$.

For the lagged regression model \eqref{eq:regression_model} we consider regression models where only certain filter coefficients $b_k$ are nonzero functions. In particular, the nonzero filter coefficients are considered to be either of the two following options

\begin{itemize}
\Item[\textbf{(A)}]
\begin{equation} \label{eq:shape_A}
\beta_A(x) = \cos( 4 \pi x ), \qquad x\in[0,1]
\end{equation}
\Item[\textbf{(B)}]
\begin{equation} \label{eq:shape_B}
\beta_B(x) = \sin( 2 \pi x ), \qquad x\in[0,1]
\end{equation}
\end{itemize}
The considered kernels and filter functions are visualised on Figure \ref{fig:kernels}.

We consider 3 lagged regression schemata with a varying set of nonzero filter coefficients.
\begin{itemize}
\item[\Ba] The filter coefficients $b_0, b_1$ are set to either $\beta_A$ or $\beta_B$.
\item[\Bb] The filter coefficients $b_0, b_3$ are set to either $\beta_A$ or $\beta_B$. 
\item[\Bc] The filter coefficients $b_0, b_1, b_2, b_3, b_4, b_5$ are nonzero but with decaying magnitude. They are set to either
$$ \left( b_0, b_1, b_2, b_3, b_4, b_5 \right) = \left(
\beta_A, 0.9 \beta_A, 0.7 \beta_A, 0.5 \beta_A, 0.3 \beta_A, 0.1 \beta_A
\right)  $$
or 
$$ \left( b_0, b_1, b_2, b_3, b_4, b_5 \right) = \left(
\beta_B, 0.9 \beta_B, 0.7 \beta_B, 0.5 \beta_B, 0.3 \beta_B, 0.1 \beta_B
\right)  $$
depending on the chosen shape of the filter coefficients.
\end{itemize}
The variance of the measurement error $\{e_t\}_{t\in\mathbb{Z}}$ is set to be $\tau^2 = 0.001$.

\begin{figure}[h]
\centering
\makebox[\textwidth][c]{
\includegraphics[width=0.99\textwidth]{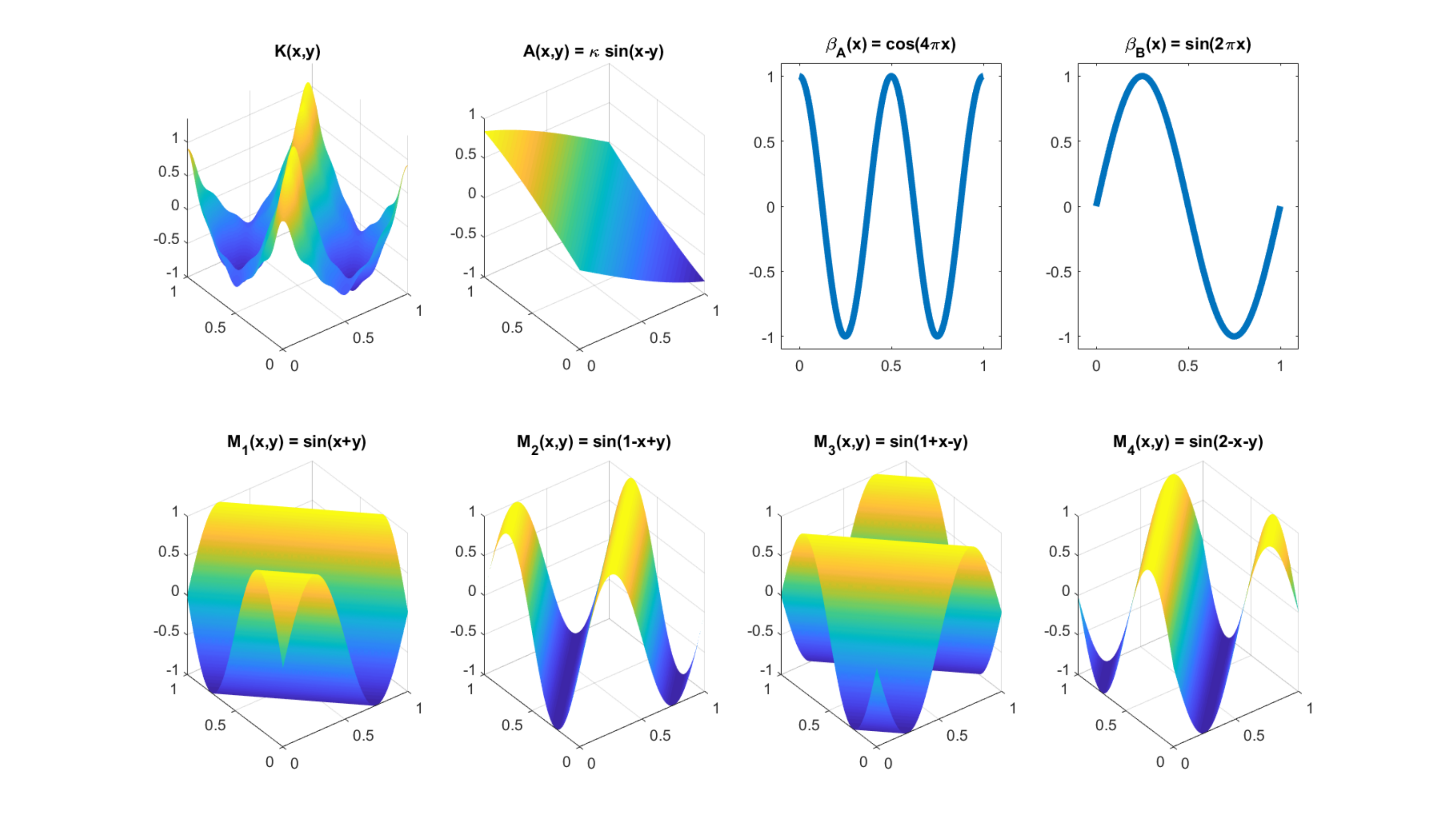}
}
\caption{
\textbf{Top row:} The covariance kernel of the stochastic innovation $K$ \eqref{eq:covariance_kernel_in_simulations}, the autoregressive kernel $A$ \eqref{eq:item:sim_FAR}, the filter coefficients $\beta_A$ \eqref{eq:shape_A} and $\beta_B$ \eqref{eq:shape_B} respectively.
\textbf{Bottom row:} The kernels of the moving average process $M_1,\dots,M_4$ \eqref{eq:item:sim_FMA}.
}
\label{fig:kernels}
\end{figure}

For each combination of the settings, i.e. each of the 2 linear processes of $\{X_t\}_{t\in\mathbb{Z}}$, each of 4 length parameters $T\in\{300,600,900,1200\}$, each of 4 sampling density parameters $N^{max}\in\{10,20,40,60\}$, each of 3 regression schemata, and 2 shapes of the filter coefficients, we run 90 independent runs.
The simulations have to be obviously performed in a finite dimension. We approximate the infinite dimensional dynamics of the process $\{X_t\}_{t\in\mathbb{Z}}$ by the B-Spline basis of dimension 21.  This follows the same scheme as \cite{rubin2020sparsely}, where further details on the definition of the basis as well as the approximation procedure can be found.
Moreover, we consider also the regime of complete functional observations in the setting of \cite{hormann2015estimation} in order to compare how much information is lost due to sparse sampling.

\subsection{Estimation procedure}
\label{subsec:Estimation Procedure}


The proposed methodology requires the selection of the tuning parameters. We implemented the choice of the bandwidths $B_R$ and $B_C$ for the estimation of $\{\mathscr{F}_\omega^X\}_{\omega\in[-\pi,\pi]}$ and $\sigma^2$ by means of K-fold cross validation, as explained in detail in \cite{rubin2020sparsely}. The Bartlett span parameter for the estimation of the spectral density is set to
$ L = \lfloor 2 T^{1/3} \rfloor $.

In order to select the regularisation parameters for either of the two proposed regularsation methods we resort to holdout cross-validation. We split the response time series $Z_1,\dots,Z_T$ into the training set $Z_1,\dots,Z_S$ and the test set $Z_{S+1},\dots,Z_T$.
The split is set to be 80:20 in favour of the training set, i.e. $S=0.8 T$.
The cross-spectral density $\{\mathscr{F}_\omega^{ZX}\}_{\omega\in[-\pi,\pi]}$ is estimated from the data in the training set $Z_1,\dots,Z_S$ and the bandwidth parameter $B_C$ is selected by K-fold cross-validation within the training set.

The truncation estimator \eqref{eq:regularisation_spectral_transfer_truncation} is constructed by means of  \textit{eigenvalue thresholding} \citep{hormann2015estimation}. Specifically, we set 
$K_\omega(\upsilon) = \argmax_{m\geq 1} \left\{ \hat{\lambda}_m^\omega > \upsilon \right\}$
where $\upsilon>0$ is a parameter to be chosen by holdout cross-validation in the following way. Having estimated  $\{\mathscr{F}_\omega^X\}_{\omega\in[-\pi,\pi]}$ and $\sigma^2$ from the sparsely observed regressor time series $\{X_t\}_{t=1}^T$, and the cross-spectral density
$\{\mathscr{F}_\omega^{ZX}\}_{\omega\in[-\pi,\pi]}$ from $\{X_t\}_{t=1}^T$ and the training partition of the response $Z_1,\dots,Z_S$, the spectral transfer function is estimated by the formula \eqref{eq:regularisation_spectral_transfer_truncation} using a candidate value of $K_\omega(\upsilon)$. The forecasts $\hat{Z}_{S+1},\dots,\hat{Z}_T$ are produced by the methodology outlined in Section \ref{subsec:forecasting_the_response_process}. Comparing the forecasts with the true values of $Z_{S+1},\dots,Z_T$ yields a mean square forecast error on the holdout partition which we minimse with respect to $\upsilon$. 
The Tikhonov estimator \eqref{eq:regularisation_spectral_transfer_Tikhonov} involves the selection of the parameter $\rho$. In the same way as for the truncation regularisation, we chose $\rho$ by holdout cross-validation based on the mean square forecast error on $Z_{S+1},\dots,Z_T$.

In the case of complete functional observations we again use holdout cross-validation for the selection of the {eigenvalue thresholding} parameter as well as the Tikhonov parameter.

\subsection{Evaluation criteria}
\label{subsec:evaluation_criteria}

We asses the estimation error of the filter coefficients by the mean square error criterion:
\begin{equation}
\label{eq:MSE_delta_B}
\delta^{\mathcal{B}} =
\sum_{k\in\mathbb{Z}} \left\| \hat{\mathcal{B}}_k-\mathcal{B}_k \right\|^2
\end{equation}


Next we want to asses the forecasting performance of the proposed methodology. Because the entire sample was used for fitting the model dynamics, we simulate an independent copy of the regressor time series, denoted as
$\{ X^{copy}_t \}_{t=1}^T$,
and the response process $\{Z^{copy}_t\}_{t=1}^T$. Using the estimates of the model dynamics from the original data, we produce the predictions $\{\hat{Z}^{copy}_t\}_{t=1}^T$ from the sparsely observed $\{ X^{copy}_t \}_{t=1}^T$ and compare with the true values $\{Z^{copy}_t\}_{t=1}^T$. The prediction relative mean square error is then defined
\begin{equation}
\label{eq:RMSE_delta_pred}
\delta^{pred} = \frac{1}{T} \sum_{t=S}^T \frac{\left( \hat{Z}^{copy}_t -Z^{copy}_t \right)^2}{\var(Z_0)} .
\end{equation}

Moreover, we include the prediction error of the oracle estimator that assumes that both the dynamics of the regressor time series $\{X_t\}_{t\in\mathbb{Z}}$ and the filter coefficients $\{\mathcal{B}_k\}_{k\in\mathbb{Z}}$ are known. The oracle estimator completes the steps \ref{item:algorithm_step_3} and \ref{item:algorithm_step_4} of the algorithm of Section \ref{subsec:forecasting_the_response_process} where the estimates $\{\hat{R}^X_h(\cdot,\cdot)\}_{h\in\mathbb{Z}}$ and $\{\hat{\mathcal{B}}^{trunc}_k\}_{k\in\mathbb{Z}}$ (or $\{\hat{\mathcal{B}}^{Tikh}_k\}_{k\in\mathbb{Z}}$) are replaced by the true values of $\{R^X_h(\cdot,\cdot)\}_{h\in\mathbb{Z}}$ and $\{\mathcal{B}_k\}_{k\in\mathbb{Z}}$.

\subsection{Results of numerical experiments}

\begin{figure}
\centering
\makebox[\textwidth][c]{
\includegraphics[width=1\textwidth]{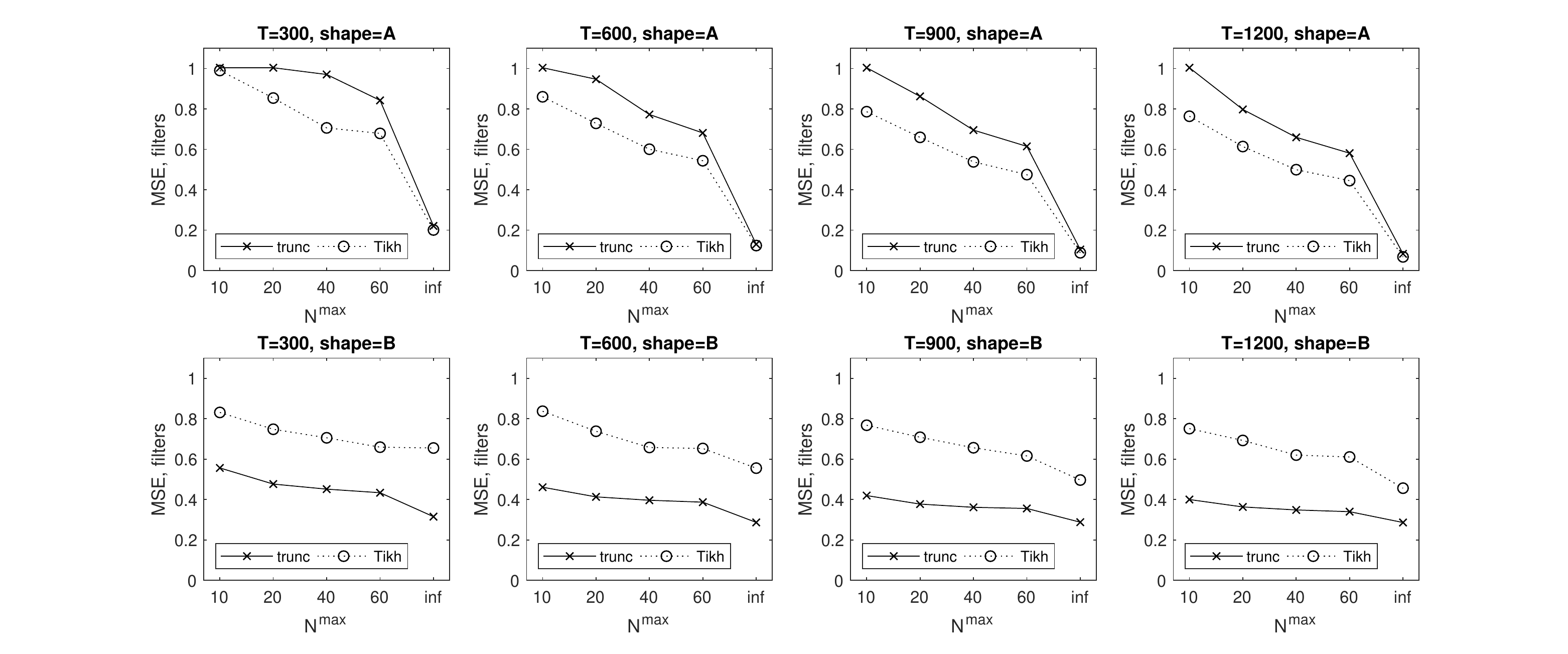}
}
\caption{
The median mean square error $\delta^{\mathcal{B}}$ of the filter coefficient estimates \eqref{eq:MSE_delta_B} for the truncation regularsation (``trunc'') and Tikhonov regularsation method (``Tikh''), displayed as a function of the time series length $T\in\{300,600,900,1200\}$ and the the maximum number of the observation locations $N^{max}\in\{10,20,40,60,\text{inf}\}$ where ``inf'' stands for the fully observed functional data. The top and the bottom row show the results for the filter coefficients of the shapes \eqref{eq:shape_A} and \eqref{eq:shape_B} respectively.
}
\label{fig:filter_mse}
\end{figure}

\begin{figure}
\centering
\includegraphics[width=1\textwidth]{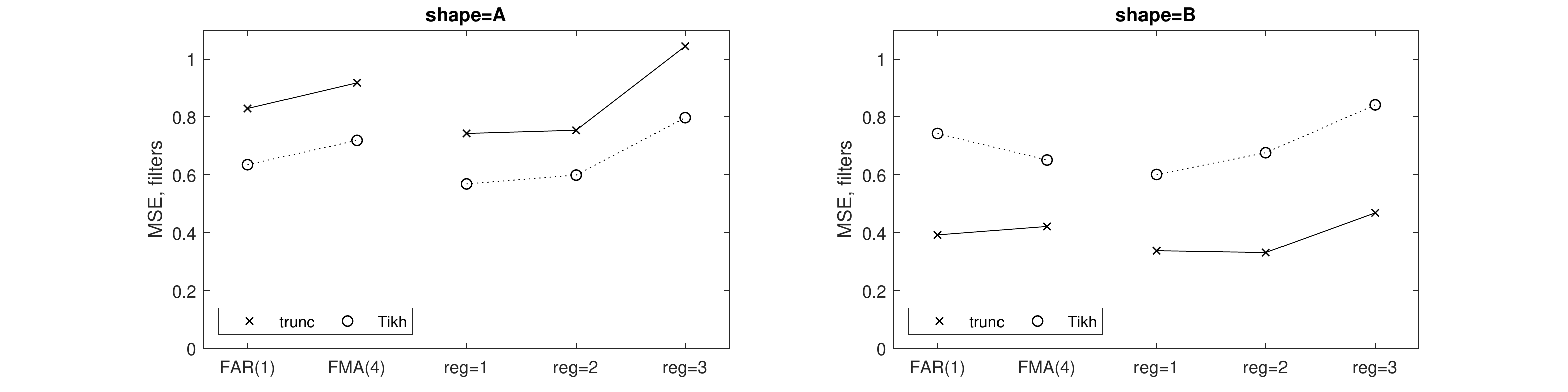}
\caption{
The median mean square error $\delta^{\mathcal{B}}$ of the filter coefficient estimates \eqref{eq:MSE_delta_B} for the truncation regularsation (``trunc'') and Tikhonov regularsation method (``Tikh'') with respect to the simulated dynamics of $\{X_t\}_{t}$ and the regression scheme. The results are aggregated over all sparse observation setups $T\in\{300,600,900,1200\}$ and $N^{max}\in\{10,20,40,60\}$. The left and the right figures show the results for the filter coefficients of the shapes \eqref{eq:shape_A} and \eqref{eq:shape_B} respectively.
}
\label{fig:filter_mse_per_process}
\end{figure}

\begin{figure}
\centering
\makebox[\textwidth][c]{
\includegraphics[width=1\textwidth]{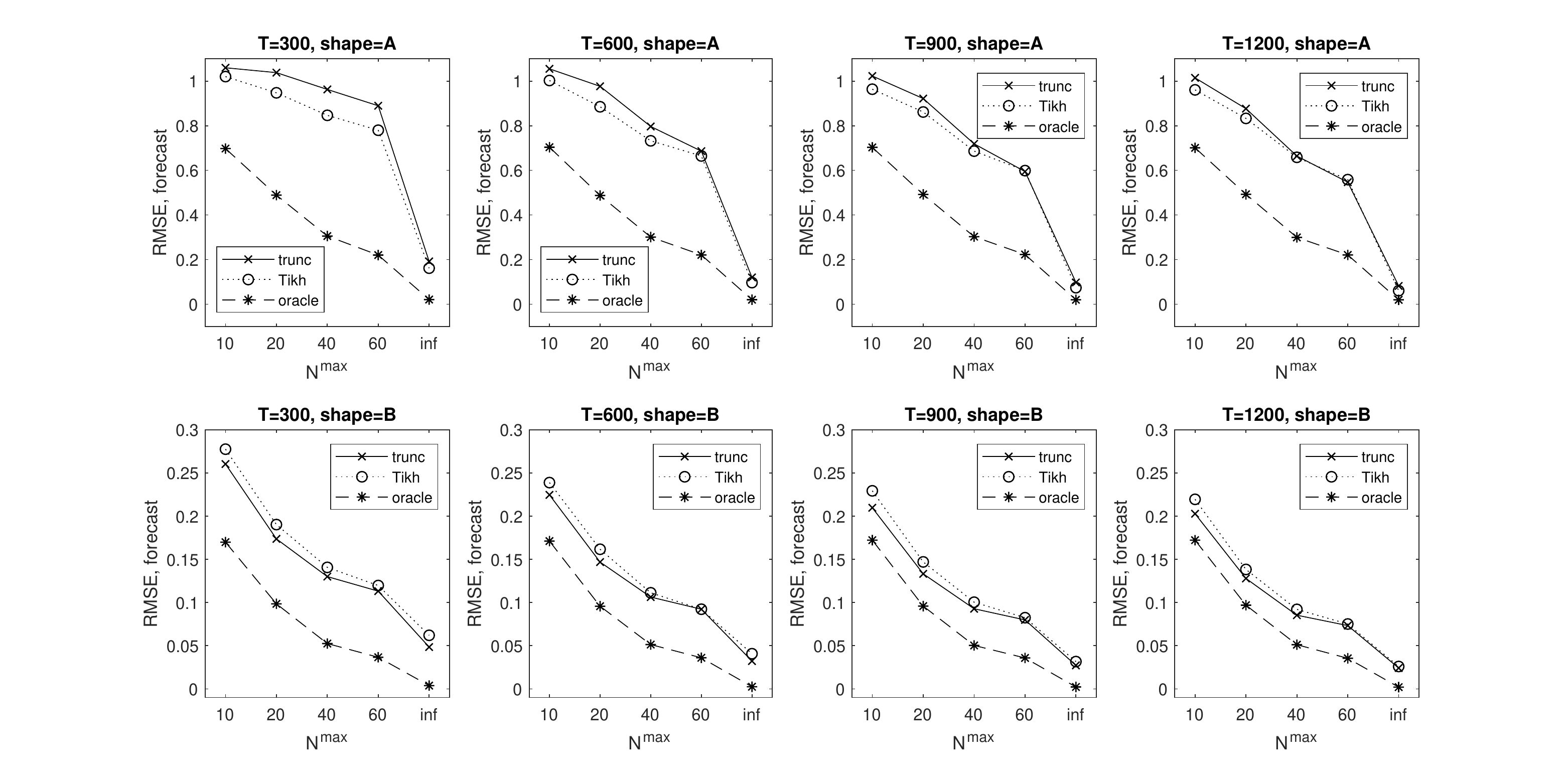}
}
\caption{
The median mean square prediction error $\delta^{pred}$ \eqref{eq:RMSE_delta_pred} for the truncation regularisation (``trunc''), Tikhonov regularisation method (``Tikh''), and the oracle estimator, displayed as a function of the time series length $T\in\{300,600,900,1200\}$ and the the maximum number of the observation locations $N^{max}\in\{10,20,40,60,\text{inf}\}$ where ``inf'' stands for the fully observed functional data. The top and the bottom row show the results for the filter coefficients of the shapes \eqref{eq:shape_A} and \eqref{eq:shape_B} respectively.
}
\label{fig:forecast_rmse}
\end{figure}

\begin{figure}
\centering
\includegraphics[width=1\textwidth]{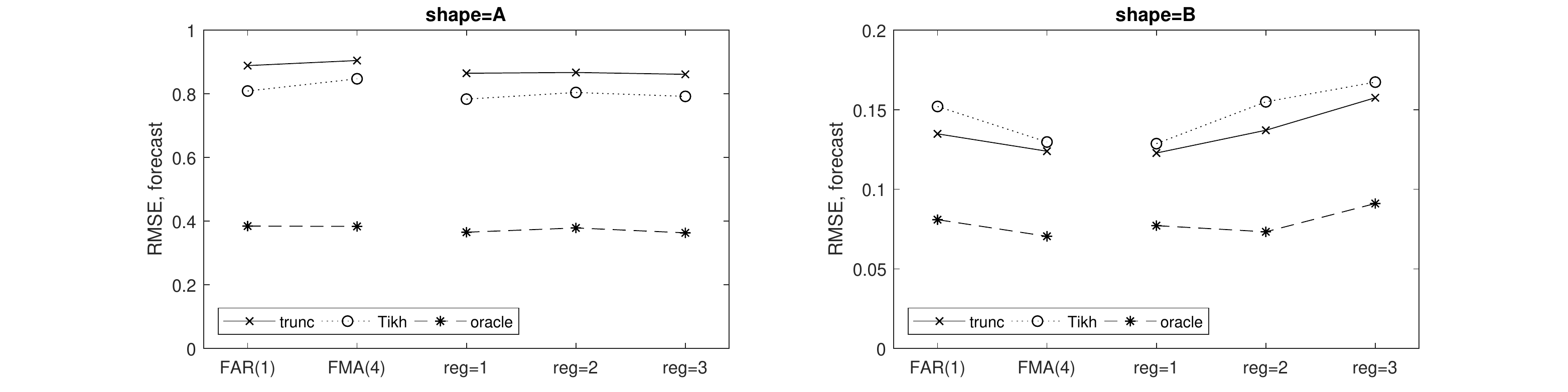}
\caption{
The median mean square prediction error $\delta^{pred}$ \eqref{eq:RMSE_delta_pred} for the truncation regularisation (``trunc''), Tikhonov regularisation method (``Tikh''), and the oracle estimator with respect to the simulated dynamics of $\{X_t\}_{t}$ and the regression scheme. The results are aggregated over all sparse observation setups $T\in\{300,600,900,1200\}$ and $N^{max}\in\{10,20,40,60\}$. The left and the right figures show the results for the filter coefficients of the shapes \eqref{eq:shape_A} and \eqref{eq:shape_B} respectively.
}
\label{fig:forecast_rmse_per_process}
\end{figure}

Due to the large number of simulation settings considered, we display the results in an aggregated form. Figures \ref{fig:filter_mse} and \ref{fig:forecast_rmse} present the results for the filter coefficient estimation and the prediction performance as a function of the sample size parameters $T$ and $N^{max}$. The results are aggregated over both types of simulated dynamics of $\{X_t\}_t$, i.e. the functional autoregressive process \FARno{ }and the functional moving average process \FMAno, and over all three regression schemes \Ba, \Bb, and \Bc. Figures \ref{fig:filter_mse_per_process} and \ref{fig:forecast_rmse_per_process}, on the other hand, present the results for different the simulated dynamics of the process and the different considered regression schemes separately, aggregated over all time series length parameters $T\in\{300,600,900,1200\}$ and all sparse observation regimes $N^{max}\in\{10,20,40,60\}$.

An inspection of figures \ref{fig:filter_mse} and \ref{fig:forecast_rmse} reveals that there is no clear winner between the truncation and the Tikhonov methods. The numerical experiments with the shape \textbf{(A)} defined by \eqref{eq:shape_A} show that the Tikhonov method dominates the truncation method in all considered settings for the estimation of the filter coefficients. The simulations with the shape \textbf{(B)} defined by \eqref{eq:shape_B} yield the opposite behaviour: the truncation regularisation prevails. We attribute this dichotomy to the following reasons:

\begin{itemize}

\item The shape \textbf{(B)} corresponds to the leading eigenfunction of the covariance kernel \eqref{eq:covariance_kernel_in_simulations}. Even though the functional regression is performed in the spectral domain, the spectral transfer function $\mathscr{B}_\omega$ is still well aligned with the first eigenfunction of $\mathscr{F}_\omega^X$ and therefore it is enough to cut off after the first eigenvalue,  thus favouring truncation regularisation. See Figure \ref{fig:percentage_explained} to visualise the alignment in the spectral domain.

\item The shape \textbf{(A)} corresponds to the fourth eigenfunction of the covariance kernel \eqref{eq:covariance_kernel_in_simulations}. Moreover the fourth eigenvalue is tied with the fifth and the sixth one. This structure is preserved also in the spectral domain, c.f. Figure \ref{fig:percentage_explained}. Since the Tikhonov regularisation does not discard the eigenspace corresponding to any of these eigenvalues, it achieves lower estimation error in this non-aligned case. Moreover, the Tikhonov regularisation enjoys the advantage of being stable to spectral eigenvalue ties \citep{hall2007methodology, pham2018methodology}.

\item Note that the shape \textbf{(A)} is generally more difficult to estimate than shape \textbf{(B)}. This is not surprising because the signal in the fourth eigenfunction is much weaker than in the first one.
\end{itemize}


\begin{figure}[h]
\centering
\makebox[\textwidth][c]{
\includegraphics[width=1\textwidth]{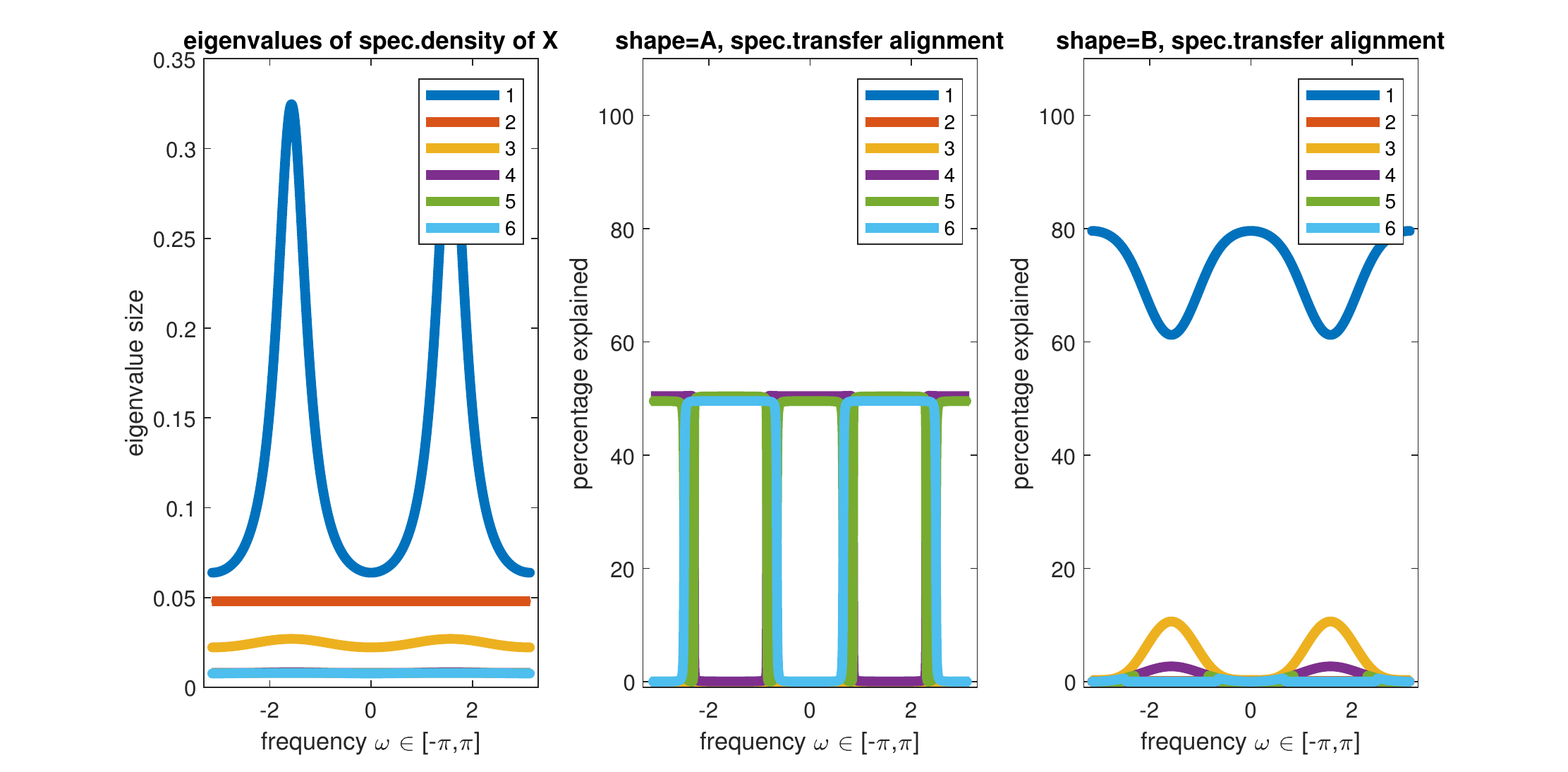}
}
\caption{
\textbf{Left:} The six leading eigenvalues $\lambda_1^\omega,\dots,\lambda_6^\omega$ of the spectral density operator $\mathscr{F}^X_\omega$ as a function of the frequency $\omega\in[-\pi,\pi]$.
\textbf{Center:} The norm percentage of the spectral transfer function $\mathscr{B}_\omega$ for the shape \textbf{(A)}, the simulated process \textbf{(FAR(1))}, and the regression scheme \textbf{(reg1)} explained by the eigenspaces corresponding to the six leading eigenvalues as a function of the frequency $\omega\in[-\pi,\pi]$.
\textbf{Right:} The same as the central plot but for the shape \textbf{(B)}.
}
\label{fig:percentage_explained}
\end{figure}

There is no notable difference in the difficulty of estimation between the two simulated dynamics of the regressor time series $\{X_t\}_t$, i.e. the functional autoregressive process \FARno{ }and the functional moving average process \FMAno. The considered regression schemes do not reveal any surprises: the longer the lagged dependence is, the more difficult estimation becomes. Therefore the \Bc{ } scheme produces the largest errors while \Ba{ }the lowest.

The prediction error of the response process $\delta^{pred}$, which is presented in Figures \ref{fig:forecast_rmse} and \ref{fig:forecast_rmse_per_process}, follows the same conclusions as the estimation of the filter coefficients. The shape \textbf{(A)} is more challenging to predict and the Tikhonov regularisation is seen to be preferable not only for estimation, but  for prediction too. The shape \textbf{(B)} is easier to predict using the truncation regularisation.
The predictions by either of the two techniques, feature twice to thrice greater prediction error $\delta^{pred}$ than the oracle estimator, i.e. the prediction assuming the model of the data to be known and the uncertainty coming only from sparse noisy sampling regime.

\section{Data analysis}
\label{sec:data analysis}

In this section we illustrate the proposed methodology on measurements recorded \citep{tammet2009joint} at the scientific observatory located at mount Wank located in southern Germany.
We remark that this analysis should be seen primarily as an example of the type of data that fall in our framework, rather than a complete data analysis, since there are presumably further important covariates that otherwise should be included.
The considered period is January 1, 1977 -- December 31, 1979 consisting of $T=1095$ days.
In particular, we analyse the interdependence of three time series:
\begin{itemize}

\item \textit{Atmospheric electricity.}
The ionisation processes in the atmosphere causes the air to be conductive and the conductivity can be measured in terms of electric potential difference per distance, expressed in volts per meter ($V/m$).
The atmospheric electricity is an important indicator for climate research \citep{tammet2009joint} and air pollution \citep{israelsson2001variation}.
However, the atmospheric electricity can be reliably measured only under fair-weather conditions, otherwise the atmospheric ionisation processes are changed and a different quantity is recorded. The standard meteorological methodologies \citep{xu2013periodic,israelsson2001variation} suggest to discard the observations under unfair conditions and analyse only those observations recorded under fair weather. Given these guidelines, we take into account only those hourly observations of atmospheric electricity where the wind speed was below $20\, km/h$ and the atmospheric electricity $E$ itself satisfies $0<E<250\, V/m$. The meteorological community also advocate discarding the data based on cloud coverage \citep{xu2013periodic,israelsson2001variation}, but, unfortunately, the mount Wank dataset \citep{tammet2009joint} does not contain cloud coverage information.

Based on the above criteria, we eventually retain an unevenly sampled scalar time series which we consequently decatenate into individual days. This technique is useful \citep{aue2015prediction,hormann2015dynamic,hormann2010weakly,hormann2018testing} in separating the intra-day variability and the temporal dependence across days. Thus we arrive at a sparsely observed functional time series, where the latent functional data are interpreted as the ``atmospheric electricity had the weather been fair''. 

In what follows, we denote the fair weather electricity time series as $\{X^{(E)}_t\}_{t=1}^T$. The time series features total of 18326 measurements or 16.7 measurement per day on average.

\item \textit{Temperature.} The temperature was recorded hourly at Wank over the considered period. First, we remove the yearly periodicity in the data, then divide the time domain into individual days, and finally convert the hourly observations into functional data using  B-splines. The produced fully observed functional time series is denoted as $\{X^{(\tau)}_t\}_{t=1}^T$. The time series includes 21 missing days which we treat as missing completely at random.

\item \textit{Recorded visibility.} The reported visibility was recorded hourly at a range of locations. We define the response time series $\{Z_t\}_{t=1}^T$ as the average visibility on the given day. The time series includes 42 missing values which we treat as missing completely at random.

\end{itemize}

Since the goal of our analysis is to illustrate the lagged regression methodology and compare the Tikhonov and the truncation regularisation, we split the response time series into two parts, the training component $Z_1,\dots,Z_{822}$ and the test component $Z_{823},\dots,Z_{1095}$ consisting of roughly $75 \%$ and $25 \%$ of the observations respectively. We fit three models on the entire time span of
$\{X^{(E)}_t\}_{t=1}^T$ and/or $\{X^{(\tau)}_t\}_{t=1}^T$ and the training set $\{Z_t\}_{t=1}^{822}$:
\begin{itemize}
\item \textit{Atmospheric electricity model} \textbf{(E)}. In this model we use the sparsely observed functional time series $\{X^{(E)}_t\}_{t=1}^T$ as the regressor time series for the response $\{Z_t\}_{t=1}^T$ exploiting the methodology outlined in this paper.

\item \textit{Temperature model} \textbf{(T)}. This model handles the fully observed functional time series $\{X^{(\tau)}_t\}_{t=1}^T$ as the regressor for the response $\{Z_t\}_{t=1}^T$ using the methodology developed by \citet{hormann2015estimation}.

\item \textit{Joint model} \textbf{(E+T)}. Finally this model includes the information from both the sparsely observed functional time series $\{X^{(E)}_t\}_{t=1}^T$ and the fully observed functional time series $\{X^{(\tau)}_t\}_{t=1}^T$ to predict the response process $\{Z_t\}_{t=1}^T$ outlined by the equation
\begin{equation}
Z_t = a
+ \sum_{k\in\mathbb{Z}} \mathcal{B}_k^{(E)} X_{t-k}^{(E)}
+ \sum_{k\in\mathbb{Z}} \mathcal{B}_k^{(\tau)} X_{t-k}^{(\tau)}
+ e_t
\end{equation}
which extends the model \eqref{eq:regression_model}.

The estimation and the prediction in a model that includes multiple functional time series, some under sparse and others under full observation, is explained in detail in Appendix~\ref{sec:appendix_muptiple_inputs}. Figure \ref{fig:lagged_regression} displays a schematic visualisation of the prediction for a specific day.
\end{itemize}

\begin{figure}[h]
\centering
\makebox[\textwidth][c]{
\includegraphics[width=1\textwidth]{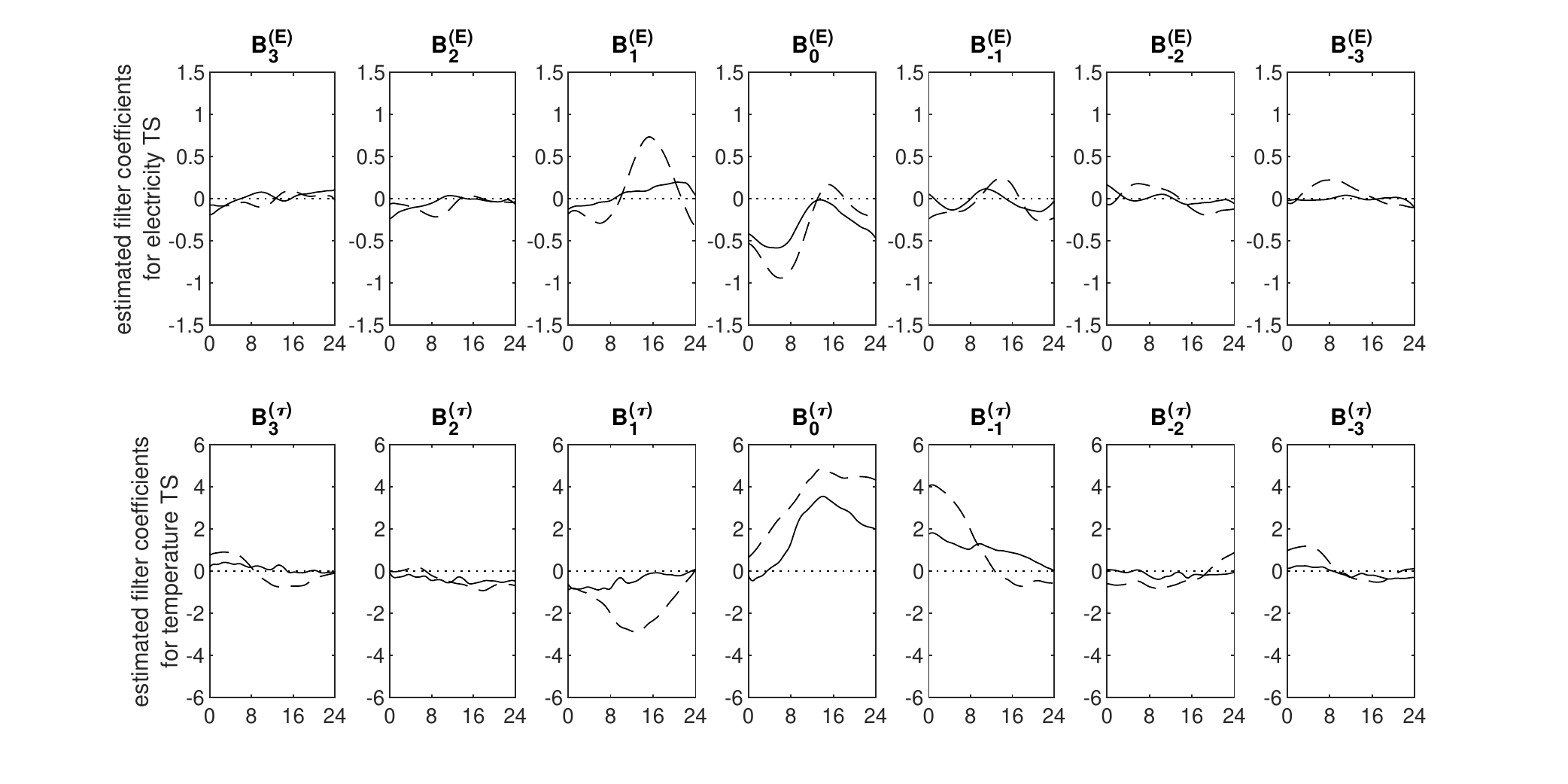}
}
\caption{
The estimated filter coefficients for lags $k\in\{-3,-2,-1,0,1,2,3\}$ for the joint model \textbf{(E+T)}. Solid line: the estimates by Tikhonov regularisation, dashed line: the estimates by truncation regularisation, dotted line: the reference line for zero. \textbf{Top row:} the filter coefficients $\mathcal{B}^{(E)}_k$ for the atmospheric electricity, \textbf{bottom row:} the filter coefficients $\mathcal{B}^{(\tau)}_k$ for the temperature time series
}
\label{fig:Wank_filters}
\end{figure}

\begin{figure}[hbtp]
\centering
\makebox[\textwidth][c]{
\includegraphics[width=0.86\textwidth]{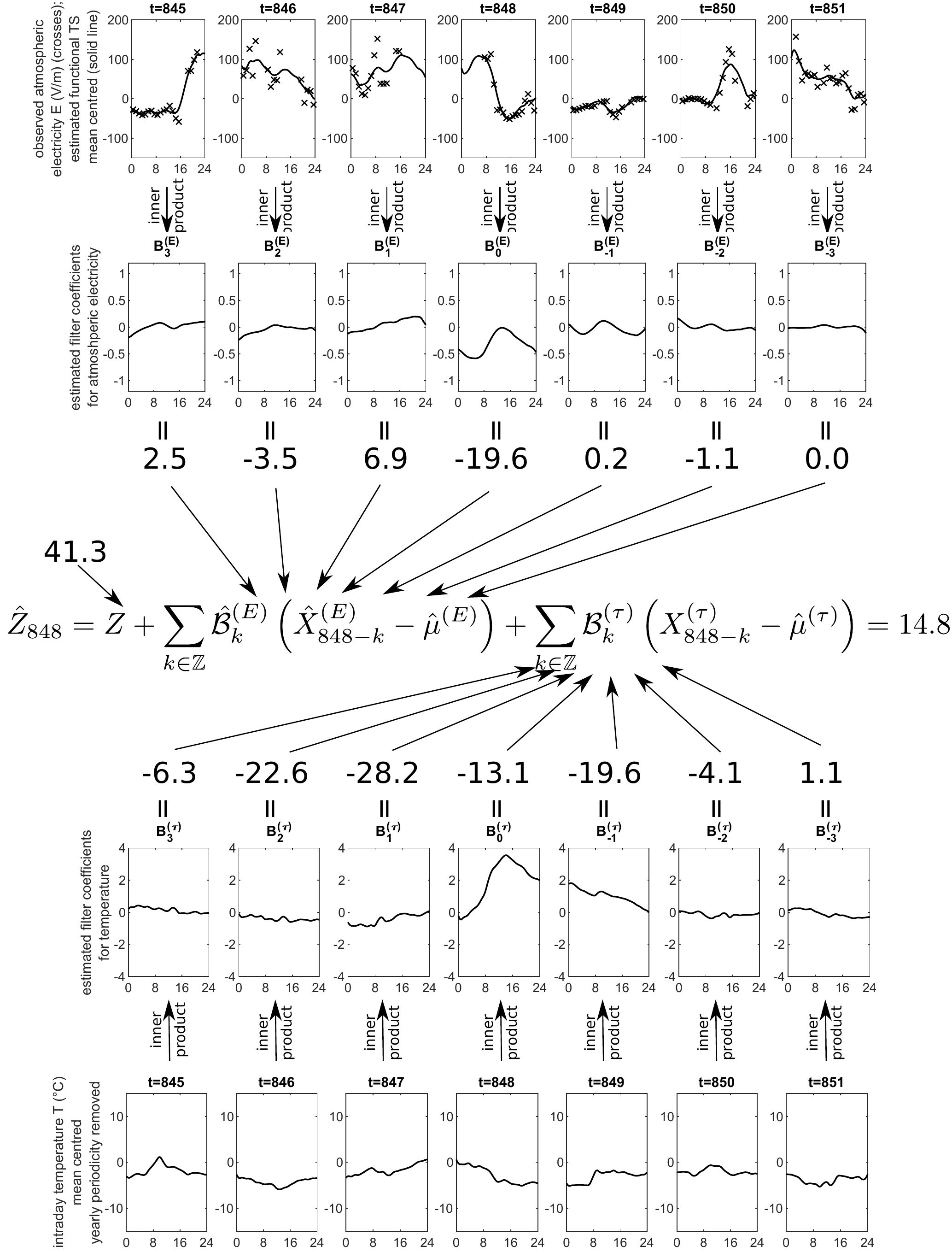}
}
\caption{
A schema demonstrating prediction of the response time series at time $t=848$ in the joint model \textbf{(E+T)}. Recall that because the response $\{Z_t\}$ is scalar, the filter coefficients $\mathcal{B}_k^{(E)}$ and $\mathcal{B}_k^{(\tau)}$ are functionals and thus can be viewed as inner products with fixed functions which are visualised here. \textbf{Top part of the figure:} the contribution of the atmospheric electricity, \textbf{bottom part:} the contribution of the temperature.
}
\label{fig:lagged_regression}
\end{figure}

We estimated the filter coefficients in all three models. Figure \ref{fig:Wank_filters} displays the estimated filter coefficients in the joint model \textbf{(E+T)}. The filter coefficients estimates in the marginal models \textbf{(E)} and \textbf{(T)} (not presented here) are very similar to the corresponding filters estimated in \textbf{(E+T)} thanks to the fact that the time series $\{X^{(E)}_t\}_{t=1}^T$ and $\{X^{(\tau)}_t\}_{t=1}^T$ are essentially uncorrelated.

A first look at Figure \ref{fig:Wank_filters} reveals that truncation-based estimates (depicted via their Riesz-representers) feature more spikes. This, combined with a worse prediction performance commented upon later, suggests that the Tikhonov regularisation provides a better fit. Therefore we comment only on the interpretation of the filter coefficients estimated by the Tikhonov regularisation.
The filter coefficient $\mathcal{B}_0^{(E)}$ is negative, especially in the morning hours, suggesting an obvious interpretation: high atmospheric electricity (which is linked to pollution \citep{israelsson2001variation}) implies a reduction of visibility.
The filter coefficients $\mathcal{B}_k^{(\tau)}$ corresponding to the temperature time series reveal an opposite effect. The filters $\mathcal{B}_0^{(\tau)}$ and $\mathcal{B}_{-1}^{(\tau)}$ are positive, therefore high temperatures on the same day and the next day predict a higher visibility today. We recall that our model is not causal, therefore the filter coefficient $\mathcal{B}_{-1}^{(\tau)}$ indeed predicts an effect backwards in time.

\begin{table}[h]
\caption{The mean square prediction error (MSE) and $R^2$ coefficients of determination of each of the model with either truncation or Tikhonov regularisation based estimation of the filter coefficients. Both the MSE's and the $R^2$ coefficients were determined on the test partition of the response $\{Z_t\}_{t=823}^{1095}$ }
\label{table:r2_coefficients_of_determination}
\centering
\begin{tabular}{ccccc}
\hline
Model &  \multicolumn{2}{c}{Truncation} & \multicolumn{2}{c}{Tikhonov} \\
& MSE & $R^2$ & MSE & $R^2$ \\
\hline
\textbf{(E)}   & 541 & 0.14 & 481 & 0.23 \\
\textbf{(T)}   & 407 & 0.35 & 379 & 0.39 \\
\textbf{(E+T)} & 392 & 0.37 & 335 & 0.46 \\
\hline
\end{tabular}
\end{table}

Table \ref{table:r2_coefficients_of_determination} presents the prediction performance of the considered models. We calculate the mean square error (MSE) on the test partition $\{Z_t\}_{t=823}^ {1095}$ and calculate the $R^2$ coefficient of determination.
The table reveals that the Tikhonov regularisation delivers a better prediction performance for all considered models.

\section{Discussion of the considered regularisation techniques}
\label{sec:discussion}

This article presents the methodology for the functional lagged regression problem where the regressor time series is observed sparsely and with noise contamination. We have shown how to estimate the (cross)-spectral density using surface smoothers. The estimation of the spectral transfer function and consequently the filter coefficients using the estimated (cross)-spectral density are ill-posed problems and therefore require regularisation.
We considered two regularisation strategies, namely spectral truncation and Tikhonov regularisation, and compared them on a simulation study and the analysis of a data set. In the following we summarise some observations  on the differences, strengths, and weaknesses of the two approaches.

The simulation study presented in Section \ref{sec:numerical experiments} illustrates that neither of the two regularisation method can dominate the other. In one of the considered scenarios the spectral transfer function $\mathscr{B}_\omega$ is well-aligned with the leading eigenfunction of the spectral density operator $\mathscr{F}_\omega^{X}$ thus being estimated better by means of truncation. In the other considered setting, $\mathscr{B}_\omega$ is explained by the fourth, the fifth and the sixth leading eigenvalue of $\mathscr{F}_\omega^{X}$, which are moreover nearly tied, resulting in a more challenging estimation task where the stability of Tikhonov regularisation to ties leads to better results than truncation.

The data analysis illustration given in Section \ref{sec:data analysis} analysed the dependence of  visibility on atmospheric electricity and temperature. The comparative analysis of the the two regularisation methods revealed that the estimates obtained by the Tikhonov regularisation feature better predictive performance. Moreover we found that the filter coefficients estimated by the Tikhonov regularisation were easier to interpret. At least in this application setting, the spectral transfer function $\mathscr{B}_\omega$ does not seem to be well-aligned with the spectral density operator $\mathscr{F}_\omega^X$.

We conclude that both of the regularisation techniques should belong to the statistician's repertoire as neither can dominate the other. However, if we were to choose only one to broadly recommend, this would be the Tikhonov approach, as it seems more robust to ``spectral misalignment" and eigenvalue ties, and its theoretical treatment requires fewer assumptions.

\FloatBarrier

\appendix

\section{Lagged regression with functional response}
\label{sec:appendix_functional_response}

\subsection{Functional response model and used notation}
In the main body of the article we commented on the fact that the main challenge of the methodology lies with the ill-posed inversion of the spectral density operator of the regressor time series $\{ \mathscr{F}^X_\omega \}_{\omega\in[-\pi,\pi]}$, and less so on whether the response is scalar or functional time series. Therefore we opted to present the model with a scalar response in the main body for of the article easier reading and postpone the treatment of the functional response to this appendix.

Now we consider the case where both the regressor time series $\{ X_t \} \equiv \{ X_t(\cdot) \}_{t\in\mathbb{Z}}$
and the response time series $\{ Z_t \} \equiv \{ Z_t(\cdot) \}_{t\in\mathbb{Z}}$ are zero-mean \emph{functional} time series. For simplicity of presentation we assume that they take values in the same Hilbert space $\mathcal{H} = \Ltwo$. The lagged regression model (without intercept) becomes
\begin{equation}\label{eq:S.model_functional_response}
Z_t = \sum_{k\in\mathbb{Z}} \mathcal{B}_k X_{t-k} + e_t 
\end{equation}
where the filter coefficients $\mathcal{B}_k$ are now assumed to be Hilbert-Schmidt operators from $\mathcal{H}$ to $\mathcal{H}$, and $\{e_t\}_{t\in\mathbb{Z}}$ is a sequence of independent identically distributed zero-mean random elements in $\mathcal{H}$ that is further independent of $\{X_t\}_{t\in\mathbb{Z}}$.
We denote by $B_k(\cdot,\cdot)$ the integral kernel corresponding to the operator $\mathcal{B}_k$ for $k\in\mathbb{Z}$.
The assumption \ref{assumption:A.3} is replaced by
\begin{enumerate}
\Item[S.(A3)'] \label{assumption:S.A.3}
$$
\sum_{k\in\mathbb{Z}}\| B_k \|_{\infty} =
\sum_{k\in\mathbb{Z}} \sup_{x,y\in[0,1]} | B_k(x,y) |
< \infty, \qquad
\sum_{k\in\mathbb{Z}}\| \mathcal{B}_k \|_2 < \infty,
$$
\end{enumerate}
where $\|\cdot\|_2$ is the Hilbert-Schmidt norm.

Both the regressor time series $\{X_t\}$ and the response time series $\{Z_t\}$ are assumed to be observed sparsely
\begin{align*}
Y_{tj} &= X_t(x_{tj}) + \epsilon_{tj}^X, \qquad j = 1,\dots,N_t^X, \quad t=1,\dots,T, \\
V_{tj} &= Z_t(z_{tj}) + \epsilon_{tj}^Z, \qquad j = 1,\dots,N_t^Z, \quad t=1,\dots,T,
\end{align*}
where $x_{tj}$ and $z_{tj}$ are the measurement locations, $N_t^X$ and $N_t^Z$ are the number of measurement locations at time $t$, and the $\epsilon_{tj}^X$ and $\epsilon_{tj}^Z$ are additive measurement error contaminants. The $\{\epsilon_{tj}^X\}_{tj}$ and $\{\epsilon_{tj}^Z\}_{tj}$ are assumed to be two ensembles of independent identically distributed zero-mean real random variables with variances $\sigma^2_X>0$ and $\sigma^2_Z>0$ respectively. Moreover, the two ensembles are independent of each other and the underlying pair of functional time series $\{X_t\}$ and $\{Z_t\}$.

\subsection{Nonparametric estimation of the cross-spectral density and spectral transfer function}

In this section we describe how to estimate the cross-spectral density operator $\{ \mathscr{F}_\omega^{ZX} \}_{\omega\in[-\pi,\pi]}$.
For $h\in\mathbb{Z}, |h|<T,
t=\max(1,1+h),\dots,\min(T,T-h),
\quad j=1,\dots,N^Z_{t+h},
\quad k=1,\dots,N^X_{t},$
define the raw covariances as
$$ G^{ZX}_{h,t}(z_{t+h,j}, x_{tk})
= V_{t+h,j} Y_{tk}.
$$
The cross-spectral density is estimated by smoothing the raw covariances using a local-linear surface smoother. For fixed frequency $\omega\in[-\pi,\pi]$, Bartlett span parameter $L\in\mathbb{N}$, and smoothing bandwidth $B_R>0$ we set
$$ \hat{f}_\omega^{ZX}(x,y) = \frac{L}{2\pi} \bar{d}_0 $$
where
\begin{multline*}
(\bar{d}_0,\bar{d}_1,\bar{d}_2) = \argmin_{(d_0,d_1,d_2)\in\mathbb{C}^3}
\frac{1}{T}
\sum_{-L}^L
\sum_{t=\max(1,1-h)}^{\min(T,T-h)}
\sum_{j=1}^{N_{t+h}^Z}
\sum_{k=1}^{N_t^X}
\Big|
G^{ZX}_{h,t}(z_{t+h,j},x_{tj'})e^{-\I h \omega} -\\
-d_0 - d_1(z_{t+h,j}) - d_2(x_{tk})
\Big|
W_h \frac{1}{B_R^2}
K\left( \frac{z_{t+h,j} - x}{B_R} \right)
K\left( \frac{x_{tk} - y}{B_R} \right)
\end{multline*}

Once the cross-spectral density $\{\mathscr{F}_\omega^{ZX}\}_{\omega\in[-\pi,\pi]}$ and the spectral density $\{\mathscr{F}_\omega^{ZX}\}_{\omega\in[-\pi,\pi]}$ have been estimated we turn our attention to the estimation of the spectral transfer function.

The spectral function is estimated by the formulae 
\eqref{eq:regularisation_spectral_transfer_truncation} and
\eqref{eq:regularisation_spectral_transfer_Tikhonov} where $\hat{\mathscr{F}}_\omega^{ZX}\left( \hat{\varphi}_j^\omega \right)$ is now the application of the operator to the function $\hat{\varphi}_j^\omega$, yielding a function as output. Thus $\hat{\mathscr{B}}_\omega^{trunc}$ and $\hat{\mathscr{B}}_\omega^{Tikh}$ are operators.
The filter coefficients, now operators, are again recovered by the formulae
\eqref{eq:fourier_truncation_filter} and
\eqref{eq:fourier_tikhonov_filter}.

\subsection{Forecasting the response process}

We extend the setting outlined in Section \ref{subsec:forecasting_the_response_process} in the main body of the article. We assume that the response functional time series $\{Z_t\}$ is sparsely observed only up to point $S$ for $S<T$ and we wish to forecast the values $Z_{S+1},\dots,Z_T$.

The forecasting algorithm follows in fact the same steps as presented in Section \ref{subsec:forecasting_the_response_process}.
The only difference is that the formulae in step \ref{item:algorithm_step_4} for the forecast construction
$ \hat{\Pi}(Z_s \vert \mathbb{Y}) = \sum_{k = -M}^{M} \hat{\mathcal{B}}_k^{trunc} \hat{\Pi}({X}_{s-k}\vert \mathbb{Y}) $ (or alternatively
$ \hat{\Pi}(Z_s \vert \mathbb{Y}) = \sum_{k = -M}^{M} \hat{\mathcal{B}}_k^{Tikh} \hat{\Pi}({X}_{s-k}\vert\mathbb{Y} ) $)
output elements of the function space $\mathcal{H}$ because $\hat{\mathcal{B}}_k^{trunc}$ and $\hat{\mathcal{B}}_k^{Tikh}$ are now operators.


\section{Lagged regression with multiple inputs}
\label{sec:appendix_muptiple_inputs}

\subsection{Joint regression model and its spectral analysis}
It may very well happen that we want to analyse the relationship between multiple time series, including functional time series observed either sparsely or fully, multivariate time series, and scalar time series. To keep the notation simple we focus our presentation on the setting of a scalar response and two functional time series regressors, one of which is observed sparsely and the other fully. This setting corresponds to the model \textbf{(E+T)} of the data analysis in Section~\ref{sec:data analysis}. The generalisation to a higher number of functional time series, observed sparsely or fully, is clear, treating multivariate time series as time series in a finite-dimensional Hilbert space; and, the further incorporation of multivariate or scalar time series is uncomplicated because the regularisation of the spectral density inversion is no longer required.

Consider the model with two functional time series regressors $\{X_{t}^{(1)}\} \equiv \{X_{t}^{(1)}(\cdot)\}_{t\in\mathbb{Z}}$ and $\{X_{t}^{(2)}\}\equiv\{X_{t}^{(2)}(\cdot)\}_{t\in\mathbb{Z}}$. To simplify the notation we assume the both functional time series to take values in the same Hilbert space $\mathcal{H}$. The lagged regression with a scalar response $\{Z_t\} \equiv \{Z_t\}_{t\in\mathbb{Z}}$ becomes:
$$
Z_t = a
+ \sum_{k\in\mathbb{Z}} \mathcal{B}_k^{(1)} X_{t-k}^{(1)}
+ \sum_{k\in\mathbb{Z}} \mathcal{B}_k^{(2)} X_{t-k}^{(2)}
+ e_t
$$
where $a\in\mathbb{R}$ is the intercept, $\{\mathcal{B}_k^{(1)}\}_{k\in\mathbb{Z}}$ and $\{\mathcal{B}_k^{(1)}\}_{k\in\mathbb{Z}}$ are two sequences of linear mappings from $\mathcal{H}$ to $\mathbb{R}$, and $\{e_t\}$ is a sequence of zero mean independent identically distributed real random variables.
Denote $\{b_k^{(1)}\}_{k\in\mathbb{Z}}$ and $\{b_k^{(1)}\}_{k\in\mathbb{Z}}$ the filter functions corresponding to the functionals  $\{\mathcal{B}_k^{(1)}\}$ and $\{\mathcal{B}_k^{(1)}\}$ by the Riesz representation theorem.
We assume that the first functional time series, $\{X_{t-k}^{(1)}\}$, is sparsely observed, that is we have access to only the observations generated by
$$ Y^{(1)}_{tj} = X^{(1)}_t(x_{tj}) + \epsilon_{tj}, \qquad j = 1,\dots,N_t, \quad t=1,\dots,T, $$
where $N_t$ is a number of observation locations $\{x_{tj}\}$ at time $t=1,\dots,T$.
The second functional time series is fully observed, therefore we have access to $X_1^{(2)},\dots,X_T^{(2)}$. The mean functions of $\{X_{t}^{(1)}\}$ and $\{X_{t}^{(2)}\}$ are denoted as $\mu^{(1)}(\cdot)$ and $\mu^{(2)}(\cdot)$ respectively.

Denote $R_h^{\ast}$ and $\mathscr{R}_h^{\ast}$ the (auto)covariance kernel and (auto)covariance operator where $\ast$ is substituted by a single functional time series or a pair theoreof, as the case may be.
Likewise we denote the (cross)-spectral density kernel and the (cross)-spectral density operator as $f_\omega^{\ast}$ and $\mathscr{F}_\omega^{\ast}$.

Assuming
$$
\sum_{h\in\mathbb{Z}} \left\|
\begin{bmatrix}
R_h^{X^{(1)}} & R_h^{X^{(1)}X^{(2)}} \\
R_h^{X^{(2)}X^{(1)}} & R_h^{X^{(2)}}
\end{bmatrix}
\right\|_\infty  < \infty, \qquad
\sum_{h\in\mathbb{Z}} \left\|
\begin{bmatrix}
\mathscr{R}_h^{X^{(1)}} & \mathscr{R}_h^{X^{(1)}X^{(2)}} \\
\mathscr{R}_h^{X^{(2)}X^{(1)}} & \mathscr{R}_h^{X^{(2)}}
\end{bmatrix}
\right\|_1  < \infty$$
and $ \sum_{k\in\mathbb{Z}} \|b_k^{(1)}\|_\infty < \infty $,
$ \sum_{k\in\mathbb{Z}} \|b_k^{(2)}\|_\infty < \infty $,
$ \sum_{k\in\mathbb{Z}} \|\mathcal{B}_k^{(1)}\|_{\mathcal{H}} < \infty $,
$ \sum_{k\in\mathbb{Z}} \|\mathcal{B}_k^{(2)}\|_{\mathcal{H}} < \infty $
implies that the (cross)-spectral density kernels and operators
$\{ f_\omega^{X^{(1)}} \}_\omega$,
$\{ f_\omega^{X^{(2)}} \}_\omega$,
$\{ f_\omega^{X^{(1)}X^{(2)}} \}_\omega$,
$\{ f_\omega^{ZX^{(1)}} \}_\omega$,
$\{ f_\omega^{ZX^{(2)}} \}_\omega$,
$\{ \mathscr{F}_\omega^{X^{(1)}} \}_\omega$,
$\{ \mathscr{F}_\omega^{X^{(2)}} \}_\omega$,
$\{ \mathscr{F}_\omega^{X^{(1)}X^{(2)}} \}_\omega$,
$\{ \mathscr{F}_\omega^{ZX^{(1)}} \}_\omega$,
$\{ \mathscr{F}_\omega^{ZX^{(2)}} \}_\omega$ are well defined.

The joint frequency response operator
$$ \mathscr{B}_\omega =
\begin{bmatrix}
\mathscr{B}_\omega^{(1)} &
\mathscr{B}_\omega^{(2)}
\end{bmatrix}
=
\sum_{k\in\mathbb{Z}}
\begin{bmatrix}
\mathcal{B}_k^{(1)} &
\mathcal{B}_k^{(2)}
\end{bmatrix}
e^{-\I k\omega}
$$
satisfies the relation
\begin{equation}
\label{eq:joint_frequency_response_operator}
\begin{bmatrix}
\mathscr{F}_\omega^{ZX^{(1)}} &
\mathscr{F}_\omega^{ZX^{(2)}}
\end{bmatrix}
=
\begin{bmatrix}
\mathscr{B}_\omega^{(1)} &
\mathscr{B}_\omega^{(2)}
\end{bmatrix}
\begin{bmatrix}
\mathscr{F}_\omega^{X^{(1)}} & \mathscr{F}_\omega^{X^{(1)}X^{(2)}} \\
\mathscr{F}_\omega^{X^{(2)}X^{(1)}} & \mathscr{F}_\omega^{X^{(2)}} \\
\end{bmatrix}
.
\end{equation}

The filter coefficients can be recovered by the following formula, for $k\in\mathbb{Z}$,
\begin{align*}
\mathcal{B}_k^{(1)} &= \frac{1}{2\pi} \int_{-\pi}^\pi \mathscr{B}_\omega^{(1)} e^{\I k\omega} \D \omega,\\
\mathcal{B}_k^{(2)} &= \frac{1}{2\pi} \int_{-\pi}^\pi \mathscr{B}_\omega^{(2)} e^{\I k\omega} \D \omega.
\end{align*}

\subsection{Nonparametric estimation}
The mean function of the sparsely observed $\{X_{t}^{(1)}\}$ can be estimated by the local linear smoother \citep{rubin2020sparsely}. The mean function of the fully observed $\{X_{t}^{(2)}\}$ is estimated by the classical empirical pointwise mean. We denote these estimates $\hat{\mu}^{(1)}(\cdot)$ and $\hat{\mu}^{(2)}(\cdot)$ respectively.

The estimation of the spectral density $\{ \mathscr{F}_\omega^{X^{(1)}} \}_{\omega\in[-\pi,\pi]}$ and the cross-spectral density $\{ \mathscr{F}_\omega^{Z X^{(1)}} \}_{\omega\in[-\pi,\pi]}$ from the sparse observations is explained in the main body of the article. For the fully observed functional case we use functional version of Bartlett's estimator \citep{hormann2015dynamic,hormann2015estimation}
\begin{align*}
\mathscr{F}_\omega^{X^{(2)}} &=
\frac{1}{2\pi}
\sum_{|h|<L}
\left( 1 - \frac{|h|}{L} \right) \hat{\mathscr{R}}_h^{X^{(2)}}
e^{-\I h \omega}, \qquad \omega\in[-\pi,\pi], \\
\mathscr{F}_\omega^{ZX^{(2)}} &=
\frac{1}{2\pi}
\sum_{|h|<L}
\left( 1 - \frac{|h|}{L} \right) \hat{\mathscr{R}}_h^{ZX^{(2)}}
e^{-\I h \omega}, \qquad \omega\in[-\pi,\pi],
\end{align*}
where $\hat{\mathscr{R}}_h^{X^{(2)}}$ and $\hat{\mathscr{R}}_h^{ZX^{(2)}}$ are the classical empirical lag-$h$ autocovariance operators of $\{X^{(2)}_t\}$ and the lag-$h$ covariance operator of $\{Z_t\}$ and $\{X^{(2)}_t\}$ respectively.

It remains to estimate the cross-spectral density $\{ \mathscr{F}_\omega^{X^{(1)} X^{(2)}} \}_{\omega\in[-\pi,\pi]}$.
Aiming to estimate $f_\omega^{X^{(1)} X^{(2)}}(x,y)$ for $x,y\in[0,1],\omega\in[-\pi,\pi]$, we define the raw covariances
$$ G_{h,t}^{X^{(1)}X^{(2)}}(x_{t+h,j},y) =
\left( Y_{t+h,j} - \hat{\mu}^{(1)}( x_{t+h,j} ) \right)
\left( X_t^{(2)}(y) - \hat{\mu}^{(2)}(y) \right)
,$$
and for fixed $(x,y)\in[0,1]^2$ we estimate the cross-spectral density as
$$ \hat{f}_\omega^{ X^{(1)} X^{(2)} }(x,y) = \frac{L}{2\pi} \hat{d}_0 $$
where $\hat{d}_0\in\mathbb{C}$ is the minimiser of the following weighted sum of squares
\begin{multline*}
( \hat{d}_0,\hat{d}_1 )
= \argmin_{(d_0,d_1)\in\mathbb{C}^2}
\sum_{h=-L}^L
\sum_{t=\max(1,1-h)}^{\min(T,T-h)}
\sum_{j=1}^{N_{t+h}}
\left|
G^{X^{(1)}X^{(2)}}_{h,t}(x_{t+h,j},y) e^{-\I h\omega}
- d_0 - d_1 \left( x_{t+h,j}-x \right)
\right|^2
\times\\\times
W_h \frac{1}{B_R} K\left( \frac{x_{t+h,j}-x}{B_R} \right).
\end{multline*}

Once we have estimated all the above (cross-)spectral densities, we wish to recover the filter coefficients from the equation \eqref{eq:joint_frequency_response_operator}. The inversion of the joint spectral density of $\{X^{(1)}_t\}$ and $\{X^{(1)}_t\}$ is ill-conditioned due to the reasons explained in Section \ref{subsec:nonparam_estimation} of the main body of the paper where we overcame this issue by two regularisation techniques: spectral truncation and Tikhonov regularisation. Here we do the same while allowing each of the regressor time series to have a different degree of regularisation. This is important because, generally speaking, the estimation from sparse data will require more regularisation.

Denote the spectral decompositions of the (cross)-spectral density operators by
\begin{align*}
\mathscr{\hat{F}}_\omega^{X^{(1)}} &= \sum_{i=1}^\infty \hat{\lambda}_i^\omega \hat{\varphi}_i^\omega \otimes \hat{\varphi}_i^\omega, \qquad\omega\in[-\pi,\pi], \\
\mathscr{\hat{F}}_\omega^{X^{(2)}} &= \sum_{i=1}^\infty \hat{\eta}_i^\omega \hat{\psi}_i^\omega \otimes \hat{\psi}_i^\omega, \qquad\omega\in[-\pi,\pi], \\
\mathscr{\hat{F}}_\omega^{X^{(1)}X^{(2)}} &= \sum_{i=1}^\infty\sum_{j=1}^\infty \hat{\gamma}_{ij}^\omega \hat{\psi}_j^\omega \otimes \hat{\varphi}_i^\omega, \qquad\omega\in[-\pi,\pi], \\
\end{align*}
where $\{\hat\lambda_i^\omega\}_{i=1}^\infty$ and $\{\hat\eta_i^\omega\}_{i}^\infty$ are the harmonic eigenvalues of
$\mathscr{F}_\omega^{X^{(1)}}$ and $\mathscr{F}_\omega^{X^{(2)}}$ respectively, for given $\omega\in[-\pi,\pi]$. The sequences of functions $\{\hat\varphi^\omega\}_{i}^\infty$ and $\{\hat\psi^\omega\}_{i}^\infty$ are orthonormal bases of $\mathcal{H}$ and are called the harmonic eigenfunctions of $\mathscr{F}_\omega^{X^{(1)}}$ and $\mathscr{F}_\omega^{X^{(2)}}$ respectively. The complex numbers $\{\hat{\gamma}_{ij}^\omega\}_{i,j=1}^\infty$ are the basis coefficients of $\mathscr{\hat{F}}_\omega^{X^{(1)}X^{(2)}}$ with respect to $\{ \hat{\varphi}_i^\omega \otimes \hat{\psi}_j^\omega \}_{i,j=1}^\infty$, a basis of the space of Hilbert-Schmidt operators on $\mathcal{H}$.

\begin{enumerate}
\item \textit{Truncation regularisation.}
Select two (possibly different) truncation parameters $\Komega1\in\mathbb{N}$ and $\Komega2\in\mathbb{N}$ that may depend on $\omega\in[-\pi,\pi]$. The idea of the truncation regularisation is to replace the empirical joint spectral density operator on the right hand side of the empirical version of \eqref{eq:joint_frequency_response_operator} by
\begin{equation}\label{eq:joint_response_truncation_eq1}
\begin{bmatrix}
\sum_{i=1}^{\Komega1} \hat\lambda_i^\omega \hat\varphi_i^\omega \otimes \hat\varphi_i^\omega &
\sum_{i=1}^{\Komega1}\sum_{j=1}^{\Komega2} \hat{\gamma}_{ij}^\omega \hat{\psi}_j^\omega \otimes \hat{\varphi}_i^\omega \\
\sum_{i=1}^{\Komega2}\sum_{j=1}^{\Komega1} \hat{\gamma}_{ji}^\omega \hat{\varphi}_j^\omega \otimes \hat{\psi}_i^\omega &
\sum_{i=1}^{\Komega2} \hat{\eta}_i^\omega \hat{\psi}_i^\omega \otimes \hat{\psi}_i^\omega
\end{bmatrix}.
\end{equation}
Alternatively, the operator in the form \eqref{eq:joint_response_truncation_eq1} can be expressed with respect to the reduced basis of $\mathcal{H}\times\mathcal{H}$ composed of
$\{ [\hat\varphi_1^\omega, \textbf{0}],\dots,[\hat\varphi_{\Komega1}^\omega, \textbf{0}], [\textbf{0}, \hat\psi_1^\omega],\dots,[\textbf{0}, \hat\psi_{\Komega2}^\omega] \}$ where $\textbf{0}$ is the zero element of $\mathcal{H}$.
Denote $M \in\mathbb{C}^{ ({\Komega1+\Komega2})\times({\Komega1+\Komega2}) }$ the complex matrix given by the inverse of the following matrix, assuming it is invertible,
$$
M=
\left(
\begin{array}{ccc|ccc}
\hat\lambda_1^\omega & 0 & 0           & \hat\gamma_{11} & \cdots & \hat\gamma_{1,\Komega2}\\
0 & \ddots & 0                     & \vdots & \ddots & \vdots \\
0 & 0 & \hat\lambda_{\Komega1}^\omega & \hat\gamma_{\Komega1,1} & \cdots & \hat\gamma_{\Komega1,\Komega2} \\
\hline
\hat\gamma_{11} & \cdots & \hat\gamma_{\Komega1,1} 							& \hat\eta_1^\omega & 0 & 0\\
\vdots & \ddots & \vdots 											& 0 & \ddots & 0 \\
\hat\gamma_{1,\Komega2} & \cdots & \hat\gamma_{\Komega1,\Komega2} & 0 & 0 & \hat\eta_{\Komega2}^\omega
\end{array}
\right)^{-1}.
$$
Denote the elements of $M$ according to the blocks using the following scheme
$$
M =
\left(
\begin{array}{ccc|ccc}
m^{(1)}_{11} & \cdots & m^{(1)}_{ 1,\Komega1} 					& m^{(12)}_{11} & \cdots & m^{(12)}_{1,\Komega2} \\
\vdots & \ddots & \vdots 										& \vdots & \ddots & \vdots \\
m^{(1)}_{\Komega1,1} & \cdots & m^{(1)}_{ \Komega1,\Komega1} 	& m^{(12)}_{\Komega1,1} & \cdots & m^{(12)}_{\Komega1,\Komega2} \\
\hline
m^{(21)}_{11} & \cdots & m^{(21)}_{ 1,\Komega1} 				& m^{(2)}_{11} & \cdots & m^{(2)}_{1,\Komega2} \\
\vdots & \ddots & \vdots 										& \vdots & \ddots & \vdots \\
m^{(21)}_{\Komega2,1} & \cdots & m^{(21)}_{ \Komega2,\Komega1} 	& m^{(2)}_{\Komega2,1} & \cdots & m^{(2)}_{\Komega2,\Komega2}
\end{array}
\right).
$$

Then, the frequency response operators can be recovered by
\begin{align*}
\hat{\mathscr{B}}_\omega^{(1)} &=
\sum_{i=1}^{\Komega1}
\sum_{j=1}^{\Komega1}
m_{ij}^{(1)}
\left\langle
\hat\varphi_j^\omega,\cdot
\right\rangle
\mathscr{F}_\omega^{ZX^{(1)}}
\left(
\hat\varphi_i^\omega
\right)
+
\sum_{i=1}^{\Komega2}
\sum_{j=1}^{\Komega1}
m_{ij}^{(21)}
\left\langle
\hat\varphi_j^\omega,\cdot
\right\rangle
\mathscr{F}_\omega^{ZX^{(2)}}
\left(
\hat\psi_i^\omega
\right), 
\\ 
\hat{\mathscr{B}}_\omega^{(2)} &=
\sum_{i=1}^{\Komega1}
\sum_{j=1}^{\Komega2}
m_{ij}^{(12)}
\left\langle
\hat\psi_j^\omega,\cdot
\right\rangle
\mathscr{F}_\omega^{ZX^{(1)}}
\left(
\hat\varphi_i^\omega
\right)
+
\sum_{i=1}^{\Komega2}
\sum_{j=1}^{\Komega2}
m_{ij}^{(2)}
\left\langle
\hat\psi_j^\omega,\cdot
\right\rangle
\mathscr{F}_\omega^{ZX^{(2)}}
\left(
\hat\psi_i^\omega
\right). 
\end{align*}

\item \textit{Tikhonov regularisation.} We consider again two possibly different regularisation parameters $\rho^{(1)}>0$ and $\rho^{(2)}>0$.

The empirical joint spectral density operator on the right hand side of the empirical version of \eqref{eq:joint_frequency_response_operator} is replaced by
$$
\begin{bmatrix}
\hat{\mathscr{F}}_\omega^{X^{(1)}} + \rho^{(1)} \mathcal{I} & \hat{\mathscr{F}}_\omega^{X^{(1)}X^{(2)}} \\
\hat{\mathscr{F}}_\omega^{X^{(2)}X^{(1)}} & \hat{\mathscr{F}}_\omega^{X^{(2)}} + \rho^{(2)} \mathcal{I}
\end{bmatrix}
$$
and the joint frequency response operator is estimated as
$$
\begin{bmatrix}
\hat{\mathscr{B}}_\omega^{(1)} &
\hat{\mathscr{B}}_\omega^{(2)}
\end{bmatrix}
=
\begin{bmatrix}
\hat{\mathscr{F}}_\omega^{ZX^{(1)}} &
\hat{\mathscr{F}}_\omega^{ZX^{(2)}}
\end{bmatrix}
\begin{bmatrix}
\hat{\mathscr{F}}_\omega^{X^{(1)}} + \rho^{(1)} \mathcal{I} & \hat{\mathscr{F}}_\omega^{X^{(1)}X^{(2)}} \\
\hat{\mathscr{F}}_\omega^{X^{(2)}X^{(1)}} & \hat{\mathscr{F}}_\omega^{X^{(2)}} + \rho^{(2)} \mathcal{I}
\end{bmatrix}^{-1}.
$$

\end{enumerate}

Once the estimates of the frequency response operators have been constructed, the filter coefficients are estimated by integrating back into the temporal domain
\begin{align*}
\hat{\mathcal{B}}_k^{(1)} &= \frac{1}{2\pi} \int_{-\pi}^\pi \hat{\mathscr{B}}_\omega^{(1)} e^{\I\omega k}, \qquad k\in\mathbb{Z}, \\
\hat{\mathcal{B}}_k^{(2)} &= \frac{1}{2\pi} \int_{-\pi}^\pi \hat{\mathscr{B}}_\omega^{(2)} e^{\I\omega k}, \qquad k\in\mathbb{Z}.
\end{align*}

\subsection{Forecasting the response process}

In this section we extend the forecasting algorithm introduced in Section \ref{subsec:forecasting_the_response_process} of the main body of the article.
As in that Section, we assume here that the response time series is observed only at times $1,\dots,S$ where $S<T$ and we wish to predict $Z_{S+1},\dots,Z_T$.

\begin{enumerate}

\item From the measurements $Y_{tj}, t=1,\dots,T, j=1,\dots,N_t$, of sparsely observed $\{ X_t^{(1)} \}$, the fully observed functional observations  $X_1^{(1)},\dots,X_T^{(1)}$, and the response time series $Z_1,\dots,Z_S$, construct the estimates of the model:
the mean functions
$\hat{\mu}^{(1)}, \hat{\mu}^{(2)}$,
the spectral densities
$
\{ \hat{\mathscr{F}}^{X^{(1)}}_\omega \}_{\omega\in[-\pi,\pi]},
\{ \hat{\mathscr{F}}^{X^{(2)}}_\omega \}_{\omega\in[-\pi,\pi]}
$
and the cross-spectral densities
$
\{ \hat{\mathscr{F}}^{X^{(1)}X^{(2)}}_\omega \}_{\omega\in[-\pi,\pi]},
\{ \hat{\mathscr{F}}^{ZX^{(1)}}_\omega \}_{\omega\in[-\pi,\pi]},
\{ \hat{\mathscr{F}}^{ZX^{(2)}}_\omega \}_{\omega\in[-\pi,\pi]}.
$
Using either the truncation regularisation or the Tikhonov regularisation, estimate the spectral transfer functions $\{ \hat{\mathscr{B}}_\omega^{(1)} \}_{\omega\in[-\pi,\pi]}$ and $\{ \hat{\mathscr{B}}_\omega^{(2)} \}_{\omega\in[-\pi,\pi]}$ and the filter coefficients
$\{ \hat{\mathcal{B}}_k^{(1)} \}_{k\in\mathbb{Z}}$ and
$\{ \hat{\mathcal{B}}_k^{(2)} \}_{k\in\mathbb{Z}}$.

\item Denoting the ensemble of the measurements $Y_{tj}, t=1,\dots,T, j=1,\dots,N_t$ as $\mathbb{Y}^{(1)}$, predict the latent values of $X_{-M+1}^{(1)},\dots,X_{T+M}^{(1)}$ using the method recalled in Section \ref{subsec:forecasting_the_response_process} of the main body of the article. Denote these predictions as $\hat{\Pi}( X_{-M+1}^{(1)} \vert \mathbb{Y}^{(1)} ), \dots, \hat{\Pi}( X_{T+M}^{(1)} \vert \mathbb{Y}^{(1)} )$. The constant $M$ is determined in such a way that the filter coefficients $\hat{\mathcal{B}}_k^{(1)}$ and $\hat{\mathcal{B}}_k^{(2)}$ are negligible for $|k|>M$.

\item Pad the fully observed functional time series $\{ X_t^{(2)} \}$ by the mean function $\hat{\mu}^{(2)}$, i.e. set $X_t = \hat{\mu}^{(2)}$ for $t<1$ or $t>T$. This approach was suggested in functional lagged prediction for fully functional regressor \citep{hormann2015estimation} where the zero-mean functional regressor time series is padded by zeros.

\item \label{item:S.algorithm_step_4}
For each $s=S+1,\dots,T$, construct the forecast 
$$ \hat{Z}_s = \bar{Z} +
\sum_{k=-M}^M \hat{\mathcal{B}}_k^{(1)}
\left( \hat{\Pi}( X_{s-k} \vert \mathbb{Y}^{(1)}) - \hat{\mu}^{(1)} \right) +
\sum_{k=-M}^M \hat{\mathcal{B}}_k^{(2)} \left( X_{s-k} - \hat{\mu}^{(2)}\right)
$$
where $\bar{Z}$ is the sample mean of the response process.

\end{enumerate}


\addtocontents{toc}{\protect\setcounter{tocdepth}{1}} 
\section{Proofs of formal statements}
\label{sec:appendix_proofs}

\subsection{Proof of Proposition \ref{prop:equivalence_of_BLUP}}
\begin{proof}
The proof follows directly from the formula for the best linear unbiased predictors.
\begin{align*}
\Pi\left(Z_s \vert \mathbb{Y}_{T} \right) &=
\cov( Z_s, \mathbb{Y}_{T} )
\left(\var(\mathbb{Y}_{T}) \right)^{-1} \mathbb{Y}_{T} \\
&=
\cov\left( \sum_{k\in\mathbb{Z}} \mathcal{B}_k X_{s-k} + e_s, \mathbb{Y}_{T}  \right)
\left(\var(\mathbb{Y}_{T}) \right)^{-1} \mathbb{Y}_{T} \\
&=
\sum_{k\in\mathbb{Z}} \mathcal{B}_k \cov\left(  X_{s-k}, \mathbb{Y}_{T}  \right)
\left(\var(\mathbb{Y}_{T}) \right)^{-1} \mathbb{Y}_{T} \\
&=
\sum_{k\in\mathbb{Z}} \mathcal{B}_k 
\Pi\left( X_{s-k}  \vert \mathbb{Y}_{T} \right)
\end{align*}
\end{proof}

\subsection{Auxiliary Lemmas}
\begin{lemma}
\label{lemma_aux1}
Assuming \ref{assumption:A.1} --- \ref{assumption:A.3}, the time series $\{Z_t\}_{t\in\mathbb{Z}}$ defined by \eqref{eq:regression_model} is stationary. The cross-covariance function between $\{Z_t\}_{t\in\mathbb{Z}}$ and $\{X_t(\cdot)\}_{t\in\mathbb{Z}}$ satisfies
\begin{align}
\label{eq:lemma_aux1_claim1}
\mathscr{R}^{ZX}_h &= \sum_{k\in\mathbb{Z}} \mathcal{B}_k \mathscr{R}^X_{h-k},\\
\label{eq:lemma_aux1_claim2}
R^{ZX}_h (y) &= \sum_{k\in\mathbb{Z}} \int b_k(x)R^X_{h-k}(x,y)\D x,
\end{align}
where the series converge in the vector norm and the supremum norm respectively.
Furthermore
\begin{equation}
\label{eq:lemma_aux1_claim1b}
\sum_{h\in\mathbb{Z}} \left\| \mathscr{R}^{ZX}_h \right\| < \infty,
\end{equation}
\begin{equation}
\label{eq:lemma_aux1_claim2b}
\sum_{h\in\mathbb{Z}} \sup_{y\in[0,1]} \left| R^{ZX}_h(y) \right| =\sum_{h\in\mathbb{Z}} \left\| R^{ZX}_h \right\|_\infty < \infty.
\end{equation}

Assuming further \ref{assumption:newB.3}, the function $R^{ZX}_h (\cdot)$ is twice continuously differentiable in $y$ and
\begin{equation}
\label{eq:lemma_aux1_claim3}
\sup_{y\in[0,1]} \left| \frac{\partial^2}{\partial y^2} R^{ZX}_h(y) \right|
\end{equation}
is uniformly bounded in $h\in\mathbb{Z}$.
\end{lemma}
\begin{proof}
The claims \eqref{eq:lemma_aux1_claim1} and  \eqref{eq:lemma_aux1_claim1b} were proven by \cite[Lemma 2]{hormann2015estimation}.

For the claim \eqref{eq:lemma_aux1_claim2} first consider a functional $\mathcal{A}:\mathcal{H}\to\mathbb{R}$ represented by $a(\cdot)\in\mathcal{H}$, i.e.
$$ \mathcal{A} f = \int a(x)f(x)\D x,\qquad f\in\mathcal{H}, $$
and a Hilbert-Schmidt operator $\mathcal{T}:\mathcal{H}\to\mathcal{H}$ with integral kernel $T(\cdot,\cdot)$. Then its composition $\mathcal{A}\mathcal{T}$ is a functional represented by a function $y\mapsto\int a(x)T(x,y)\D x$ because
\begin{align*}
(\mathcal{AT}f)(y) &= \int a(x) \left( \int T(x,y)f(y)\D y \right) \D x \\
&= \iint a(x) T(x,y) f(y) \D x\D y \\
&= \int \left(a(x)T(x,y)\D x\right)f(y)\D y
\end{align*}
for all $f\in\mathcal{H}$.

The convergence of the sum on the right-hand side of \eqref{eq:lemma_aux1_claim2} and the claim \eqref{eq:lemma_aux1_claim2b} are derived analogously to \cite[Lemma 2]{hormann2015estimation} where the Hilbert-Schmidt norm is replaced by the supremum norm.

The second derivative of $R^{ZX}_h(y)$ is given by
$$ \frac{\partial^2}{\partial y^2} R^{ZX}_h(y) = \sum_{k\in\mathbb{Z}} \int b_k(x)\left( \frac{\partial^2}{\partial y^2} R^X_{h-k}(x,y) \right)\D x .$$
The exchange of the integration-summation and the derivative is justified by the fact that the following bound is integrable in $x$ and absolutely summable in $k$
$$ b_k(x)\left( \frac{\partial^2}{\partial y^2} R^X_{h-k}(x,y) \right)
\leq
C \sup_{y\in[0,1]} \left| b_k(y) \right|, \qquad x\in[0,1], $$
where $C$ is the uniform bound from the assumption \ref{assumption:newB.3} with $(\alpha_1,\alpha_2)=(0,2)$.
\end{proof}

\subsection{Proof of Proposition \ref{prop:asymptotics_of_f^ZX}}
\label{subsec:Proof of Proposition prop:asymptotics_of_fZX}

The minimiser of the optimisation problem \eqref{eq:estimator_cross_density_minimization}, and hence the estimator \eqref{eq:estimator_cross_density}, can be expressed explicitly:
\begin{equation}
\label{eq:proposition1_proof_explicite_formula}
\hat{f}^{ZX}_\omega(x) = \frac{1}{2\pi}\frac{Q_0^\omega S_2 - Q_1^\omega S_1}{S_0 S_2 - S_1^2}
\end{equation}
where
\begin{align*}
S_r &= \frac{1}{L} \sum_{h=-L}^L \frac{T-|h|}{T} S_r^{(h)} W_h,\\
S_r^{(h)} &= \frac{1}{T-|h|}
\sum_{t=\max(1,1-h)}^{\min(T,T-h)} \sum_{j=1}^{N_t}
\left(\frac{x_{tj}-x}{B_C}\right)^r
\frac{1}{B_C} K\left( \frac{x_{tj}-x}{B_C} \right), \\
Q_r^\omega &= \sum_{h=-L}^L \frac{T-|h|}{T} Q_r^{(h)} W_h e^{-\I h\omega},\\
Q_r^{(h)} &= \frac{1}{T-|h|} \sum_{t=\max(1,1-h)}^{\min(T,T-h)} \sum_{j=1}^{N_t}
G^{ZX}_{h,t}(x_{tj})
\left(\frac{x_{tj}-x}{B_C}\right)^r
\frac{1}{B_C} K\left( \frac{x_{tj}-x}{B_C} \right)
\end{align*}
for $r\in\{0,1\}$.
The above quantities are defined as functions of $x\in[0,1]$ and all of the operations are understood pointwise, including the division operation.

The derivation of the formula \eqref{eq:proposition1_proof_explicite_formula} is a simple modification of the standard argument in the local-polynomial regression literature \citep[Section 3.1]{FanJianqing1996Lpma} or \citep[Section B.2]{rubin2020sparsely}.

\begin{lemma}
\label{lemma1}
Under the conditions \ref{assumption:newB.1}, \ref{assumption:newB.2}, and \ref{assumption:newB.8}
\begin{equation}
\label{eq:lemma1_claim1}
\Ez{ \sup_{x\in[0,1]} \left|
S_r^{(h)} - M_{[S_r]}
\right|}
\leq U \frac{1}{\sqrt{T-|h|}}\frac{1}{B_C}
\end{equation}
where the constant $U$ is uniform in $r=0,1,2$, $T\in\mathbb{N}$, $|h|<T$, and $B_C$, and where
\begin{equation}
\label{eq:lemma1_claim1_formulae_for_M}
M_{[S_0]}=\Ez{N}g(x), \qquad
M_{[S_1]}=0, \qquad
M_{[S_2]}= \Ez{N}\sigma^2_K g(x), \qquad
\sigma^2_K = \int v^2 K(v) \D{v}.
\end{equation}
Furthermore,
\begin{equation}
\label{eq:lemma1_claim2}
S_r = M_{[S_r]} + O_p\left( \frac{1}{\sqrt{T}}\frac{1}{B_C} \right)
\end{equation}
uniformly in $x\in[0,1]$.
\end{lemma}
\begin{proof}
We have the usual bias-variance decomposition
\begin{equation}
\label{eq:lemma_1_proof_eq1}
\Ez{ \sup_{x\in[0,1]} \left| S_r^{(h)} - M_{[S_r]} \right| } \leq
\Ez{ \sup_{x\in[0,1]} \left| S_r^{(h)} - \E{S_r^{(h)}} \right| } +
\sup_{x\in[0,1]} \left| \Ez{S_r^{(h)}} - M_{[S_r]} \right|
\end{equation}

\noindent
The bias term, i.e. the second term on the right-hand side of \eqref{eq:lemma_1_proof_eq1}, can be Taylor expanded to second order and the formulae in display \eqref{eq:lemma1_claim1_formulae_for_M} follow. Moreover
\begin{equation}
\label{eq:lemma_1_proof_eq1.5}
\sup_{x\in[0,1]} \left| \Ez{S_r^{(h)}} - M_{[S_r]} \right| = O(B_C^2)
\end{equation}
uniformly in $x\in[0,1]$, $h$, and $T$.

To bound the variance term, i.e. the first term on the right-hand side of \eqref{eq:lemma_1_proof_eq1}, we make use of the Fourier transform. Denote the inverse Fourier transform of the function $u\mapsto K(u)u^r$ by $\zeta_r(t) = \int e^{-\I ut}K(u)u^r du$. Hence
\begin{equation}
\label{eq:lemma_1_proof_eq1.7}
S_r^{(h)} = \frac{1}{T-|h|}
\sum_{t=\max(1,1-h)}^{\min(T,T-h)}
\sum_{j=1}^{N_t}
\frac{1}{2\pi B_C} \int e^{ \I u \frac{x_{tj}-x}{B_C} } \zeta_r(u)du
=
\frac{1}{2\pi} \int \phi_r(v)e^{-\I vx}\zeta_r(B_C v) dv
\end{equation}
where
$$ \phi_r(v) =
\frac{1}{T-|h|}
\sum_{t=\max(1,1-h)}^{\min(T,T-h)}
\sum_{j=1}^{N_t}
e^{\I v x_{tj}}.
$$
Thanks to the above introduced notation we may separate the stochastic terms and the dependence on $x$:
\begin{multline}
\label{eq:lemma_1_proof_eq2}
\Ez{ \sup_{x\in[0,1]} \left|S_r^{(h)} - \E{S_r^{(h)}} \right| } =
\Ez{ \sup_{x\in[0,1]} \left| \frac{1}{2\pi} \int \left( \phi_r(v)-\E\phi_r(v)  \right) e^{-\I vx} \zeta_r(B v)dv \right| }
\leq\\\leq
\frac{1}{2\pi}\Ez{ \int \left|\phi_r(v)-\E{\phi_r(v)} \zeta_r(Bv)dv \right| }
\leq
\left( \sup_{v\in[0,1]} \E\left| \phi_r(v)-\E\phi_r(v) \right| \right)
\frac{1}{2\pi} \int \left|\zeta_r(Bv) \right|dv.
\end{multline}
We firstly focus on $\E\left| \phi_r(v)-\E\phi_r(v) \right|$. By Jensen's inequality,
\begin{equation}
\label{eq:lemma_1_proof_eq2.5}
\E\left| \phi_r(v)-\E\phi_r(v) \right| \leq \sqrt{\var\left|\phi_r(v)\right|},
\end{equation}
hence it suffices to bound the variance. By the independence of $\{N_t\}$ and $\{x_{tj}\}$
\begin{equation}
\label{eq:lemma_1_proof_eq3}
\var\left| \phi_r(v) \right| \leq \E \left| \phi_r(v) \right|^2
\leq \frac{1}{T-|h|} \E{\left| e^{\I v x_{tj}} \right|^2} \leq \frac{1}{T-|h|}.
\end{equation}
Now we treat the integral on the right hand side of \eqref{eq:lemma_1_proof_eq2}
\begin{equation}
\label{eq:lemma_1_proof_eq4}
\int \left| \zeta_r(Bv) \right| dv = \frac{1}{B} \int\left|\zeta_r(\tilde{v})\right| d\tilde{v} = \frac{1}{B} c
\end{equation}
for a constant $c$.
Combining \eqref{eq:lemma_1_proof_eq2}, \eqref{eq:lemma_1_proof_eq3}, and \eqref{eq:lemma_1_proof_eq4} yields
\begin{equation}
\label{eq:lemma_1_proof_eq5}
\Ez{ \sup_{x\in[0,1]} \left|S_r^{(h)} - \E{S_r^{(h)}} \right| }
\leq \sqrt{\frac{1}{T-|h|}} \frac{1}{2\pi B} c.
\end{equation}
Combining the bounds \eqref{eq:lemma_1_proof_eq1.5} and \eqref{eq:lemma_1_proof_eq5} yields the existence of a constant $U$ as in the claim \eqref{eq:lemma1_claim1}.

We now turn our attention to the claim \eqref{eq:lemma1_claim2}. Since $L=o(T)$ we may assume $L<T/2$ and obtain $1/(T-|h|)\leq 2/T$ for $|h|<L$. Noting that $L^{-1} \sum_{h=-L}^L W_h = 1$ we calculate
\begin{multline*}
\left| \frac{1}{L} \sum_{h=-L}^L \left( \frac{T-|h|}{T} W_h S_r^{(h)} \right) - M_{[S_r]} \right| =
\left| \frac{1}{L} \sum_{h=-L}^L \left( \frac{T-|h|}{T} W_h S_r^{(h)} -W_h M_{[S_r]} \right) \right|
\leq\\\leq
\frac{1}{L} \sum_{h=-L}^L \frac{T-|h|}{T} W_h \left| S_r^{(h)} - M_{[S_r]} \right|
+ \frac{1}{L} \sum_{h=-L}^L \frac{|h|}{T} W_h M_{[S_r]}.
\end{multline*}
Taking the supremum norm and the expectation we bound the first term by $\sqrt{2}U/(\sqrt{T} B_C)$. The second term is bounded by $LM_{[S_r]}/T$ which is a faster rate than the one above. Hence the claim \eqref{eq:lemma1_claim2}.
\end{proof}

\begin{lemma}
\label{lemma2}
Under the conditions \ref{assumption:A.1} --- \ref{assumption:A.3}, \ref{assumption:newB.1} --- \ref{assumption:newB.3}, and \ref{assumption:newB.8}
\begin{equation}
\label{eq:lemma2_claim1}
\Ez{ \sup_{x\in[0,1]} \left| Q_r^{(h)} - M_{[Q_r^{(h)}]} \right| } \leq
\tilde{U} \frac{1}{\sqrt{T-|h|}}\frac{1}{B_C}
\end{equation}
where the constant $\tilde{U}$ is uniform in $r=0,1$, $T$, $|h|<T$, and
\begin{equation}
\label{eq:lemma2_claim1_formulae_for_M}
M_{[Q_0^{(h)}]} = \Ez{N} R_h^{ZX}(x) g(x), \qquad M_{[Q_1^{(h)}]} = 0.
\end{equation}
\end{lemma}
\begin{proof}
The standard bias-variance decomposition of the left-hand side of \eqref{eq:lemma2_claim1} yields
\begin{equation}
\label{eq:lemma_2_proof_eq1}
\Ez{ \sup_{x\in[0,1]} \left| Q_r^{(h)} - M_{[Q_r^{(h)}]} \right| } \leq
\Ez{ \sup_{x\in[0,1]} \left| Q_r^{(h)} - \E Q_r^{(h)} \right| } +
\sup_{x\in[0,1]} \left| \E Q_r^{(h)} - M_{[Q_r^{(h)}]} \right|
\end{equation}
The bias therm, i.e. the second term on the right-hand side of \eqref{eq:lemma_2_proof_eq1}, is treated by first conditioning on $x_{tj}$ and then expanding in a Taylor series of order 2. The bias term is of order $O(B_C^2)$ uniformly in $|h|<T$.
The Taylor expansion is justified by \eqref{eq:lemma_aux1_claim3} in Lemma \ref{lemma_aux1}.
 The formulae \eqref{eq:lemma2_claim1_formulae_for_M} follow.

Now, using similar steps as in \eqref{eq:lemma_1_proof_eq1.7}, we obtain:
$$
Q_r^{(h)} = \frac{1}{T-|h|} 
\sum_{t=\max(1,1-h)}^{\min(T,T-h)} \sum_{j=1}^{N_t}
G_{h,t}^{ZX}(x_{tj}) \frac{1}{2\pi B_C}
\int e^{\I u \frac{x_{tj}-x}{B_C}} \zeta_r(u)du =
\frac{1}{2\pi} \int \varphi_r(v) e^{-\I vx} \zeta_r(v)dv
$$
where
$$ \varphi_r(v) = \frac{1}{T-|h|}
\sum_{t=\max(1,1-h)}^{\min(T,T-h)} \sum_{j=1}^{N_t}
G_{h,t}^{ZX}(x_{tj})
e^{\I v x_{tj}}.
$$
Replicating the steps in \eqref{eq:lemma_1_proof_eq2} we arrive at the inequality
\begin{equation}
\label{eq:lemma_2_proof_eq2}
\Ez{\sup_{x\in[0,1]} \left| Q_r^{(h)} - \E{Q_r^{(h)}}  \right| }
\leq
\left( \sup_{v\in[0,1]} \E\left| \varphi_r(v) - \E\varphi_r(v)  \right| \right)
\frac{1}{2\pi} \int \left| \zeta_r(Bv) \right| dv.
\end{equation}
The integral on the right-hand side of \eqref{eq:lemma_2_proof_eq2} is treated in \eqref{eq:lemma_1_proof_eq4}. Using \eqref{eq:lemma_1_proof_eq2.5}, it remains to bound the variance of $\varphi_r(v)$. First, remark that for an arbitrary stationary time-series $\{W_t\}$ with a summable autocovariance function $\{R^W_h\}_{h\in\mathbb{Z}}$, one has the bound:
$$
\var\left( \frac{1}{T} \sum_{t=1}^T W_t \right) = \frac{1}{T} \sum_{h=-T+1}^{T-1} R_h^W \left(1-\frac{|h|}{T}\right) \leq \frac{1}{T} \sum_{h=-\infty}^\infty \left| R_h^W \right| .
$$
Define $W_t = \sum_{j=1}^{N_t} G_{h,t}^{ZX}(x_{tj})e^{\I v x_{tj}}$. This sequence of complex-valued random variables is clearly a stationary (scalar) time series. By conditioning on $N_t$ and $x_{tj}$, and applying the law of total covariance, we can bound the autocovariance function of $\{W_t\}$ by $\left| R_h^{W} \right| \leq \Ez{N} \sup_{x\in[0,1]} \left| R_h^{ZX}(x) \right|$ which is summable in $h$ due to the claim \eqref{eq:lemma_aux1_claim2b} of Lemma \ref{lemma_aux1}. Hence
\begin{equation}
\label{eq:lemma_2_proof_eq3}
\var\left( \varphi_r(v) \right) \leq \frac{K}{T-|h|}
\end{equation}
for a constant $K$ uniform in $h$ and $T$.
Combining \eqref{eq:lemma_2_proof_eq1}, \eqref{eq:lemma_2_proof_eq2}, and \eqref{eq:lemma_2_proof_eq3} yields the claim \eqref{eq:lemma2_claim1}.
\end{proof}

\begin{lemma}
\label{lemma3}
For $r=0,1$:
\begin{enumerate}
\item under the conditions \ref{assumption:A.1} --- \ref{assumption:A.3}, \ref{assumption:newB.1} --- \ref{assumption:newB.3}, \ref{assumption:newB.8}, \ref{assumption:newB.9}
\begin{equation}
\label{eq:lemma3_claim1}
Q_r^\omega = M_{[Q_r^\omega]} + o_p(1),
\end{equation}
\item if \ref{assumption:newB.5} is further assumed, then
\begin{equation}
\label{eq:lemma3_claim2}
Q_r^\omega = M_{[Q_r^\omega]} + O_p\left( L \frac{1}{\sqrt{T}} \frac{1}{B_C} \right),
\end{equation}
\end{enumerate}
where all convergences are uniform in $\omega\in[-\pi,\pi]$ and $x\in[0,1]$, and where
\begin{equation}
\label{eq:lemma3_formulae}
M_{[Q_0^\omega]} = 2\pi \Ez{N} g(x) f_\omega^{ZX}(x), \qquad
M_{[Q_1^\omega]} = 0.
\end{equation}
\end{lemma}
\begin{proof}
We start with the observation that the sum $\sum_{h\in\mathbb{Z}} M_{[Q_r^{(h)}]}$ converges absolutely in the supremum norm and $\sum_{h\in\mathbb{Z}} M_{[Q_r^{(h)}]} e^{-\I h\omega} = M_{[Q_r^\omega]}$.
Moreover,
\begin{multline}
\label{eq:lemma3_proof_eq1}
\left| Q_r^\omega - M_{[Q_r^\omega]} \right| =
\left| \sum_{h=-L}^L \left( 1 - \frac{|h|}{T} \right) W_h e^{-\I h\omega}Q_r^\omega
- \sum_{h\in\mathbb{Z}} M_{[Q_r^\omega]} e^{-\I h\omega} \right|
\leq\\\leq
\sum_{h=-L}^L \left( 1-\frac{|h|}{T} \right)W_h
\left| Q_r^\omega - M_{[Q_r^{(h)}]} \right|
+
\sum_{h=-L}^L
\left(\frac{T-|h|}{T} W_h - 1\right) W_h \left| M_{[Q_r^{(h)}]} \right|
+
\sum_{|h|>L} \left| M_{[Q_r^{(h)}]} \right|
\leq\\\leq
\sum_{h=-L}^L \left( 1-\frac{|h|}{T} \right)\left( 1-\frac{|h|}{L} \right)
\left| Q_r^\omega - M_{[Q_r^{(h)}]} \right|
+
\frac{1}{L^2} \sum_{h=-L}^L |h|^2 \left| M_{[Q_r^{(h)}]} \right|
+
\frac{2}{L} \sum_{h=-L}^L |h| \left| M_{[Q_r^{(h)}]} \right|
+
\sum_{|h|>L} \left| M_{[Q_r^{(h)}]} \right|
\end{multline}
The first term on the right hand side is bounded in the supremum by
\begin{multline*}
\Ez{ \sup_{\omega\in[-\pi,\pi]} \sup_{x\in[0,1]}
\sum_{h=-L}^L \left( 1-\frac{|h|}{T} \right)\left( 1-\frac{|h|}{L} \right)
\left| Q_r^\omega - M_{[Q_r^{(h)}]} \right|
} \leq\\\leq
\sum_{h=-L}^L \left( 1-\frac{|h|}{T} \right)\left( 1-\frac{|h|}{L} \right)
\Ez{ \sup_{x\in[0,1]}
\left| Q_r^\omega - M_{[Q_r^{(h)}]} \right|}
\leq
2L\tilde{U}\frac{1}{\sqrt{T-|h|}}\frac{1}{B_C}
= O\left( L \frac{1}{\sqrt{T-|h|}}\frac{1}{B_C} \right).
\end{multline*}
The last three terms on the right-hand side of \eqref{eq:lemma3_proof_eq1} converge uniformly to zero due to the summability of $\sum_{h\in\mathbb{Z}} M_{[Q_r^{(h)}]}$ in the supremum norm and by Kroncker's lemma. Hence the claim \eqref{eq:lemma3_claim1}.

If \ref{assumption:newB.5} is further assumed then
$\sum_{h\in\mathbb{Z}}|h| \left| M_{[Q_r^{(h)}]} \right| < \infty$
and therefore the last three terms on the right-hand side of \eqref{eq:lemma3_proof_eq1} are of order $O(1/L)$, hence \eqref{eq:lemma3_claim2}.
\end{proof}

\begin{proof}[Proof of Proposition \ref{prop:asymptotics_of_f^ZX}]
\eqref{eq:proposition1_proof_explicite_formula}
The proof follows directly from formula \eqref{eq:proposition1_proof_explicite_formula} and Lemmas \ref{lemma1}, \ref{lemma2}, \ref{lemma3}:
\begin{align*}
Q_0^\omega S_2 - Q_1^\omega S_1 &= 2\pi\left( \E N \right)^2 g(x)^2 f_\omega^{ZX}(x) \sigma_K^2 + o_p(1),\\
S_0 S_2 - S_1^2 &= \left( \E N\right)^2 g(x)^2 f_\omega^{ZX}(x) \sigma_K^2 + O_p\left(\frac{1}{\sqrt{T}}\right)
\end{align*}
uniformly in $\omega\in[-\pi,\pi]$ and $x\in[0,1]$. Hence
\begin{equation}
\hat{f}_\omega^{ZX}(x) = f_\omega^{ZX}(x) + o_p(1)
\end{equation}
uniformly in $\omega\in[-\pi,\pi]$ and $x\in[0,1]$. If the condition \ref{assumption:newB.5} is further assumed, we replace $o_p(1)$ by $O_p(L/(\sqrt{T}B_C))$.
\end{proof}


\subsection{Proof of Theorem \ref{thm:asymptotics_of_B_Tikh}}

Thanks to the fact that
$$
\sup_{k\in\mathbb{Z}} \left\| \hat{\mathcal{B}}_k - \mathcal{B}_k \right\| \leq
\frac{1}{2\pi} \int_{-\pi}^\pi \left\| \hat{\mathscr{B}}_\omega - \mathscr{B}_\omega \right\|\D\omega
$$
the proof reduces to establishing the convergence rate of the frequency response operator $\mathscr{B}_\omega$.

For the Tikhonov regularisation parameter $\rho = \rho(T)>0$ define
$$ \tilde{\mathscr{B}}_\omega = \mathscr{F}^{ZX}_\omega \left( \mathscr{F}^{X}_\omega + \rho \mathcal{I} \right)^{-1}, \qquad\omega\in[-\pi,\pi]. $$
Further, split the desired difference into three terms:
\begin{multline*}
\left\| \hat{\mathscr{B}}_\omega - \mathscr{B}_\omega \right\| \leq
\left\| \hat{\mathscr{B}}_\omega - \tilde{\mathscr{B}}_\omega \right\|
+
\left\| \tilde{\mathscr{B}}_\omega - \mathscr{B}_\omega \right\|
\leq\\\leq
\underbrace{
\left\| \left(\mathscr{F}^{ZX}_\omega-\hat{\mathscr{F}}^{ZX}_\omega \right) \left(\mathscr{F}^{X}_\omega+\rho\mathcal{I} \right)^{-1}  \right\|}_{\mathscr{S}_1}
+
\underbrace{
\left\| \hat{\mathscr{F}}^{ZX}_\omega
\left[ \left(\mathscr{F}^{X}_\omega+\rho\mathcal{I} \right)^{-1} - \left(\hat{\mathscr{F}}^{X}_\omega+\rho\mathcal{I} \right)^{-1} \right]
\right\|}_{\mathscr{S}_2}
+
\underbrace{\left\| \tilde{\mathscr{B}}_\omega - \mathscr{B}_\omega \right\|}_{\mathscr{S}_3}
.
\end{multline*}

We bound each of the terms $\mathscr{S}_1$, $\mathscr{S}_2$, and $\mathscr{S}_3$ and show the convergence of the bound to zero uniformly in $\omega\in[-\pi,\pi]$.
We start with bounding $\mathscr{S}_1$.
\begin{equation}
\label{eq:thm2_proof_mathscr_S_1}
\mathscr{S}_1
\leq
\left\|\mathscr{F}^{ZX}_\omega-\hat{\mathscr{F}}^{ZX}_\omega \right\|
\left\|\left(\mathscr{F}^{X}_\omega+\rho\mathcal{I} \right)^{-1}  \right\|
\leq 
\left\|\mathscr{F}^{ZX}_\omega-\hat{\mathscr{F}}^{ZX}_\omega \right\|
\frac{1}{\rho}
\end{equation}
Now, bounding $\mathscr{S}_2$:
\begin{equation}
\label{eq:thm2_proof_mathscr_S_2}
\mathscr{S}_2 =
\left\| \hat{\mathscr{F}}^{ZX}_\omega
\left[
\left(\mathscr{F}^{X}_\omega+\rho\mathcal{I} \right)^{-1}
\left(\hat{\mathscr{F}}^{X}_\omega - \mathscr{F}^{X}_\omega \right)
\left(\hat{\mathscr{F}}^{X}_\omega+\rho\mathcal{I} \right)^{-1} \right]
\right\|
\leq
\left\| \hat{\mathscr{F}}^{ZX}_\omega \right\|
\frac{1}{\rho}
\left\| \hat{\mathscr{F}}^{X}_\omega - \mathscr{F}^{X}_\omega \right\|
\frac{1}{\rho}
\end{equation}
The right-hand sides of \eqref{eq:thm2_proof_mathscr_S_1} tend to zero uniformly in $\omega\in[-\pi,\pi]$ as $T\to\infty$ thanks to assumption \ref{assumption:D} and Proposition \ref{prop:asymptotics_of_f^ZX}.
The right-hand side of \eqref{eq:thm2_proof_mathscr_S_2} tend to zero uniformly in $\omega\in[-\pi,\pi]$ as $T\to\infty$ thanks to assumption \ref{assumption:D} and \citet[Theorem~2]{rubin2020sparsely}.

It remains to handle the deterministic term $\mathscr{S}_3$.
Since the spectral density operator $\mathscr{F}^X_\omega$ is self-adjoint and trace-class, it admits the series decomposition
$$ \mathscr{F}^X_\omega = \sum_{j=1}^\infty \lambda_j^\omega \varphi_j^\omega\otimes\varphi_j^\omega =
\sum_{j=1}^\infty \lambda_j^\omega
\left\langle \varphi_j^\omega , \cdot\right\rangle \varphi_j^\omega 
$$
where $\{\lambda_j^\omega\}_{j\in\mathbb{N}}$ and $\{\varphi_j^\omega\}_{j\in\mathbb{N}}$ are the harmonic (a.k.a. dynamic) eigenvalues and the harmonic (dynamic) eigenfunctions at frequency $\omega\in[-\pi,\pi]$ \citep{panaretos2013cramer,hormann2015dynamic}.
From the relation \eqref{eq:spectral_normal_equation} we obtain the following series expansions
\begin{align*}
\mathscr{F}^{ZX}_\omega  &=
\sum_{j=1}^\infty \lambda_j^\omega \left\langle \varphi_j^\omega , \cdot\right\rangle \mathscr{B}_\omega\varphi_j^\omega, \\
\mathscr{F}^{ZX}_\omega \left( \mathscr{F}^X_\omega \right)^{-1} &=
\sum_{j=1}^\infty
\left\langle \varphi_j^\omega , \cdot\right\rangle \mathscr{B}_\omega\varphi_j^\omega, \\
\mathscr{F}^{ZX}_\omega \left( \mathscr{F}^X_\omega + \rho\mathcal{I} \right)^{-1} &=
\sum_{j=1}^\infty
\frac{\lambda_j^\omega}{\lambda_j^\omega+\rho}
\left\langle \varphi_j^\omega , \cdot\right\rangle \mathscr{B}_\omega\varphi_j^\omega,
\end{align*}
Combining the above expansions yields
\begin{multline}
\label{eq:thm2_proof_mathscr_S_3}
\mathscr{S}_3^2 = \left\| \tilde{\mathscr{B}}_\omega - \mathscr{B}_\omega \right\|^2
\leq
\sum_{j=1}^\infty \left\|
\left( \frac{\lambda_j^\omega}{\lambda_j^\omega+\rho} -1 \right)
\left\langle \varphi_j^\omega,\cdot \right\rangle
\mathscr{B}_\omega\varphi_j^\omega
\right\|^2
\leq\\\leq
\sum_{j=1}^\infty \left( \frac{\lambda_j^\omega}{\lambda_j^\omega+\rho} -1 \right)^2
\left\| \mathscr{B}_\omega \varphi_j^\omega \right\|^2
=
\sum_{j=1}^\infty
\left(\frac{\rho}{\lambda_j^\omega+\rho}\right)^2
\left\| \mathscr{B}_\omega \varphi_j^\omega \right\|^2.
\end{multline}
Since $\sum_{j=1}^\infty \left\| \mathscr{B}_\omega \varphi_j^\omega \right\|^2 = \sum_{j=1}^\infty \|\mathscr{B}_\omega\|^2 <\infty$ and $\rho / (\lambda_j^\omega + \rho) \to 0$ as $\rho \searrow 0$ for each $j\in\mathbb{N}$,
the right-hand side of \eqref{eq:thm2_proof_mathscr_S_3} tends to zero, completing the proof.

\subsection{Proof of Theorem \ref{thm:asymptotics_of_B_trunc}}

The proof of this theorem is an application of the result by \citet[Theorem~1]{hormann2015estimation}. Their theorem requires that the spectral density operators $\{\mathscr{F}^X_\omega\}_{\omega\in[-\pi,\pi]}$ are estimated with a certain rate,  say $(\psi^X_T)_{T=1}^\infty$, and the cross-density operators $\{\mathscr{F}^{ZX}_\omega\}_{\omega\in[-\pi,\pi]}$ with another rate, say $(\psi_T^{ZX})_{T=1}^\infty$.
For $T\in\mathbb{N}$ we put $\psi^X_T = L T^{-1/2} B_R^{-2}$ and $\psi_T^{ZX} = L T^{-1/2} B_C^{-1}$ by \citet[Theorem~2]{rubin2020sparsely} and Proposition \ref{prop:asymptotics_of_f^ZX} respectively. Note that the convergence rates for the spectral density kernels $\{f_\omega^X\}_{\omega\in[-\pi,\pi]}$ and the cross-spectral density kernels $\{f_\omega^{ZX}\}_{\omega\in[-\pi,\pi]}$ are in the supremum norm which is stronger than the operator norm of the assumptions of \citet[Theorem~1]{hormann2015estimation}.

We can now replicate all the steps of the proof of \citet[Theorem~1]{hormann2015estimation} and see that they only require the above stated rates and the assumptions \ref{assumption:C}, \ref{assumption:E.1}, and \ref{assumption:E.2}.

It is worth noting that the results by \citet[Theorem~1]{hormann2015estimation} are derived under the assumption of $L^p$-$m$-approximability as opposed to cumulant mixing conditions considered in this paper. Nevertheless, they use the assumption of the $L^p$-$m$-approximability only to prove the rate of convergence of the spectral density and cross-spectral density estimators in \citep[Lemma~1]{hormann2015estimation}. Once these rates are established, the proof of \citep[Theorem~1]{hormann2015estimation} does not use the $L^p$-$m$-approximability. Therefore our proof of Theorem~\ref{thm:asymptotics_of_B_trunc} is indeed its simple adaptation which takes as inputs the convergence rates $\psi^X_T = L T^{-1/2} B_R^{-2}$ and $\psi_T^{ZX} = L T^{-1/2} B_C^{-1}$ proved under the cumulant mixing conditions.


\section*{Data availability statement}
The MATLAB code and the results that support the findings of this study are openly available on GitHub at
\href{http://doi.org/10.5281/zenodo.3190740}{DOI 10.5281/zenodo.3190740}, reference \cite{rubin2019github_lagged_regression}.

\bibliographystyle{imsart-nameyear}
\bibliography{biblio}

\begin{thebibliography}{35}

\bibitem[\protect\citeauthoryear{Aue, Norinho and
  H{\"o}rmann}{2015}]{aue2015prediction}
\begin{barticle}[author]
\bauthor{\bsnm{Aue},~\bfnm{Alexander}\binits{A.}},
  \bauthor{\bsnm{Norinho},~\bfnm{Diogo~Dubart}\binits{D.~D.}} \AND
  \bauthor{\bsnm{H{\"o}rmann},~\bfnm{Siegfried}\binits{S.}}
(\byear{2015}).
\btitle{On the prediction of stationary functional time series}.
\bjournal{Journal of the American Statistical Association}
\bvolume{110}
\bpages{378--392}.
\end{barticle}
\endbibitem

\bibitem[\protect\citeauthoryear{Bosq}{2012}]{bosq2012linear}
\begin{bbook}[author]
\bauthor{\bsnm{Bosq},~\bfnm{Denis}\binits{D.}}
(\byear{2012}).
\btitle{Linear {P}rocesses in {F}unction {S}paces: {T}heory and {A}pplications}
\bvolume{149}.
\bpublisher{Springer Science \& Business Media}.
\end{bbook}
\endbibitem

\bibitem[\protect\citeauthoryear{Brillinger}{1981}]{brillinger1981time}
\begin{bbook}[author]
\bauthor{\bsnm{Brillinger},~\bfnm{David~R.}\binits{D.~R.}}
(\byear{1981}).
\btitle{Time {S}eries: {D}ata {A}nalysis and {T}heory}
\bvolume{36}.
\bpublisher{Siam}.
\end{bbook}
\endbibitem

\bibitem[\protect\citeauthoryear{Fan and Gijbels}{1996}]{FanJianqing1996Lpma}
\begin{bbook}[author]
\bauthor{\bsnm{Fan},~\bfnm{Jianqing}\binits{J.}} \AND
  \bauthor{\bsnm{Gijbels},~\bfnm{Irene}\binits{I.}}
(\byear{1996}).
\btitle{Local {P}olynomial {M}odelling and {I}ts {A}pplications: {M}onographs
  on {S}tatistics and {A}pplied {P}robability 66}
\bvolume{66}.
\bpublisher{CRC Press}.
\end{bbook}
\endbibitem

\bibitem[\protect\citeauthoryear{Hall and Horowitz}{2007}]{hall2007methodology}
\begin{barticle}[author]
\bauthor{\bsnm{Hall},~\bfnm{Peter}\binits{P.}} \AND
  \bauthor{\bsnm{Horowitz},~\bfnm{Joel~L.}\binits{J.~L.}}
(\byear{2007}).
\btitle{Methodology and convergence rates for functional linear regression}.
\bjournal{The Annals of Statistics}
\bvolume{35}
\bpages{70--91}.
\end{barticle}
\endbibitem

\bibitem[\protect\citeauthoryear{Hall, M{\"u}ller and
  Wang}{2006}]{hall2006properties}
\begin{barticle}[author]
\bauthor{\bsnm{Hall},~\bfnm{Peter}\binits{P.}},
  \bauthor{\bsnm{M{\"u}ller},~\bfnm{Hans-Georg}\binits{H.-G.}} \AND
  \bauthor{\bsnm{Wang},~\bfnm{Jane-Ling}\binits{J.-L.}}
(\byear{2006}).
\btitle{Properties of principal component methods for functional and
  longitudinal data analysis}.
\bjournal{The annals of statistics}
\bpages{1493--1517}.
\end{barticle}
\endbibitem

\bibitem[\protect\citeauthoryear{Hays et~al.}{2012}]{hays2012functional}
\begin{barticle}[author]
\bauthor{\bsnm{Hays},~\bfnm{Spencer}\binits{S.}},
  \bauthor{\bsnm{Shen},~\bfnm{Haipeng}\binits{H.}},
  \bauthor{\bsnm{Huang},~\bfnm{Jianhua~Z}\binits{J.~Z.}} \betal{et~al.}
(\byear{2012}).
\btitle{Functional dynamic factor models with application to yield curve
  forecasting}.
\bjournal{The Annals of Applied Statistics}
\bvolume{6}
\bpages{870--894}.
\end{barticle}
\endbibitem

\bibitem[\protect\citeauthoryear{H{\"o}rmann, Kidzi{\'n}ski and
  Hallin}{2015}]{hormann2015dynamic}
\begin{barticle}[author]
\bauthor{\bsnm{H{\"o}rmann},~\bfnm{Siegfried}\binits{S.}},
  \bauthor{\bsnm{Kidzi{\'n}ski},~\bfnm{{\L}ukasz}\binits{{\L}.}} \AND
  \bauthor{\bsnm{Hallin},~\bfnm{Marc}\binits{M.}}
(\byear{2015}).
\btitle{Dynamic functional principal components}.
\bjournal{Journal of the Royal Statistical Society: Series B (Statistical
  Methodology)}
\bvolume{77}
\bpages{319--348}.
\end{barticle}
\endbibitem

\bibitem[\protect\citeauthoryear{H{\"o}rmann, Kidzi{\'n}ski and
  Kokoszka}{2015}]{hormann2015estimation}
\begin{barticle}[author]
\bauthor{\bsnm{H{\"o}rmann},~\bfnm{Siegfried}\binits{S.}},
  \bauthor{\bsnm{Kidzi{\'n}ski},~\bfnm{{\L}ukasz}\binits{{\L}.}} \AND
  \bauthor{\bsnm{Kokoszka},~\bfnm{Piotr}\binits{P.}}
(\byear{2015}).
\btitle{Estimation in functional lagged regression}.
\bjournal{Journal of Time Series Analysis}
\bvolume{36}
\bpages{541--561}.
\end{barticle}
\endbibitem

\bibitem[\protect\citeauthoryear{H{\"o}rmann and
  Kokoszka}{2010}]{hormann2010weakly}
\begin{barticle}[author]
\bauthor{\bsnm{H{\"o}rmann},~\bfnm{Siegfried}\binits{S.}} \AND
  \bauthor{\bsnm{Kokoszka},~\bfnm{Piotr}\binits{P.}}
(\byear{2010}).
\btitle{Weakly dependent functional data}.
\bjournal{The Annals of Statistics}
\bvolume{38}
\bpages{1845--1884}.
\end{barticle}
\endbibitem

\bibitem[\protect\citeauthoryear{H{\"o}rmann, Kokoszka and
  Nisol}{2018}]{hormann2018testing}
\begin{barticle}[author]
\bauthor{\bsnm{H{\"o}rmann},~\bfnm{Siegfried}\binits{S.}},
  \bauthor{\bsnm{Kokoszka},~\bfnm{Piotr}\binits{P.}} \AND
  \bauthor{\bsnm{Nisol},~\bfnm{Gilles}\binits{G.}}
(\byear{2018}).
\btitle{Testing for periodicity in functional time series}.
\bvolume{46}
\bpages{2960--2984}.
\end{barticle}
\endbibitem

\bibitem[\protect\citeauthoryear{Israelsson and
  Tammet}{2001}]{israelsson2001variation}
\begin{barticle}[author]
\bauthor{\bsnm{Israelsson},~\bfnm{Sven}\binits{S.}} \AND
  \bauthor{\bsnm{Tammet},~\bfnm{Hannes}\binits{H.}}
(\byear{2001}).
\btitle{Variation of fair weather atmospheric electricity at {M}arsta
  Observatory, {S}weden, 1993--1998}.
\bjournal{Journal of atmospheric and solar-terrestrial physics}
\bvolume{63}
\bpages{1693--1703}.
\end{barticle}
\endbibitem

\bibitem[\protect\citeauthoryear{Klepsch, Kl{\"u}ppelberg and
  Wei}{2017}]{klepsch2017prediction}
\begin{barticle}[author]
\bauthor{\bsnm{Klepsch},~\bfnm{Johannes}\binits{J.}},
  \bauthor{\bsnm{Kl{\"u}ppelberg},~\bfnm{Claudia}\binits{C.}} \AND
  \bauthor{\bsnm{Wei},~\bfnm{Taoran}\binits{T.}}
(\byear{2017}).
\btitle{Prediction of functional {ARMA} processes with an application to
  traffic data}.
\bjournal{Econometrics and Statistics}
\bvolume{1}
\bpages{128--149}.
\end{barticle}
\endbibitem

\bibitem[\protect\citeauthoryear{Kokoszka et~al.}{2017}]{kokoszka2017dynamic}
\begin{barticle}[author]
\bauthor{\bsnm{Kokoszka},~\bfnm{Piotr}\binits{P.}},
  \bauthor{\bsnm{Miao},~\bfnm{Hong}\binits{H.}},
  \bauthor{\bsnm{Reimherr},~\bfnm{Matthew}\binits{M.}} \AND
  \bauthor{\bsnm{Taoufik},~\bfnm{Bahaeddine}\binits{B.}}
(\byear{2017}).
\btitle{Dynamic functional regression with application to the cross-section of
  returns}.
\bjournal{Journal of Financial Econometrics}
\bvolume{16}
\bpages{461--485}.
\end{barticle}
\endbibitem

\bibitem[\protect\citeauthoryear{Kolmogoroff}{1941}]{kolmogoroff1941interpolation}
\begin{barticle}[author]
\bauthor{\bsnm{Kolmogoroff},~\bfnm{A}\binits{A.}}
(\byear{1941}).
\btitle{Interpolation und extrapolation von stationaren zufalligen folgen}.
\bjournal{Izvestiya Rossiiskoi Akademii Nauk. Seriya Matematicheskaya}
\bvolume{5}
\bpages{3--14}.
\end{barticle}
\endbibitem

\bibitem[\protect\citeauthoryear{Kowal}{2018}]{kowal2018dynamic}
\begin{barticle}[author]
\bauthor{\bsnm{Kowal},~\bfnm{Daniel~R.}\binits{D.~R.}}
(\byear{2018}).
\btitle{Dynamic Function-on-Scalars Regression}.
\bjournal{arXiv preprint arXiv:1806.01460}.
\end{barticle}
\endbibitem

\bibitem[\protect\citeauthoryear{Kowal, Matteson and
  Ruppert}{2017a}]{kowal2017functional}
\begin{barticle}[author]
\bauthor{\bsnm{Kowal},~\bfnm{Daniel~R.}\binits{D.~R.}},
  \bauthor{\bsnm{Matteson},~\bfnm{David~S.}\binits{D.~S.}} \AND
  \bauthor{\bsnm{Ruppert},~\bfnm{David}\binits{D.}}
(\byear{2017}a).
\btitle{Functional autoregression for sparsely sampled data}.
\bjournal{Journal of Business \& Economic Statistics}
\bpages{1--13}.
\end{barticle}
\endbibitem

\bibitem[\protect\citeauthoryear{Kowal, Matteson and
  Ruppert}{2017b}]{kowal2017bayesian}
\begin{barticle}[author]
\bauthor{\bsnm{Kowal},~\bfnm{Daniel~R.}\binits{D.~R.}},
  \bauthor{\bsnm{Matteson},~\bfnm{David~S.}\binits{D.~S.}} \AND
  \bauthor{\bsnm{Ruppert},~\bfnm{David}\binits{D.}}
(\byear{2017}b).
\btitle{A Bayesian Multivariate Functional Dynamic Linear Model}.
\bjournal{Journal of the American Statistical Association}
\bvolume{112}
\bpages{733-744}.
\end{barticle}
\endbibitem

\bibitem[\protect\citeauthoryear{Li and Hsing}{2010}]{li2010uniform}
\begin{barticle}[author]
\bauthor{\bsnm{Li},~\bfnm{Yehua}\binits{Y.}} \AND
  \bauthor{\bsnm{Hsing},~\bfnm{Tailen}\binits{T.}}
(\byear{2010}).
\btitle{Uniform convergence rates for nonparametric regression and principal
  component analysis in functional/longitudinal data}.
\bjournal{The Annals of Statistics}
\bvolume{38}
\bpages{3321--3351}.
\end{barticle}
\endbibitem

\bibitem[\protect\citeauthoryear{M{\"u}ller, Sen and
  Stadtm{\"u}ller}{2011}]{muller2011functional}
\begin{barticle}[author]
\bauthor{\bsnm{M{\"u}ller},~\bfnm{Hans-Georg}\binits{H.-G.}},
  \bauthor{\bsnm{Sen},~\bfnm{Rituparna}\binits{R.}} \AND
  \bauthor{\bsnm{Stadtm{\"u}ller},~\bfnm{Ulrich}\binits{U.}}
(\byear{2011}).
\btitle{Functional data analysis for volatility}.
\bjournal{Journal of Econometrics}
\bvolume{165}
\bpages{233--245}.
\end{barticle}
\endbibitem

\bibitem[\protect\citeauthoryear{Panaretos and
  Tavakoli}{2013a}]{panaretos2013fourier}
\begin{barticle}[author]
\bauthor{\bsnm{Panaretos},~\bfnm{Victor~M.}\binits{V.~M.}} \AND
  \bauthor{\bsnm{Tavakoli},~\bfnm{Shahin}\binits{S.}}
(\byear{2013}a).
\btitle{Fourier analysis of stationary time series in function space}.
\bjournal{The Annals of Statistics}
\bvolume{41}
\bpages{568--603}.
\end{barticle}
\endbibitem

\bibitem[\protect\citeauthoryear{Panaretos and
  Tavakoli}{2013b}]{panaretos2013cramer}
\begin{barticle}[author]
\bauthor{\bsnm{Panaretos},~\bfnm{Victor~M.}\binits{V.~M.}} \AND
  \bauthor{\bsnm{Tavakoli},~\bfnm{Shahin}\binits{S.}}
(\byear{2013}b).
\btitle{Cram{\'e}r--{K}arhunen--{L}o{\`e}ve representation and harmonic
  principal component analysis of functional time series}.
\bjournal{Stochastic Processes and their Applications}
\bvolume{123}
\bpages{2779--2807}.
\end{barticle}
\endbibitem

\bibitem[\protect\citeauthoryear{Pham and
  Panaretos}{2018}]{pham2018methodology}
\begin{barticle}[author]
\bauthor{\bsnm{Pham},~\bfnm{Tung}\binits{T.}} \AND
  \bauthor{\bsnm{Panaretos},~\bfnm{Victor}\binits{V.}}
(\byear{2018}).
\btitle{Methodology and Convergence Rates for Functional Time Series
  Regression}.
\bjournal{Statistica Sinica}
\bvolume{28}
\bpages{2521-2539}.
\bnote{(Special Issue in Memory of Peter Hall)}.
\end{barticle}
\endbibitem

\bibitem[\protect\citeauthoryear{Priestley}{1981}]{priestley1981spectral}
\begin{bbook}[author]
\bauthor{\bsnm{Priestley},~\bfnm{Maurice~Bertram}\binits{M.~B.}}
(\byear{1981}).
\btitle{Spectral {A}nalysis and {T}ime {S}eries}.
\bpublisher{Academic press}.
\end{bbook}
\endbibitem

\bibitem[\protect\citeauthoryear{Rub\'{i}n and
  Panaretos}{2019}]{rubin2019github_lagged_regression}
\begin{bmisc}[author]
\bauthor{\bsnm{Rub\'{i}n},~\bfnm{Tom\'{a}\v{s}}\binits{T.}} \AND
  \bauthor{\bsnm{Panaretos},~\bfnm{Victor~M.}\binits{V.~M.}}
(\byear{2019}).
\btitle{GitHub Repository: {F}unctional Lagged Regression with Sparse Noisy
  Observations}.
\bhowpublished{\url{https://doi.org/10.5281/zenodo.3190740}}.
\bdoi{10.5281/zenodo.3190740}
\end{bmisc}
\endbibitem

\bibitem[\protect\citeauthoryear{Rub{\'\i}n and
  Panaretos}{2020}]{rubin2020sparsely}
\begin{barticle}[author]
\bauthor{\bsnm{Rub{\'\i}n},~\bfnm{Tom{\'a}{\v{s}}}\binits{T.}} \AND
  \bauthor{\bsnm{Panaretos},~\bfnm{Victor~M.}\binits{V.~M.}}
(\byear{2020}).
\btitle{Sparsely observed functional time series: Estimation and prediction}.
\bjournal{Electronic Journal of Statistics}
\bvolume{14}
\bpages{1137--1210}.
\end{barticle}
\endbibitem

\bibitem[\protect\citeauthoryear{Sen and
  Kl\"{u}ppelberg}{2019}]{sen2019timeseries}
\begin{barticle}[author]
\bauthor{\bsnm{Sen},~\bfnm{Rituparna}\binits{R.}} \AND
  \bauthor{\bsnm{Kl\"{u}ppelberg},~\bfnm{Claudia}\binits{C.}}
(\byear{2019}).
\btitle{Time series of functional data with application to yield curves}.
\bjournal{Applied Stochastic Models in Business and Industry}
\bpages{1-16}.
\end{barticle}
\endbibitem

\bibitem[\protect\citeauthoryear{Shumway and Stoffer}{2000}]{shumway2000time}
\begin{barticle}[author]
\bauthor{\bsnm{Shumway},~\bfnm{Robert~H.}\binits{R.~H.}} \AND
  \bauthor{\bsnm{Stoffer},~\bfnm{David~S}\binits{D.~S.}}
(\byear{2000}).
\btitle{Time {S}eries {A}nalysis and {I}ts {A}pplications}.
\bjournal{Studies In Informatics And Control}
\bvolume{9}
\bpages{375--376}.
\end{barticle}
\endbibitem

\bibitem[\protect\citeauthoryear{Tammet}{2009}]{tammet2009joint}
\begin{barticle}[author]
\bauthor{\bsnm{Tammet},~\bfnm{Hannes}\binits{H.}}
(\byear{2009}).
\btitle{A joint dataset of fair-weather atmospheric electricity}.
\bjournal{Atmospheric Research}
\bvolume{91}
\bpages{194--200}.
\end{barticle}
\endbibitem

\bibitem[\protect\citeauthoryear{Tavakoli and
  Panaretos}{2016}]{tavakoli2016detecting}
\begin{barticle}[author]
\bauthor{\bsnm{Tavakoli},~\bfnm{Shahin}\binits{S.}} \AND
  \bauthor{\bsnm{Panaretos},~\bfnm{Victor~M.}\binits{V.~M.}}
(\byear{2016}).
\btitle{Detecting and localizing differences in functional time series
  dynamics: a case study in molecular biophysics}.
\bjournal{Journal of the American Statistical Association}
\bvolume{111}
\bpages{1020--1035}.
\end{barticle}
\endbibitem

\bibitem[\protect\citeauthoryear{Wiener}{1950}]{wiener1950extrapolation}
\begin{bbook}[author]
\bauthor{\bsnm{Wiener},~\bfnm{Norbert}\binits{N.}}
(\byear{1950}).
\btitle{Extrapolation, {I}nterpolation, and {S}moothing of {S}tationary {T}ime
  {S}eries: with {E}ngineering {A}pplications}.
\bpublisher{Technology Press}.
\end{bbook}
\endbibitem

\bibitem[\protect\citeauthoryear{Xu et~al.}{2013}]{xu2013periodic}
\begin{barticle}[author]
\bauthor{\bsnm{Xu},~\bfnm{Bin}\binits{B.}},
  \bauthor{\bsnm{Zou},~\bfnm{Dan}\binits{D.}},
  \bauthor{\bsnm{Chen},~\bfnm{Ben~Yuan}\binits{B.~Y.}},
  \bauthor{\bsnm{Zhang},~\bfnm{Jin~Ye}\binits{J.~Y.}} \AND
  \bauthor{\bsnm{Xu},~\bfnm{Guo~Wang}\binits{G.~W.}}
(\byear{2013}).
\btitle{Periodic variations of atmospheric electric field on fair weather
  conditions at YBJ, Tibet}.
\bjournal{Journal of Atmospheric and Solar-Terrestrial Physics}
\bvolume{97}
\bpages{85--90}.
\end{barticle}
\endbibitem

\bibitem[\protect\citeauthoryear{Yao, M{\"u}ller and
  Wang}{2005a}]{yao2005functional}
\begin{barticle}[author]
\bauthor{\bsnm{Yao},~\bfnm{Fang}\binits{F.}},
  \bauthor{\bsnm{M{\"u}ller},~\bfnm{Hans-Georg}\binits{H.-G.}} \AND
  \bauthor{\bsnm{Wang},~\bfnm{Jane-Ling}\binits{J.-L.}}
(\byear{2005}a).
\btitle{Functional data analysis for sparse longitudinal data}.
\bjournal{Journal of the American Statistical Association}
\bvolume{100}
\bpages{577--590}.
\end{barticle}
\endbibitem

\bibitem[\protect\citeauthoryear{Yao, M{\"u}ller and
  Wang}{2005b}]{yao2005functional_linear_regression}
\begin{barticle}[author]
\bauthor{\bsnm{Yao},~\bfnm{Fang}\binits{F.}},
  \bauthor{\bsnm{M{\"u}ller},~\bfnm{Hans-Georg}\binits{H.-G.}} \AND
  \bauthor{\bsnm{Wang},~\bfnm{Jane-Ling}\binits{J.-L.}}
(\byear{2005}b).
\btitle{Functional linear regression analysis for longitudinal data}.
\bjournal{The Annals of Statistics}
\bvolume{33}
\bpages{2873--2903}.
\end{barticle}
\endbibitem

\bibitem[\protect\citeauthoryear{Yao et~al.}{2003}]{yao2003shrinkage}
\begin{barticle}[author]
\bauthor{\bsnm{Yao},~\bfnm{Fang}\binits{F.}},
  \bauthor{\bsnm{M{\"u}ller},~\bfnm{Hans-Georg}\binits{H.-G.}},
  \bauthor{\bsnm{Clifford},~\bfnm{Andrew~J.}\binits{A.~J.}},
  \bauthor{\bsnm{Dueker},~\bfnm{Steven~R.}\binits{S.~R.}},
  \bauthor{\bsnm{Follett},~\bfnm{Jennifer}\binits{J.}},
  \bauthor{\bsnm{Lin},~\bfnm{Yumei}\binits{Y.}},
  \bauthor{\bsnm{Buchholz},~\bfnm{Bruce~A.}\binits{B.~A.}} \AND
  \bauthor{\bsnm{Vogel},~\bfnm{John~S.}\binits{J.~S.}}
(\byear{2003}).
\btitle{Shrinkage estimation for functional principal component scores with
  application to the population kinetics of plasma folate}.
\bjournal{Biometrics}
\bvolume{59}
\bpages{676--685}.
\end{barticle}
\endbibitem

\end{thebibliography}

\end{document}